\documentclass[reqno,10pt]{article}
\pdfoutput=1
\usepackage{amsthm,amsmath,amssymb}

\usepackage[dvips]{graphicx}

\usepackage{color}

\usepackage[margin=1.25in]{geometry}

\usepackage[font=small,margin=10pt,labelfont={bf},labelsep={space}]{caption}
\usepackage{subfig}
\usepackage{enumitem}

\usepackage[usenames,dvipsnames]{xcolor}
\definecolor{dark-gray}{gray}{0.3}
\definecolor{dkgray}{rgb}{.4,.4,.4}
\definecolor{dkblue}{rgb}{0,0,.5}
\definecolor{medblue}{rgb}{0,0,.75}
\definecolor{rust}{rgb}{0.5,0.1,0.1}

\usepackage[colorlinks=true]{hyperref}
\hypersetup{pdftitle={Sharp recovery bounds for convex demixing, with applications}}    %
\hypersetup{pdfauthor={McCoy \& Tropp}}     %

\hypersetup{linkcolor=dkblue}          %
\hypersetup{citecolor=rust}        %
\hypersetup{urlcolor=rust}           %

\usepackage{graphicx}

\usepackage{booktabs} %
\usepackage{multirow} %

\usepackage[scaled]{helvet} %
\usepackage{fourier} %
\usepackage{bm} %

\providecommand{\mathbold}[1]{\bm{#1}} 

\usepackage{units} %

\theoremstyle{plain}
\newtheorem{theorem}{Theorem}[section]
\newtheorem{lemma}[theorem]{Lemma}
\newtheorem{proposition}[theorem]{Proposition}
\newtheorem{fact}[theorem]{Fact}
\newtheorem{corollary}[theorem]{Corollary}

\theoremstyle{definition}
\newtheorem{definition}[theorem]{Definition}
\newtheorem{example}[theorem]{Example}
\newtheorem{remark}[theorem]{Remark}

\newcommand{\vct}[1]{\mathbold{#1}}
\newcommand{\mtx}[1]{\mathbold{#1}}

\newcommand{\defeq}{\ensuremath{\mathrel{\mathop{:}}=}}
\newcommand{\eqdef}{\ensuremath{=\mathrel{\mathop{:}}}}
\newcommand{\minprog}[3]{
  \left.
  \begin{array}{ll} 
    \text{minimize}_{#1} & {#2}  \\[4pt]
    \text{subject to} & {#3} 
  \end{array} 
  \right. 
}

\numberwithin{equation}{section}

\renewcommand{\mid}{\mathrel{\mathop{:}}}

\renewcommand{\phi}{\varphi}

\newcommand{\eps}{\varepsilon}

\newcommand{\Fcone}{\mathcal{F}} %

\newcommand{\econst}{\mathrm{e}}

\newcommand{\zerovct}{\vct{0}}

\newcommand{\Id}{\mathbf{I}}

\newcommand{\R}{\mathbb{R}}

\newcommand{\abs}[1]{\left\vert {#1} \right\vert}

\newcommand{\sgn}[1]{\operatorname{sgn}{#1}}

\newcommand{\argmin}{\operatorname*{arg\; min}}

\newcommand{\Prob}[1]{\mathbb{P}\left\{ {#1} \right\}}

\makeatletter
\def\@cite#1#2{{\normalfont[{\bfseries#1\if@tempswa , \normalfont{#2}\fi}]}}
\makeatother

\newcommand{\adj}{*}

\newcommand{\rank}{\operatorname{rank}}

\newcommand{\nnz}{\operatorname{nnz}}
\renewcommand{\vec}{\operatorname{vec}}
\newcommand{\Proj}{\ensuremath{\mtx{\Pi}}}

\newcommand{\norm}[1]{\left\Vert {#1} \right\Vert}

\newcommand{\lone}[1]{\norm{#1}_{\ell_1}}
\newcommand{\smlone}[1]{\|#1\|_{\ell_1}}
\newcommand{\linf}[1]{\norm{#1}_{\ell_\infty}}
\newcommand{\sone}[1]{\norm{#1}_{S_1}}

\newcommand{\opnorm}[1]{\norm{#1}_{\mathrm{Op}}}

\newcommand{\CDMq}{convex demixing method}
\newcommand{\Index}{\mathcal{D}}

\usepackage{microtype}

\title{Sharp recovery bounds
 for convex demixing, \\with applications}
\author{Michael B. McCoy\footnote{Corresponding author.}\,\,\,\footnote{  The authors are with the Department
    of Computing \& Mathematical Sciences, California Institute
    of Technology, 1200 E California Blvd., Pasadena, CA 91125.
    \emph{Email:}~\texttt{\{mccoy,jtropp\}@cms.caltech.edu}.   Tel.: (626) 395-4059
  Fax: (626) 578-0124.  Research
    supported by ONR awards N00014-08-1-0883 and N00014-11-1002, AFOSR
    award FA9550-09-1-0643, DARPA award N66001-08-1-2065, and a Sloan
    Research Fellowship.}~{ } \and Joel
  A. Tropp\footnotemark[\value{footnote}]}
\date{\small Received: May 20, 2012. \\  Accepted: January 8, 2014}
\newcommand{\figwidth}{0.6}
\begin{document}

\maketitle
\begin{abstract}
  Demixing refers to the challenge of identifying two structured  signals given only the sum of the two signals and prior information  about their structures. Examples include the problem of  separating a signal that is sparse with respect to one basis from a  signal that is sparse with respect to a second basis, and the problem of decomposing an observed matrix into a low-rank matrix plus a sparse matrix.  This paper describes and analyzes a framework, based on convex optimization,  for solving these demixing problems, and many others.    This work introduces a randomized signal model which ensures that  the two structures are incoherent, i.e., generically oriented.  For  an observation from this model, this approach identifies  a summary statistic that reflects the complexity of a particular  signal.  The difficulty of separating two structured, incoherent  signals depends only on the total complexity of the two structures.  Some applications include (i) demixing two signals that are  sparse in mutually incoherent bases; (ii) decoding spread-spectrum  transmissions in the presence of impulsive errors; and (iii) removing  sparse corruptions from a low-rank matrix.  In each case, the  theoretical analysis of the convex demixing method  closely matches its empirical behavior.

\end{abstract}

\vspace{12pt}
\noindent
{\footnotesize {Communicated by} Emmanuel Cand\`es.
\\[-2pt]
\noindent
{Keywords:} \emph{Demixing, sparsity, integral geometry, convex optimization}
\\[-2pt]
\noindent
{AMS subject classifications (MSC2010):} 60D05, 52B55, 52A22 (primary) 94B75 (secondary)
}

\section{Introduction}
\label{sec:introduction}
In modern data-intensive science, it is common to observe a
superposition of multiple information-bearing
signals. \emph{Demixing} refers to the challenge of separating
out the constituent signals from the observation.  A fundamental
computational question is to understand when a tractable algorithm can
successfully complete the demixing.  Problems of this sort arise
in fields as diverse as acoustics~\cite{5946407},
astronomy~\cite{Starck2003},
communications~\cite{Bobin2006}, \cite{Bobin2007},
geophysics~\cite{Taylor1979}, image
processing~\cite{Starck2005}, \cite{Elad2005}, machine
learning~\cite{Chandrasekaran2010a}, and statistics~\cite{Candes2009}.
Some well-known examples of convex methods for demixing include
morphological component analysis~\cite{Starck2003}, robust principal
component analysis~\cite{Chandrasekaran2011}, \cite{Candes2009}, and
inpainting~\cite{Elad2005}.

This work presents a general framework for demixing based on
convex optimization.  We study
the geometry of the optimization problem, and we develop 
conditions that describe precisely when our method succeeds.  Let us
illustrate the major aspects of our approach through a concrete
example. 

\subsection{A first application: Morphological component analysis}
\label{sec:first-appl-deconv}

Starck et al.\ use demixing to model the problem of
distinguishing stars from galaxies in an astronomical
image~\cite{Starck2003}.  This task requires hypotheses on the two
types of objects.  First, we must assume that stars and galaxies
exhibit different kinds of structure: stars appear as localized bright
points, while galaxies are wispy or filamented.  Second, we must
insist that the image is not so full of stars, nor of galaxies, that
they obscure one another.  These two properties are modeled by the
notions of \emph{incoherence} and \emph{sparsity}.  With these hypotheses, we can solve the demixing problem using a method known as morphological component
analysis (MCA)~\cite{Starck2003}, \cite{Starck2005}, \cite{Elad2005}, \cite{ Bobin2006}, \cite{WriMa:10}.
  
\subsubsection{The MCA signal model}
\label{sec:regi-succ-constr}
We model the observation \(\vct z_0 \in \R^d\) as the 
superposition  of two structured signals:
\begin{equation*}
  \vct z_0 = \mtx A \vct x_0 + \mtx B \vct y_0 \in
  \R^d.
\end{equation*}
The matrices \(\mtx A\) and \(\mtx B\) are known, while
the vectors \(\vct x_0\) and \(\vct y_0\) are unknown.  Each column of
\(\mtx A\) contains an elementary structure that might appear in the
first signal; the columns of \(\mtx B\) reflect the structures in the
second signal.  The vector \(\vct x_0\) selects the columns of
\(\mtx A\) that appear in the first signal, e.g., stars in different
locations, while \(\vct y_0\) selects the columns of \(\vct B\) that generate the second signal, e.g., galaxies in different locations. \emph{Incoherence}
demands that the columns of \(\mtx A\) and \(\mtx B\) are weakly
correlated, and \emph{sparsity} requires that \(\vct x_0\) and \(\vct y_0\)
have few nonzero elements.

For simplicity, we assume that \(\mtx A \) and \(\mtx B\) are
orthonormal bases. By changing coordinates, we may take \(\mtx A =
\Id\), the identity matrix.  The observation then has the form
\begin{equation*}\label{eq:mca-signal}
  \vct z_0 = \vct x_0 + \mtx Q\vct y_0
\end{equation*}
for a known orthogonal matrix \(\mtx Q\).  The specialization to
orthonormal bases is
standard~\cite{Donoho2001}, \cite{Elad2005}, \cite{Starck2005}, \cite{HegBar:12}.

Instead of restricting our attention to specific choices of \(\mtx Q\)
that are incoherent with the identity matrix, we consider an idealized
model for incoherence where \(\mtx Q\) is a uniformly random
orthogonal matrix.  This formulation ensures that the structures in
the two signals are oriented generically with respect to each other.
Other authors have also used this approach to study
incoherence~\cite{Donoho2001}, \cite{Elad2005}.

We quantify the sparsity of the two constituent signals by fixing
parameters \(\tau_{\vct x}\) and \(\tau_{\vct y}\) in the interval
\([0,1]\) such that the unknown signals \(\vct x_0\) and \(\vct y_0\)
satisfy
\begin{equation*}
  \mathrm{nnz}(\vct x_0) =\lceil \tau_{\vct x} d\rceil\; \quad
  \text{and} \quad \;
  \mathrm{nnz}(\vct y_0) =\lceil \tau_{\vct y} d\rceil,
\end{equation*}
where \(\mathrm{nnz}(\vct x)\) denotes the \textit{n}umber of \textit{n}on\textit{z}ero elements
of \(\vct x\).\footnote{We prefer the notation \(\mathrm{nnz}(\cdot)\) over   \(\|\cdot\|_{\ell_0}\) because the number of nonzero elements in a vector is not a norm.   }   In other words, \(\tau_{\vct x}\) and \(\tau_{\vct
  y}\) measure the proportion of nonzero entries in \(\vct x_0\) and
\(\vct y_0\).  These sparsity parameters emerge as the major factor
that determines how hard it is to extract  \(\vct x_0\) and \(\vct
y_0\) from the observation \(\vct z_0\). 

\subsubsection{The constrained MCA demixing procedure}
\label{sec:deconv-proc-1}

The goal of morphological component analysis is to identify the pair
\((\vct x_0,\vct y_0)\) of sparse vectors given the observation \(\vct
z_0 \) and the matrix \(\vct Q\).  A natural technique for finding a
sparse vector that satisfies certain conditions is to minimize the
\(\ell_1\) norm subject to these constraints~\cite{Chen1999}, where
the \(\ell_1\) norm is defined as \(\lone{\vct x} \defeq\sum_{i=1}^d
|x_i|\).

Assume that we have access to side information \(\alpha
=\smlone{\vct y_0}\).  Then the intuition above leads us
to frame the following convex optimization problem for demixing:
\begin{equation}\label{eq:l1-const-1}
  \minprog{}{\lone{\vct x}}{\smlone{ \vct y} \le \alpha \;\; \text{and}\;\;
    \vct x + \mtx Q\vct y = \vct z_0, }
\end{equation}
where the decision variables are \(\vct x,\vct y\in \R^d\).  We call
this optimization problem \emph{constrained MCA}, and we say that it
\emph{succeeds} if \((\vct x_0,\vct y_0)\) is the unique optimal point
of~\eqref{eq:l1-const-1}.  Since~\eqref{eq:l1-const-1} can be written
as a linear program, constrained MCA offers a tractable procedure for
attempting to identify the underlying components \((\vct x_0,\vct y_0)\), provided the
observation \(\vct z_0\), the orthogonal matrix \(\mtx Q\), and the side
information \(\alpha = \smlone{\vct y_0}\).

Constrained MCA is closely related to the standard MCA procedure,
which is a Lagrangian formulation of~\eqref{eq:l1-const-1} that does
not require the side information
\(\alpha\)~\cite[{Eq.~(4)}]{Starck2005}.  The constrained
problem~\eqref{eq:l1-const-1} is more powerful than the standard MCA
procedure, so it provides hard limits on the effectiveness of the
usual approach.  In most cases, the two methods are equivalent,
provided that we can choose the Lagrange multiplier correctly---a
nontrivial task in itself.  See Section~\ref{sec:did-we-assume} for
more details.

\begin{figure}[t!]
  \centering
  \includegraphics[width=\figwidth\columnwidth]{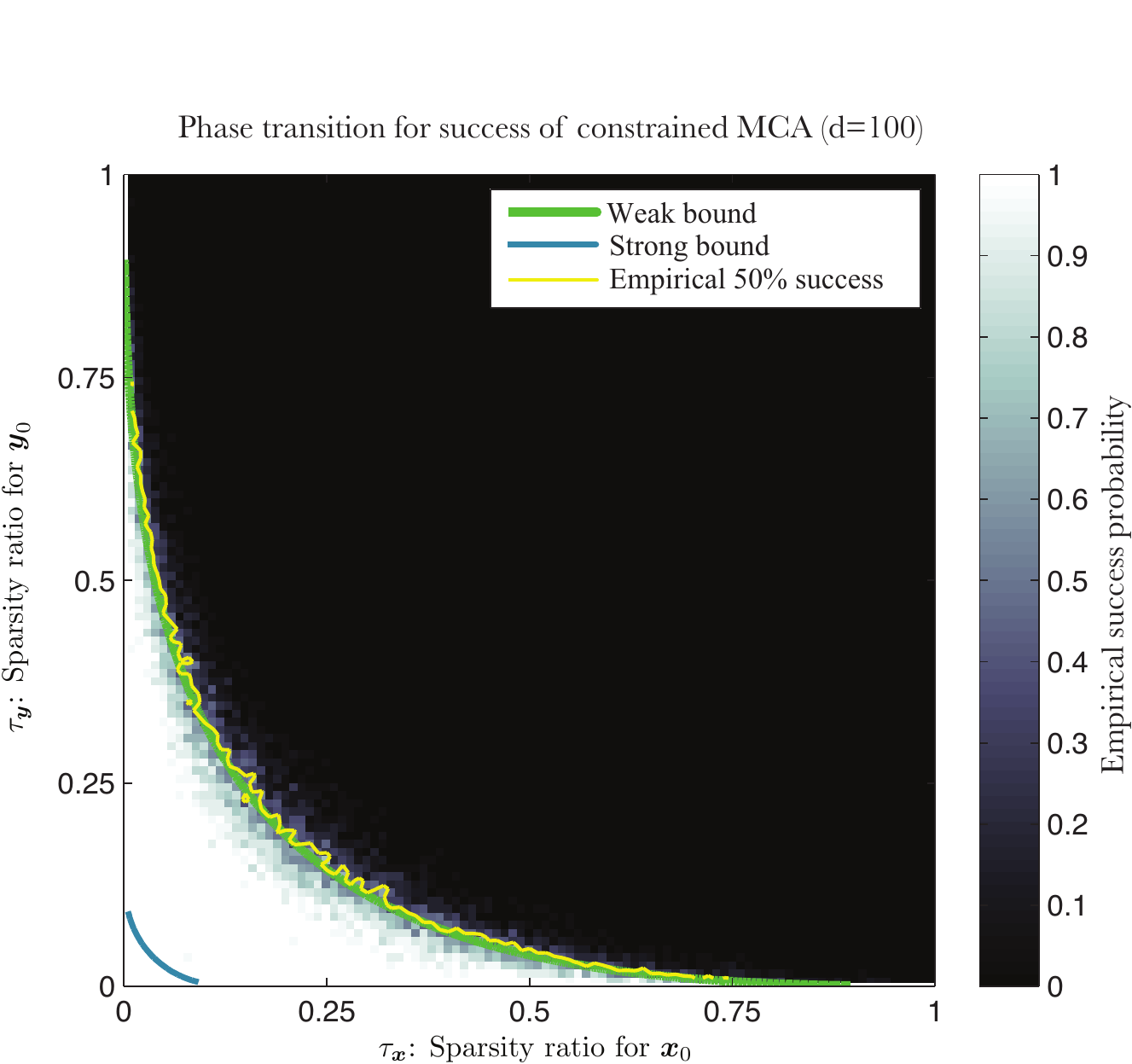}
  \caption{\textsl{Performance of constrained MCA.}  The variables
    \(\tau_{\vct x}\) and \(\tau_{\vct y}\) on the axes represent the
    fraction of components in \(\vct x_0\) and \(\vct y_0\) that are
    nonzero.  The background shading indicates the empirical
    probability that the constrained MCA problem~\eqref{eq:l1-const-1}
    identifies the pair \((\vct x_0,\vct y_0)\) from the observation
    \(\vct z_0 = \vct x_0 + \mtx Q \vct y_0\), where \(\mtx Q\) is a
    random orthogonal basis. The yellow curve marks the empirical \(50\%\)
    success threshold.  The green curve locates the theoretical phase
    transition for demixing a single pair \((\vct x_0,\vct y_0)\).
    For sparsity levels below the blue curve, constrained
    MCA~\eqref{eq:l1-const-1} provably recovers all \((\tau_{\vct
      x},\tau_{\vct y})\)-sparse pairs with high probability in the
    dimension.  Further details are available in
    Section~\ref{sec:regi-succ-fail}
    and Section~\ref{sec:deconv-sparse-vect}.  }
  \label{fig:l1-l1-demixing}
\end{figure}

\subsubsection{Numerical and theoretical results for constrained MCA}
\label{sec:regi-succ-fail}
Figure~\ref{fig:l1-l1-demixing} displays the result of a
numerical experiment on constrained MCA.  We fix the dimension
\(d=100\).  For sparsity levels \((\tau_{\vct x},\tau_{\vct y})\)
varying over the unit square \([0,1]^2\), we form vectors \(\vct x_0\)
and \(\vct y_0\) with sparsity levels \(\nnz(\vct x_0) = \lceil \tau_{\vct x}
d\rceil\) and \(\nnz(\vct y_0 ) = \lceil \tau_{\vct y} d\rceil\).
(The manner in which we choose the nonzero entries is irrelevant.)  We
draw a random orthogonal matrix \(\mtx Q\) and construct the observation \(\vct
z_0 = \vct x_0 + \mtx Q \vct y_0\).  Then we solve the constrained MCA
problem~\eqref{eq:l1-const-1} to identify the pair \((\vct x_0,\vct
y_0)\).  The background of the figure shows the empirical probability
of success over the randomness in \(\mtx Q\); dark areas denote low
probability of success, while light areas denote high success rates.
The yellow curve marks the \(50\%\) success threshold.

This work establishes two theoretical results for constrained MCA.
The first result provides a phase transition curve, parameterized by
the sparsity \((\tau_{\vct x},\tau_{\vct y})\), for the probability
that constrained MCA will demix a single pair \((\vct x_0,\vct
y_0)\) from the associated observation \(\vct z_0\).  This \emph{weak
  bound} is marked by the green line in
Figure~\ref{fig:l1-l1-demixing}.  Observe that the green line
coincides almost perfectly with the empirical phase transition.

Second, we establish a \emph{strong bound}. For a fixed instantiation of
the random orthogonal basis \(\mtx Q\), with high probability,
constrained MCA~\eqref{eq:l1-const-1} can identify \emph{every}
sufficiently sparse pair \((\vct x_0,\vct y_0)\) from the associated
observation \(\vct z_0\).  The blue curve in the bottom left corner of
Figure~\ref{fig:l1-l1-demixing} is a lower estimate for the
sparsity pairs \((\tau_{\vct x},\tau_{\vct y})\)  where this uniform
guarantee holds. Section~\ref{sec:deconv-sparse-vect} provides the
details regarding the computation of the weak and strong bounds as
well as a fully detailed description of our numerical experiment.

\subsection{A recipe for demixing}
\label{sec:recipe-demixing}

This work is not primarily about MCA.  We are interested in developing
methods that apply to a whole spectrum of demixing problems.  The
following two sections describe how to construct a convex program that
can separate two structured signals. 

\subsubsection{Structured signals and atomic gauges}
\label{sec:struct-atom-guag}
\begin{figure}[t!]
  \centering
  { 
\includegraphics[width=0.4\columnwidth]{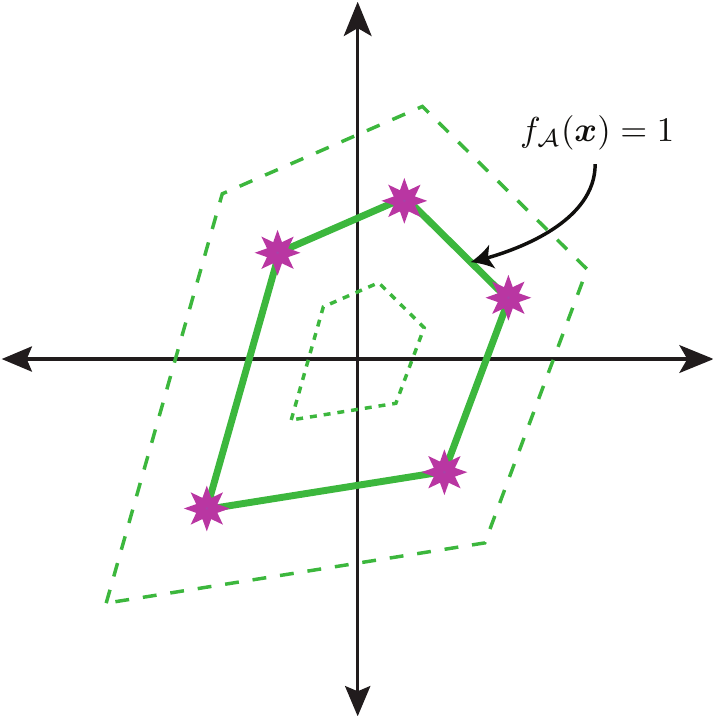}\label{subf:left-panel-atom-gauge}}
\hspace{24pt}
%
{\includegraphics[width=0.4\columnwidth]{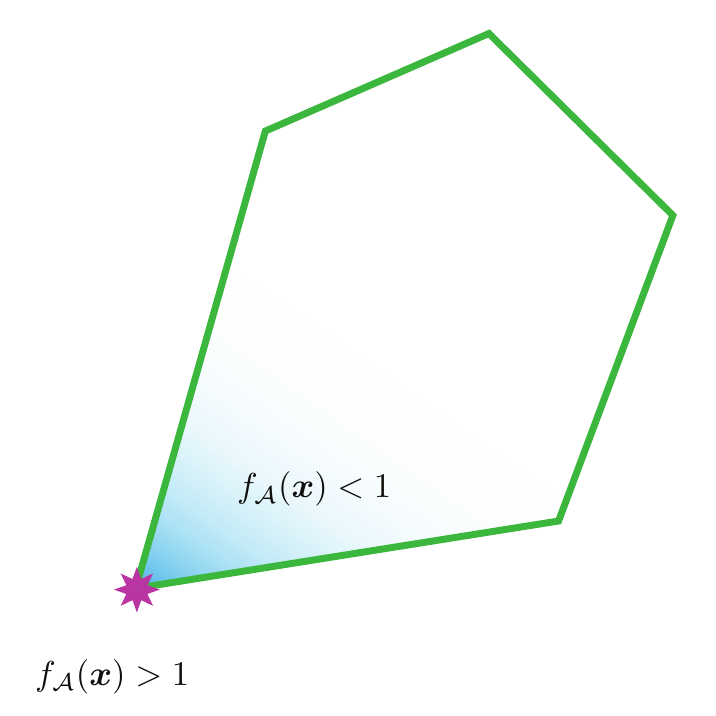}}
\caption{\textsl{Atomic gauge.} [Left] Let \(\mathcal{A}\) be an atomic
  set consisting of five atoms (stars).  The ``unit ball'' of the
  atomic gauge \(f_{\mathcal{A}}\) is the closed convex hull of
  \(\mathcal{A}\) (heavy line).  Other level sets (dashed lines) of
  \(f_{\mathcal{A}}\) are dilations of the unit ball. [Right] At an atom (star),
  the unit ball of \(f_\mathcal{A}\) tends to have sharp corners.
  Most perturbations away from this atom increase the value of
  \(f_{\mathcal{A}}\), so the atomic gauge is an effective measure of
  the complexity of an atomic signal.}
  \label{fig:atomic-gauge}
\end{figure}
The \(\ell_1\) norm is a convex complexity measure that tends to be
small near sparse vectors, so we can minimize the \(\ell_1\) norm 
to promote sparsity.  We now describe a method for building complexity
measures that are appropriate for other types of structure.  This
construction was originally introduced in the nonlinear approximation
literature~\cite{MR1399379}, \cite{MR1951502}.  The recent
paper~\cite{ChaRecPar:12} explains how to apply these ideas to
solve signal processing problems.

In practice, we often encounter signals that are formed as a positive
linear combination of a few elementary structures, called
\emph{atoms}, drawn from a fixed collection.  For example, a sparse
vector in \(\R^d\) is a conic combination of a small number elements
from the set \(\{\pm \vct e_i \mid i=1,\dotsc,d\}\) of signed standard basis
vectors.  We want to construct a function that reflects the complexity
of an atomic signal.

We define the \emph{atomic gauge} of a vector \(\vct x\in \R^d\) with respect
to a set \(\mathcal{A}\subset \R^d\) of atoms by
\begin{equation*}
  f_{\mathcal{A}}(\vct x) \defeq \inf\;\bigl
  \{\lambda\ge 0 \mid  
  \vct x \in \lambda\cdot \mathrm{conv}(\mathcal{A})\bigr\},
\end{equation*}
with the convention that \(f_{\mathcal{A}}(\vct x) = +\infty\) if the
set is empty.  Then \(f_{\mathcal{A}}\) is a homogeneous convex
function.  The ``unit ball'' of \(f_\mathcal{A}\) is
\(\overline{\mathrm{conv}}(\mathcal{A})\), and the level sets of
\(f_{\mathcal{A}}\) are dilations of this unit ball.  The atomic gauge
\(f_{\mathcal{A}}\) is a norm if and only if
\(\overline{\mathrm{conv}}(\mathcal{A})\) is a bounded, symmetric set
that contains zero in its interior.

The convex hull of an atomic set \({\mathcal{A}}\) tends to have sharp
corners at atoms; see Figure~\ref{fig:atomic-gauge}. At these sharp
points, most perturbations of the objective increase the value of the
gauge, so the atomic gauge tends to take small values at
atoms. Similar behavior occurs at a signal comprised of a relatively
small number of atoms. This
observation is a key reason that atomic gauges make good complexity
measures for atomic signals~\cite{ChaRecPar:12}.  Some common atomic gauges include
\begin{itemize}\renewcommand{\itemsep}{0pt}
\item \emph{The \(\ell_1\) norm.}  The \(\ell_1\) norm on \(\R^d\) is
  the atomic gauge generated by the set \(\mathcal{A}=\{\pm \vct
  e_i\mid i=1,\dotsc,d\}\) of signed standard basis vectors.  This
  norm is widely used to promote
  sparsity~\cite{Chen1999}, \cite{MR2236170}, \cite{MR2241189}.  The \(\ell_1\) norm
  may also be defined for matrices via the
  formula \(\lone{\mtx X}=\sum_{i,j} | X_{ij}|\).  In this context,
  the \(\ell_1\) norm reflects the sparsity of
  matrices~\cite{Sanghavi2009}, \cite{Chandrasekaran2010a}, \cite{Chandrasekaran2011}.  
\item \emph{The \(\ell_\infty\) norm.}  The \(\ell_\infty\) norm on
  \(\R^d\), given by \(\linf{\vct x}=\max_{i=1,\dotsc,d}|x_i|\), is
  the atomic gauge generated by the set \(\mathcal{A} = \{\pm
  1\}^d\subset \R^d\) of all \(2^d\) sign vectors. We use this norm to
  demix binary codewords.  See
  also~\cite{Donoho2010a}, \cite{ChaRecPar:12}, \cite{Mangasarian2011}.  For
  matrices, the \(\ell_\infty\) norm returns
  \(\linf{\vct X} = \max_{i,j} |X_{ij}|\); this function is the atomic
  gauge generated by the set of sign matrices.
\item \emph{The Schatten 1-norm.}  The Schatten 1-norm on
  \(\R^{m\times n}\) is the sum of the singular values of a matrix.
  It is the atomic gauge generated by the set of rank-one matrices in
  \(\R^{m\times n}\) with unit Frobenius norm.  Minimizing the
  Schatten 1-norm promotes low rank~\cite{Fazel2002}, \cite{MR2680543}.
\item \emph{The operator norm.}  The operator norm returns the maximum
  singular value of a matrix.  On the space
  \(\R^{n\times n}\) of square matrices, the operator norm is the
  atomic gauge generated by the set \(\mathsf{O}_n\) of orthogonal
  matrices.  This norm can be used to search for orthogonal
  matrices~\cite[Prop.~3.13]{ChaRecPar:12}
\end{itemize}
Our applications focus on these four instances, but a dizzying variety
other structure-promoting atomic gauges are available. For example, there are atomic gauges for vectors that are sparse in a dictionary (also known as analysis-sparsity)~\cite{ElaMilRub:07}, \cite{CanEldNee:11}, block- and group-sparse vectors~\cite{CotRaoEng:05}, \cite{Sto:09a},  and low-rank tensors and probability measures~\cite[Sec.~2]{ChaRecPar:12}.

\subsubsection{A generic model for incoherence}
\label{sec:random-basis-model}

Demixing is hopeless when the structures in the constituent
signals are too strongly aligned.  As an extreme example, suppose we
observe \(\vct z_0 = \vct x_0+ \vct y_0\), where both \(\vct x_0\) and
\(\vct y_0\) are sparse.  There is clearly no principled way to assign
the nonzero elements of \(\vct z_0\) correctly to \(\vct x_0\)
and \(\vct y_0\).  In contrast, if we observe \(\vct z_0 = \vct x_0 +
\mathbf{H} \vct y_0\), where \(\mathbf{H}\) is a normalized
Walsh--Hadamard transform and both \(\vct x_0\) and \(\vct y_0\) are
sparse, then the pair \((\vct x_0,\vct y_0)\) is typically
identifiable~\cite{Tropp2008c}.  The latter situation is more
favorable than the former because the Walsh--Hadamard matrix and the
identity matrix are incoherent; that is, their columns are weakly
correlated.

In order to avoid restricting our attention to special cases, such as
the Walsh--Hadamard matrix, we model incoherence by
assuming that the basis \(\mtx Q\) is drawn
randomly from the invariant Haar measure on the set of all orthogonal matrices \(\mathsf{O}_d\).  We call
this the \emph{random basis model.} This idealized approach to
incoherence guarantees that the structures in the two constituent
signals are generically oriented.    This model has precedents in the literature on sparse approximation~\cite{Donoho2001}, \cite{Elad2005}, and it is analogous to the assumption of a measurement operator with a uniformly random nullspace that appears in the context of compressed sensing~\cite{Donoho2006}, \cite{Sto:09}. 

We expect that the random basis model also sheds light on other highly incoherent problems, such as the case where \(\mtx Q = \mathbf{H}\) is the Walsh--Hadamard transform or \(\mtx Q = \mathbf{D}\) is the discrete cosine transform (DCT).   Some limited numerical simulations suggest that both the Walsh--Hadamard and the DCT matrices behave qualitatively similar to the random matrix \(\mtx Q\) in the examples considered in this work. This observation is in line with the universality of phase transitions that appear in \(\ell_1\) minimization for many classes of measurement matrices~\cite{Donoho2009a}.  However, more coherent situations may exhibit different behavior, and thus they fall outside the purview of this work.

\subsubsection{Formulating a \CDMq}
\label{sec:deconv-proc}
We are ready to introduce a computational framework for demixing structured
signals.  This approach unifies several related procedures that appear
in the literature.  See, for
example,~\cite{NIPS2010_1164}, \cite{Chandrasekaran2011}, \cite{WriGanMin:13}.

Assume we observe the superposition of two structured signals:
\begin{equation*}
  \vct z_0 = \vct x_0 + \mtx Q\vct y_0,
\end{equation*}
where \(\mtx Q\) is a known orthogonal matrix and the pair \((\vct
x_0,\vct y_0)\) is unknown.  We include the matrix \(\mtx Q\) in the
formalism because it allows us to model incoherence using the random basis model described in  Section~\ref{sec:random-basis-model} above.  Our goal is to demix the
pair \((\vct x_0,\vct y_0)\) from the observation \(\vct z_0\). 

Let \(f\) and \(g\) be convex complexity measures---such as atomic
gauges---associated with the structures we expect to find in \(\vct
x_0\) and \(\vct y_0\).  Suppose we have access to the additional side
information \(\alpha = g(\vct y_0)\).  We combine these
ingredients to reach the following \CDMq:
\begin{equation}
  \label{eq:genconv}
  \minprog{}{f(\vct x)}{ g(\vct y) \le \alpha \;\;\text{and}\;\;  \vct x + \mtx
  Q \vct y = \vct z_0,}
\end{equation}
where the decision variables are \(\vct x,\vct y \in \R^d\).  The
display~\eqref{eq:genconv} describes a convex program because \(f\)
and \(g\) are convex functions.  We say that the
\CDMq~\eqref{eq:genconv} \emph{succeeds at demixing \((\vct
  x_0,\vct y_0)\)}, or simply \emph{succeeds}, if \((\vct x_0,\vct
y_0)\) is the unique optimal point of~\eqref{eq:genconv}; otherwise,
it \emph{fails}.  In this work, we develop conditions that describe when the
\CDMq~\eqref{eq:genconv} succeeds and when it fails.

\subsubsection{The Lagrangian counterpart}
\label{sec:did-we-assume}
In practice, the value \(\alpha = g(\vct y_0)\) may not be known.  In this case, we may replace the \CDMq~\eqref{eq:genconv} with its
Lagrangian relative
\begin{equation}
  \label{eq:gen-lagrange}
  \minprog{}{f(\vct x) + \lambda\cdot g(\vct y)}{\vct x + \mtx Q\vct
    y = \vct z_0,}
\end{equation}
where \(\lambda >0\) is a regularization parameter that must be
specified.  The constrained problem~\eqref{eq:genconv} is slightly
more powerful than~\eqref{eq:gen-lagrange}, so its performance
dominates the Lagrangian formulation~\eqref{eq:gen-lagrange}.  It is
well known that~\eqref{eq:genconv} and~\eqref{eq:gen-lagrange} are
essentially equivalent when the regularization parameter \(\lambda\) is chosen
correctly and a mild regularity condition holds; see
Appendix~\ref{sec:lagrange-parameter} for details.  Thus, we can interpret our results as delineating the \emph{best possible performance} of the Lagrange problem~\eqref{eq:gen-lagrange}.

This type of best-case analysis has precedents in the literature on sparse approximation, e.g.,~\cite{Wai:09}, yet the identification of optimal Lagrange parameters for demixing  remains a significant open problem. Several works prove that specific demixing procedures succeed for specific choices of Lagrange parameters under incoherence assumptions~\cite{Donoho2001}, \cite{WriMa:10}, \cite{Chandrasekaran2011}, \cite{Candes2009}, \cite{XuCarSan:12}, but these  conservative guarantees fail to identify phase transitions.   Since this document was submitted, some additional theoretical guidance choosing Lagrange parameters has appeared~\cite{TanBhaRec:13}, \cite{Sto:13}, \cite{FoyMac:13}, \cite{OymHas:13}.  Nevertheless,  a comprehensive theory describing the optimal choices of Lagrange parameters for~\eqref{eq:gen-lagrange} does not currently exist.

\subsection{One hammer, many nails}
\label{sec:other-applications}

The \CDMq~\eqref{eq:genconv} includes many interesting special
cases. Our analysis provides detailed information about
when~\eqref{eq:genconv} is able to separate two structured, incoherent
signals.   We now describe some applications of this machinery.

\subsubsection{A secure communications protocol that is robust to
  sparse errors}
\label{sec:chann-coding-scheme}

\begin{figure}[t!]
  \centering
  \includegraphics[width=0.55\columnwidth]{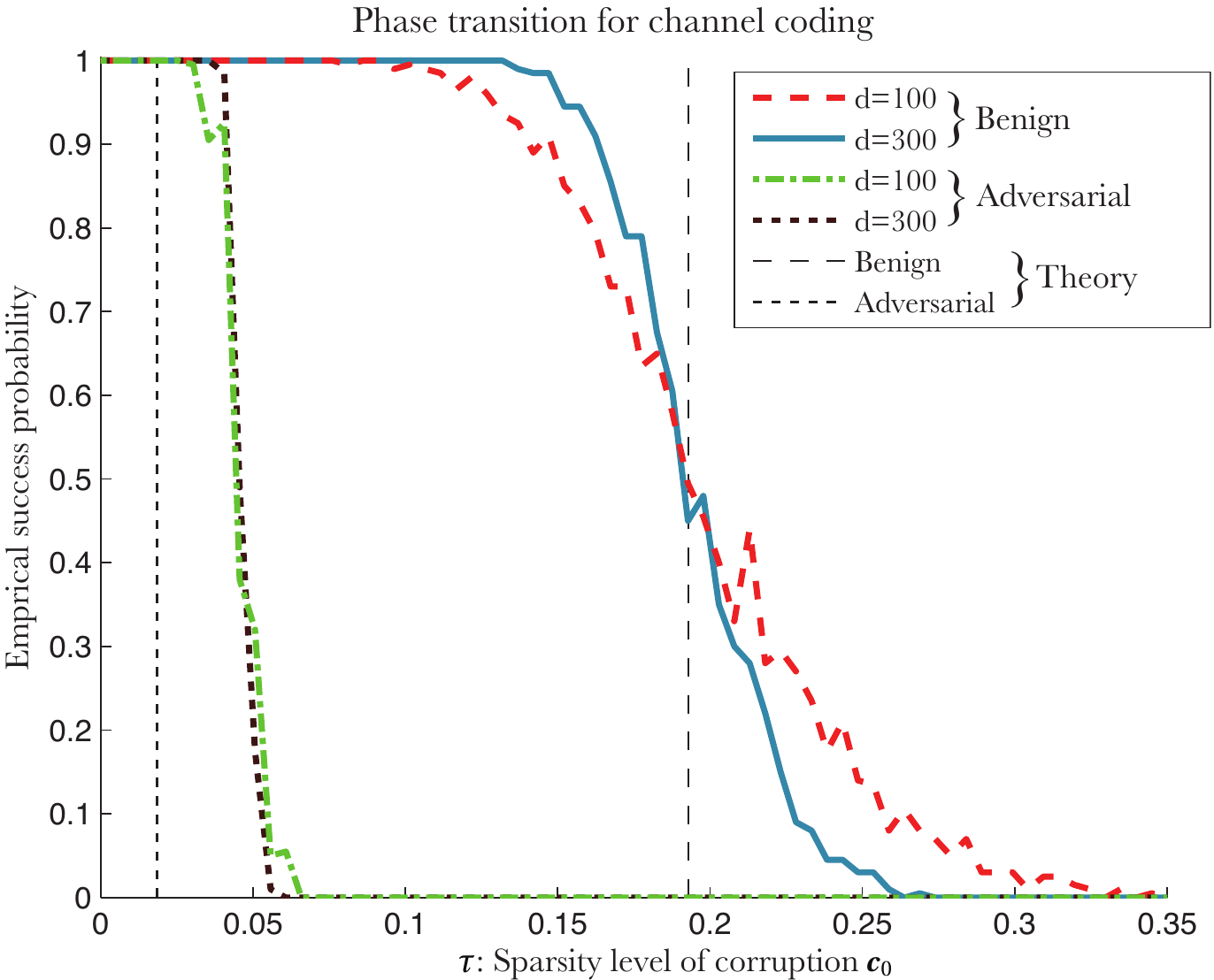}
  \caption{\textsl{Performance of the robust communication protocol.}
    The heavy curves indicate the empirical probability that
    problem~\eqref{eq:channel-code-1} decodes a \(d\)-bit message
    \(\vct m_0\in \{\pm 1\}^d\) from the observation \(\vct z_0 = \mtx
    Q\vct m_0 + \vct c_0\), where \(\mtx Q\) is a random orthogonal matrix.  The
    variable \(\tau\) on the horizontal axis measures the proportion
    of nonzero entries in the corruption \(\vct c_0\).  A benign
    corruption is independent of \(\mtx Q\), while an adversarial
    erasure zeros out the \([\tau d]\) largest components of the
    transmitted message.  The vertical line at \(\tau\approx 0.19\)
    marks the theoretical phase transition for successful decoding
    under a benign corruption.  The vertical line at \(\tau \approx
    0.018\) provides a uniform guarantee applicable to adversarial
    corruptions.  For \(\tau\lesssim 0.018\), with high probability,
    our protocol will decode a message subject to \emph{any}
    \(\tau\)-sparse corruption whatsoever.  See
    Section~\ref{sec:chann-coding-scheme}
    and Section~\ref{sec:channel-coding-1} for more
    details.}\label{fig:channel-coding}
\end{figure}

Suppose we wish to securely transmit a binary message across a
communications channel.  We can obtain strong guarantees of security
by modulating the message with a random rotation before
transmission~\cite{Wyner1979}, \cite{Wyner1979b}.  Our theory shows that
decoding the message via demixing also makes this secure scheme
perfectly robust to sparse corruptions such as erasures or malicious
interference.

Consider the following simple communications protocol. We model the
binary message as a sign vector \(\vct m_0 \in \{\pm 1\}^d\).  Choose a
random orthogonal matrix \(\mtx Q\in \mathsf{O}_d\).  The transmitter sends the
scrambled message {\(\vct s_0 = \mtx Q \vct m_0\)} across the channel,
where it is corrupted by an unknown sparse vector \(\vct c_0\in \R^d\).  The
receiver must determine the original message given only the corrupted
signal
\begin{equation*}
  \vct z_0 = \vct s_0 + \vct c_0 = \mtx Q \vct m_0 + \vct c_0
\end{equation*}
and knowledge of the scrambling matrix \(\mtx Q\).  

This signal model is perfectly suited to the demixing recipe of
Section~\ref{sec:recipe-demixing}.  The discussion in
Section~\ref{sec:struct-atom-guag} indicates that the \(\ell_1\) and
\(\ell_\infty\) norms are natural complexity measures for the structured
signals \(\vct c_0\) and \(\vct m_0\).  Since the message \(\vct m_0\)
is a sign vector, we also have the side information \(\linf{\vct
  m_0} = 1\).  Our receiver then recovers the message with the
\CDMq\ 
\begin{equation}\label{eq:channel-code-1}
  \minprog{}{\lone{{\vct c}}}{ \linf{\vct{m}} \le 1 \; \text{ and } \;
    \vct c + \mtx Q \vct m= \vct z_0,}
\end{equation}
where the decision variables are \(\vct c,\vct m \in \R^d\).  This
method succeeds if \((\vct c_0,\vct m_0)\) is the unique optimal
point of~\eqref{eq:channel-code-1}.

In Section~\ref{sec:channel-coding-1}, we apply the general theory
developed in this work to study this communications protocol.  Before
summarizing the results of this analysis, we fix some notation.
Suppose the corruption \(\vct c_0\in \R^d\) is \(\tau\)-sparse; that
is, \(\mathrm{nnz}(\vct c_0) = \lceil \tau d\rceil\) for some \(\tau
\in [0,1]\). We further distinguish between two types of corruption.
A \emph{benign} corruption \(\vct c_0\) is independent of the
scrambling matrix \(\mtx Q\).  In contrast, an \emph{adversarial}
corruption may depend on both \(\mtx Q\) and \(\mtx m_0\).
Adversarial corruptions also include nonlinear effects that are not
necessarily malicious. For example, we can model an erasure at the
\(i\)th time instant by taking \((\vct c_0)_i = - (\mtx Q \vct
m_0)_i\).

Figure~\ref{fig:channel-coding} presents the results of a
numerical experiment on this communications protocol; the complete
experimental procedure is detailed in
Section~\ref{sec:numerical-experiment-channel-coding}.  Briefly, we
consider messages of length \(d=100\) and \(d=300\), and we let the
sparsity \(\tau\) range over the interval \([0,0.35]\).  We test the benign case by
adding a \(\tau\)-sparse corruption that is independent from \(\mtx
Q\).  We also consider a particular adversarial corruption in which we
set the \([\tau d]\) {largest-magnitude} entries in the
transmitted message \(\vct s_0\) to zero.   The curves indicate the
empirical probability that the protocol succeeds as a function of \(\tau\).

In the benign case, our theory shows that there exists a phase
transition in the success probability of the
\CDMq~\eqref{eq:channel-code-1} at sparsity level \(\tau \approx
0.19\).  The empirical \(50\%\) failure threshold for benign
corruptions closely matches this prediction.  In the adversarial case,
our results guarantee that with high probability, our protocol will
tolerate all corruptions that affect no more than \(1.8\%\) of the
components in the received message \(\vct z_0\).  This bound is
conservative for the type of adversarial corruption in the numerical
experiment; this is not surprising because we may not have constructed
the worst possible corruption.

\subsubsection{Low-rank matrix recovery with generic sparse
  corruptions}
\label{sec:low-rank-matrix}

\begin{figure}[t]
  \centering
  \includegraphics[width=\figwidth\columnwidth]{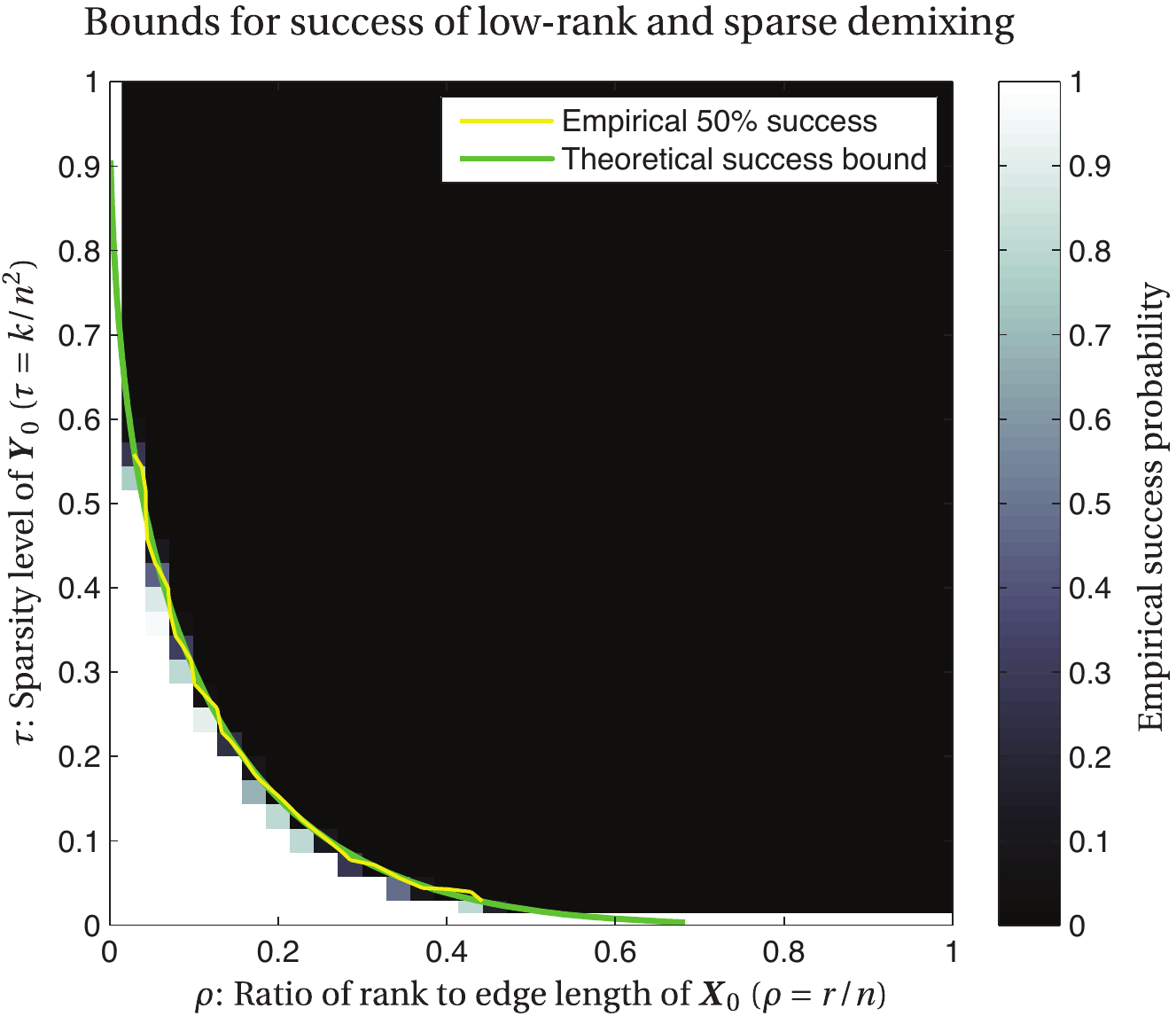}
  \caption{\textsl{Low-rank matrix recovery with sparse corruptions.}
    The horizontal axis is the normalized rank \(\rho = \rank(\vct
    X_0)/n\), and the vertical axis is the sparsity level~\(\tau =
    \mathrm{nnz}(\mtx Y_0)/n^2\).  The intensity of the background
    denotes the empirical probability that~\eqref{eq:lowrk-sparse}
    recovers \((\mtx X_0,\mtx Y_0)\) from the observation \(\mtx Z_0 =
    \mtx X_0 + \mathcal{Q}(\mtx Y_0)\).  In the region below the green
    curve, the convex demixing method~\eqref{eq:lowrk-sparse}
    recovers a low-rank matrix \(\mtx X_0\) from a randomly rotated
    sparse corruption \(\mathcal{Q}(\mtx Y_0)\) with overwhelming probability in high dimensions.  See
    Section~\ref{sec:low-rank-matrix} and
    Section~\ref{sec:low-rank-matrices} for more details.}
  \label{fig:l1-nuc-phase}
\end{figure}

Consider now the \emph{matrix} observation \(\mtx Z_0 = \mtx X_0 +
\mathcal{Q}(\mtx Y_0) \in \R^{n\times n}\), where \(\mtx X_0\) has low
rank, \(\mtx Y_0\) is sparse, and \(\mathcal{Q}\) is a random rotation
on \(\R^{n\times n}\).  This type of signal provides a highly stylized model
for applications such as latent variable
selection~\cite{Chandrasekaran2010a}, \cite{Chandrasekaran2011} and robust
principal component analysis~\cite{Candes2009}.  In these settings,
\(\mtx X_0\) has low rank because the underlying data is drawn from a
low-dimensional linear model, while \(\mathcal{Q}(\mtx Y_0)\)
represents a corruption. (Note, however, that this stylized model is not equivalent to pre- and post-multiplying \(\vct Y_0\) by independent random rotations of \(\R^n\).)   We aim to discover the matrix \(\mtx X_0\) given the corrupted observation \(\mtx Z_0\) and the rotation \(\mathcal{Q}\).  

We follow the now-familiar pattern of
Section~\ref{sec:recipe-demixing}.  The Schatten 1-norm
\(\sone{\cdot}\) serves as a natural complexity measure for the
low-rank structure of \(\mtx X_0\), and the matrix \(\ell_1\) norm
\(\lone{\cdot}\) is appropriate for the sparse structure of \(\mtx
Y_0\).  We further assume the side information \(\alpha =\lone{\mtx
  Y_0}\). We then solve
\begin{equation}
  \label{eq:lowrk-sparse}
  \minprog{}{\sone{\mtx X}}{\lone{\mtx Y} \le \alpha \;\;\text{and}\;\; \mtx X +
    \mathcal{Q}(\mtx Y) = \mtx Z_0.}
\end{equation}
This \CDMq\ succeeds if \((\mtx X_0,\mtx Y_0)\)
is the unique solution to~\eqref{eq:lowrk-sparse}.

Figure~\ref{fig:l1-nuc-phase} displays the results of a numerical
experiment on this approach to rank--sparsity demixing.  We take
the matrix side length \(n = 35\) and draw a random rotation
\(\mathcal{Q}\) for \(\R^{n\times n}\). For parameters \(0\le
\rho,\tau\le 1\), we generate matrices \(\vct
X_0\) and \(\vct Y_0\) such that \(\rank(\vct X_0) = [\rho n] \) and
\(\nnz(\vct Y_0) = [\tau n^2]\).  The background shading indicates the
empirical probability that~\eqref{eq:lowrk-sparse} succeeds given the
observation \(\mtx Z_0 = \mtx X_0 + \mathcal{Q}(\mtx Y_0)\).  We mark
the empirical \(50\%\) success probability with a yellow curve.  See
Section~\ref{sec:numerical-experiment-1} for the experimental details.

The results of this work show that, with high probability,
program~\eqref{eq:lowrk-sparse} succeeds so long as the pair
\((\rho,\tau)\) lies below the green curve on
Figure~\ref{fig:l1-nuc-phase}.  When the rank parameter \(\rho \) is
small, our theoretical bound closely tracks the phase transition
visible in the numerical experiment, although the bound appears loose
when \(\rho\) is larger.  Section~\ref{sec:low-rank-matrices} provides further details.

\subsubsection{Matrix demixing mix-and-match}
\label{sec:deconv-mix-match}

Our results are not restricted to the \CDMq
s~\eqref{eq:l1-const-1},~\eqref{eq:channel-code-1}
or~\eqref{eq:lowrk-sparse}.  Let us mention a few other situations we
can analyze using the theory developed in this work. With a tractable
convex program, it is possible to demix
\begin{itemize}\renewcommand{\itemsep}{0pt}
\item An orthogonal matrix from a matrix that is sparse in a random orthogonal
  basis,
\item A randomly oriented sign matrix from a sufficiently low-rank
  matrix, and
\item A randomly rotated low-rank matrix from an orthogonal matrix.
\end{itemize}
See Section~\ref{sec:hodgepodge} for the details.

\subsection{Theoretical insights}
\label{sec:lin-inv-intro}

Our approach reveals a number of theoretical insights.
\begin{description}\renewcommand{\itemsep}{0pt}
\item[Design of \CDMq s for incoherent
  structures.] In the incoherent regime we analyze, the parameters
  that determine when the \CDMq~\eqref{eq:genconv} succeeds reflect
  the structures in the constituent signals and the associated 
  complexity measures.  These summary parameters are independent of the
  relationship between the two incoherent structures.  We discuss this
  fact and its consequences for the design of demixing procedures
  in Section~\ref{sec:cons-design-deconv}.
\item[Connection with linear inverse problems.]  The parameters that
  determine success of the \CDMq~\eqref{eq:genconv} are closely
  related to number of random linear measurements required to identify
  a structured signal. In Section~\ref{sec:relat-line-inverse}, we
  leverage this relationship to compute these parameters from Gaussian
  width bounds developed in~\cite{ChaRecPar:12}.
\item[Phase transitions.] Our theory indicates that there is often a
  phase transition in the behavior of the demixing
  method~\eqref{eq:genconv}.  See  Section~\ref{sec:thresh-phen-conj} for a discussion of this point.
\end{description}

\subsection{Outline}
\label{sec:outline}

This work begins with demixing in the deterministic setting.
Section~\ref{sec:generic-setup-main} describes the geometry
of the \CDMq~\eqref{eq:genconv} and provides a geometric
characterization of successful demixing.

Section~\ref{sec:analys-via-integr} presents a random model for
incoherence along with some techniques from spherical integral
geometry that allow us to analyze this model.  In
Section~\ref{sec:find-weak-thresh}, these ideas yield theory that
predicts success and failure regimes for the \CDMq~\eqref{eq:genconv}.
Section~\ref{sec:line-inverse-probl-1} develops methods for computing
the parameters necessary to apply the theorems of
Section~\ref{sec:find-weak-thresh}.

In Section~\ref{sec:applications}, we analyze the application problems
described in Sections~\ref{sec:first-appl-deconv}
and~\ref{sec:other-applications}.  Section~\ref{sec:past-and-future}
concludes with a discussion of this work's place in the literature and
future directions.

\subsection{Notation and conventions}
\label{sec:notation}
All variables are real valued. We write \(\lfloor t\rfloor\), \(\lceil t\rceil\), and \([t]\) for the floor, ceiling, and rounded integer values of \(t\).  The signum function is \(\sgn(t):= t/\abs{t}\) for \(t\ne 0\) and \(\sgn(0):=0\).  Bold lowercase letters represent
vectors, and bold capital letters are matrices.  The \(i\)th element
of a vector is written \(x_i\) or \((\vct x)_i\), while the
\((i,j)\)th element of a matrix is \(X_{ij}\) or \((\mtx X)_{ij}\).
We express the transpose of \(\mtx X\) as \(\mtx X^\adj\).  The vector signum \(\sgn(\vct x)\) is defined by applying the signum elementwise.

The symbol \(\norm{\cdot}_{\ell_p}\) stands for the \(\ell_p\)
vector norm on \(\R^d\), defined by \(\norm{\vct x}_{\ell_p}^p \defeq \sum_{i=1}^d
|x_i|^p\) when \(1\le p <\infty\) and \(\linf{\vct x} \defeq \max_{i =
  1,\dotsc, d} |x_i|\) when \(p =\infty\).  The \(\ell_p\) norm of a
matrix treats the matrix as a vector and applies the corresponding
vector \(\ell_p\) norm.  The Schatten 1-norm \(\sone{\mtx X}\) is the
sum of the singular values of a matrix \(\mtx X\), while the operator, or
spectral, norm \(\opnorm{\mtx X}\) returns the maximum singular value
of \(\mtx X\).

We reserve the symbols \(f\) and \(g\) for convex functions.  A convex function may take the value
\(+\infty\), but we assume that all convex functions are \emph{proper};
that is, each convex function takes on at least one finite value and
never takes the value~\(-\infty\).

The Euclidean unit sphere in \(\R^d\) is the set \(\mathsf{S}^{d-1}\). The orthogonal group---the set of \(d\times d\) orthogonal matrices---is denoted \(\mathsf{O}_d\).  The subset of \(\mathsf{O}_d\) with determinant one (the \emph{special} orthogonal group) is \(\mathsf{SO}_d\).   For brevity, we use the term \emph{basis} to refer to an orthogonal matrix from \(\mathsf{O}_d\).  In the sequel, the letter \(\mtx Q\) will always refer to a basis.

The symbol \(\mathbb{P}\) denotes the probability of an event.  Gaussian
vectors and matrices have independent standard normal entries. A
\emph{random basis} is a matrix drawn from the Haar measure on
\(\mathsf{O}_d\).%

A special note is in order when our observations are matrices.  The
space of matrices \(\R^{m\times n}\) is equipped with a natural
isomorphism to \(\R^{mn}\) through the \(\vec(\cdot)\)
operator, which stacks the columns of a matrix on top of one another to
form a tall vector. We define a random basis \(\mathcal{Q}\) for
\(\R^{m\times n}\) by \(\mathcal{Q}(\mtx Y) =
\vec^{-1}\bigl(\mtx Q \vec(\mtx Y)\bigr)\), where
\(\mtx Q\) is a random basis for \(\R^{mn}\).  It is easily verified
that such a \(\mathcal{Q}\) is a random basis for \(\R^{m\times n}\)
equipped with the Euclidean structure induced by the
Frobenius norm.

We deal frequently with convex cones, which are positively homogeneous
convex sets. For any cone \(K\), we define the polar cone \[K^\circ
\defeq \{\vct y \mid \langle \vct y, \vct x\rangle \le 0 \text{ for
  all } \vct x \in K\}.\] A \emph{polyhedral cone} is the intersection
of a finite number of closed halfspaces, each containing the
origin. One important polyhedral cone is the nonnegative orthant,
defined by \[\R^d_+ \defeq \{ \vct x \in \R^d\mid x_i\ge 0, \; i =
1,\dotsc,d\}.\] For a set \(A \subset \R^d\), we write
\(\mathrm{conv}(A)\) for its convex hull and \(\overline {A}\) for its
closure.  Two cones \(K_1,K_2\subset \R^d\) are  \emph{congruent}, written \(K_1 \cong K_2\), if there is a basis \(\mtx U\in \mathsf{O}_{d}\) such that \(K_1 = \mtx U K_2\).  

\section{The geometry of demixing}
\label{sec:generic-setup-main}

This short section lays the geometric foundation for the rest of this
work.   Section~\ref{sec:feasible-cones} describes the local
behavior of convex functions in terms of special convex cones. In
Section~\ref{sec:geom-opt}, this geometric view yields a concise
characterization of successful demixing in terms of the
configuration of two cones.  The results in this section are deterministic, that is, they hold for any fixed basis \(\mtx Q\).

\subsection{Feasible cones}
\label{sec:feasible-cones}

The success of the \CDMq~\eqref{eq:genconv} depends on the
properties of the complexity measures \(f\) and \(g\) at the
structured vectors \(\vct x_0\) and \(\vct y_0\).   The following
definition captures the local behavior of a convex function.
\begin{definition}[Feasible cone]\label{def:feasible-cone}
  The \emph{feasible cone} of a convex function \(f\) at a point \(\vct
c  x\)  is defined as the cone of directions at which \(f\) is locally
  nondecreasing about \(\vct x\):
  \begin{equation*}
    \Fcone(f,\vct x) \defeq \bigcup_{\lambda > 0}\{\vct \delta \mid
    f(\vct x + \lambda \vct \delta) \le f(\vct x)\}.
  \end{equation*}
\end{definition}
\begin{example}[The feasible cone of the \(\ell_\infty\) norm at sign vectors]\label{ex:linf-feas}
  Let \(\vct x \in \{\pm 1\}^d\) be  a sign vector.  Then  it is easy to check that \(\linf{\vct x+\lambda \vct \delta}\le \linf{\vct x} = 1\) for some \(\lambda >0\) if and only if \( \sgn(\delta_i) = -x_i\) for all \(i = 1,2,\dotsc,d\).  Therefore, the feasible cone of the \(\ell_\infty\) norm at \(\vct x\) is congruent to the nonnegative orthant:
  \begin{equation*}
\Fcone(\linf{\cdot},\vct x) = \bigl\{\vct \delta \mid   \sgn(\vct \delta)= - \vct x \bigr\}\cong \R^d_+.
\end{equation*}
\end{example}
The feasible cone is always a convex cone containing zero, but it is
not necessarily closed.  If \(\mathcal{A}\) is a set of atoms and
\(\vct a\in\mathcal{A}\), the feasible cone
\(\Fcone(f_{\mathcal{A}},\vct a)\) of the atomic gauge
\(f_{\mathcal{A}}\) at the atom \(\vct a\) tends to be small because
the unit ball of \(f_{\mathcal{A}}\) is the smallest convex set
containing all of the atoms.  See Figure~\ref{fig:atomic-gauge} for an
illustration. The positive homogeneity of atomic gauges further
implies that the feasible cones of atomic gauges do not depend on the
scaling of a vector, in the sense that
\begin{equation*}
  \Fcone(f_{\mathcal{A}},\vct x) = \Fcone(f_{\mathcal{A}},\lambda
  \cdot \vct x)
\end{equation*}
for all \(\vct x \in \R^d\) and any \(\lambda > 0\).

\begin{remark}
  Definition~\ref{def:feasible-cone} is equivalent to the definition of the ``tangent cone'' appearing in~\cite[Eq.~(8)]{ChaRecPar:12}. However, that definition differs slightly from the standard definition of a tangent cone; cf.,~\cite[Thm.~6.9]{Rockafellar1998}.  The cone of feasible directions~\cite[p.~33]{Pataki2000} is the closest relative of the feasible cone that we have identified in the literature, and this is the source of our terminology.  \end{remark}

\begin{figure}[t!]
  \centering
  {\includegraphics[width=0.35\columnwidth]{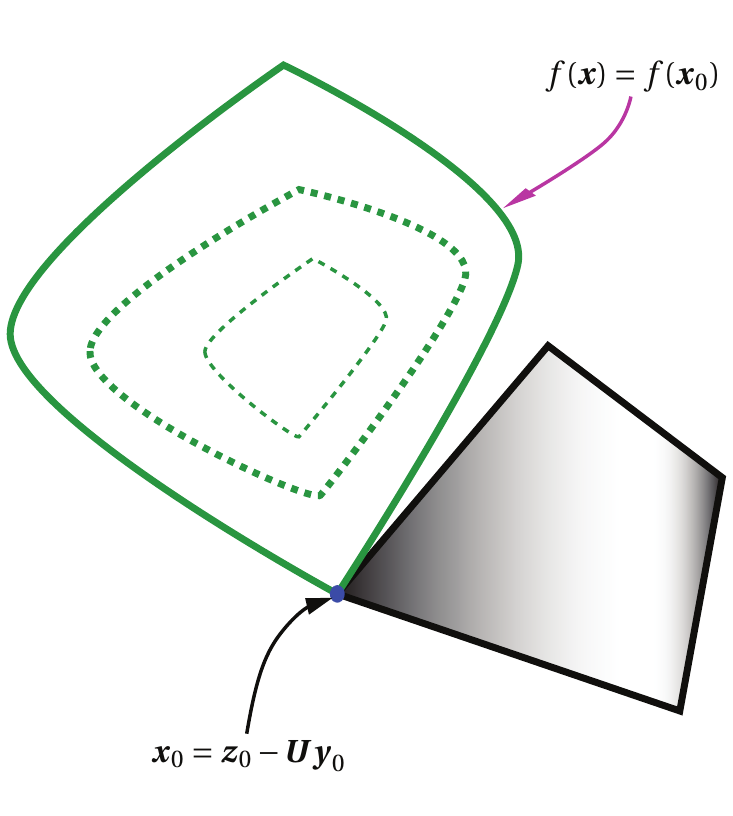}}
 \hspace{30pt} 
{\includegraphics[width=0.35\columnwidth]{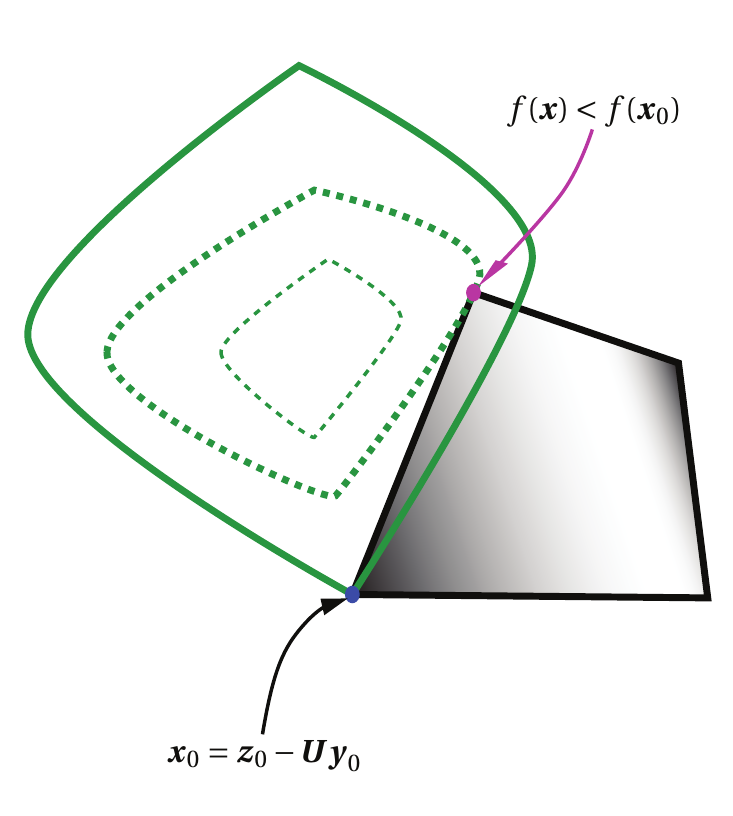}}
  \caption{\textsl{Geometry of demixing.}  [Left] Demixing
    succeeds: Every feasible perturbation about \(\vct x_0\) \textsl{(gray
    area)} increases the objective function \textsl{(green level lines)}.  [Right]
    Demixing fails: Some feasible perturbations decrease the
    objective value.  In each panel, the success or failure of the
    \CDMq~\eqref{eq:genconv} is determined by a configuration of two
    cones.  This fact forms the content of
    Lemma~\ref{lem:main-geom-introduction}.}
  \label{fig:geometry}
\end{figure}
\subsection{A geometric characterization of optimality}
\label{sec:geom-opt}

The following lemma provides a geometric characterization for success
in the \CDMq~\eqref{eq:genconv} in terms of the configuration of
two feasible cones.  This is the main result of this section.
\begin{lemma}
  \label{lem:main-geom-introduction} \label{lem:triv-intersect}
  Program~\eqref{eq:genconv} succeeds at demixing \((\vct x_0,\vct
  y_0)\) if and only if \(\Fcone(f,\vct x_0) \bigcap \bigl(-\mtx Q
  \Fcone(g,\vct y_0)\bigr) = \{\mathbf{0}\}.\)
\end{lemma}
In words, the demixing method~\eqref{eq:genconv} succeeds if and only if the two feasible cones are rotated so that they intersect trivially. Intuitively, we expect that many bases \(\mtx Q\) satisfy this condition when the two feasible cones are small.  This observation provides further support for choosing atomic gauges as our complexity measures.  We illustrate Lemma~\ref{lem:triv-intersect} in Figure~\ref{fig:geometry}.

The proof of Lemma~\ref{lem:main-geom-introduction} requires two
technical propositions.  The first is an alternative characterization
of feasible cones.
\begin{proposition}\label{prop:interval}
  Let \(f\) be a convex function.  Then \(\vct \delta \in
  \Fcone(f,\vct x)\) if and only if there is a number \(\lambda_0 > 0 \)
  such that, for all \(\lambda \in [0,\lambda_0]\), we have \(f(\vct x
  +\lambda \vct \delta)\le f(\vct x)\).
\end{proposition}
\begin{proof}[Proof]
  The ``if'' part is immediate: given any such \(\lambda_0\), the assumption \(f(\vct x + \lambda_0 \vct \delta)\le f(\vct x)\) implies \(\vct \delta \in \Fcone(f,\vct x)\) by the definition of feasible cones.  The other direction follows from convexity.  Indeed, suppose \(\vct \delta \in \Fcone(f,\vct x)\), so that there exists a number \(\lambda_0 > 0\) for which \(f(\vct x + \lambda_0\vct \delta ) \le f(\vct x)\). By convexity, the map \(\lambda \mapsto f(\vct x + \lambda \vct \delta)\) lies below the chord connecting \(0\) and \(\lambda_0\), which is the claim.
\end{proof}

Our second technical proposition is a change of variables formula for the
feasible cone under a nondegenerate affine transformation.
\begin{proposition} \label{prop:trans-rule} Let \(g\) be any convex
  function, and define \(h(\vct x) \defeq g\bigl(\mtx A^{-1}(\vct z -\vct x)\bigr) \)
  for some invertible matrix \(\mtx A\).  Then
  \begin{math}
    \Fcone(h, \vct x) 
    = -    \mtx{A}\Fcone\bigl(g,\mtx A^{-1}( \vct z-\vct x)\bigr).
 \end{math}
\end{proposition}
The proof of Proposition~\ref{prop:trans-rule} follows directly from
the definition of a feasible cone. We omit the details.
\begin{proof}[Proof of Lemma~\ref{lem:main-geom-introduction}]
 Because \(\mtx Q\) is unitary, we may eliminate the variable \(\vct y\)
  in~\eqref{eq:genconv} via the equality constraint \(\vct y = \mtx
  Q^\adj(\vct z_0 - \vct x)\).  Therefore, \((\vct x_0, \vct y_0)\) is
  the unique optimum of~\eqref{eq:genconv} if and only if \(\vct x_0\)
  is the unique optimum of
  \begin{equation}
    \label{eq:equal-removed}
    \minprog{}{f(\vct x)}{g\bigl(\mtx Q^\adj(\vct z_0 - \vct x)\bigr)\le
      g\bigl(\mtx Q^\adj(\vct z_0 - \vct x_0)\bigr) = g(\vct y_0),}
  \end{equation}
  with decision variable \(\vct x\). The equality
  in~\eqref{eq:equal-removed} follows from the definition \(\vct z_0 =
  \vct x_0 + \mtx Q \vct y_0\). The original claim thus reduces to the
  statement that \(\vct x_0\) is the unique optimum
  of~\eqref{eq:equal-removed} if and only if \(\Fcone(f,\vct x_0) \cap\bigl(
  -\mtx Q \Fcone(g,\vct y_0)\bigr) = \{\zerovct\}\).  The rest of the proof
  is devoted to this claim.

  \((\Leftarrow)\) Suppose \(\Fcone(f,\vct x_0) \cap \bigl(-\mtx Q
  \Fcone(g,\vct y_0)\bigr) = \{\zerovct\}\).  We show that \(\vct
  x_0\) is the unique optimum of \eqref{eq:equal-removed} by verifying
  that the strict inequality \(f(\vct x)> f(\vct x_0)\) holds for any
  feasible point \(\vct x\) of~\eqref{eq:equal-removed} with \(\vct x
  \ne \vct x_0\).

  To this end, assume \(\vct x \) is feasible
  for~\eqref{eq:equal-removed} and \(\vct x \ne \vct x_0\).
  Feasibility of \(\vct x\) is equivalent to
  \begin{equation*}
    g\bigl(\mtx Q^\adj (\vct z_0-\vct x )\bigr )
    \le g\bigl(\mtx Q^\adj(\vct z_0 - \vct x_0)\bigr) =g(\vct y_0).
  \end{equation*}
  By definition of the feasible cone and the transformation rule of
  Proposition~\ref{prop:trans-rule}, the inequality above implies
  \begin{equation*}
    \vct x -\vct x_0 \in
    \Fcone\Bigl(g\bigl(\mtx Q^\adj(\vct z_0 - \cdot)\bigr),\vct
    x_0\Bigr) 
    =
    -\mtx Q \Fcone\bigl(g,\mtx Q^\adj(\vct z_0-\vct x_0)\bigr) 
    =  -\mtx Q\Fcone(g,\vct y_0);
  \end{equation*}
 The final equality above follows from \(\vct z_0 = \vct x_0 + \mtx Q \vct
  y_0\).  

  Since \(\vct x \ne \vct x_0\), the assumption \(\Fcone(f,\vct
  x_0) \cap \bigl(-\mtx Q \Fcone(g,\vct y_0)\bigr) = \{\zerovct\}\)
  implies \(\vct x-\vct x_0\notin \Fcone(f,\vct x_0)\).  By the
  definition of feasible cones, we must have \(f(\vct x) = f\bigl(\vct x_0
  + (\vct x - \vct x_0)\bigr) > f(\vct x_0)\).  We have deduced 
  \(f(\vct x)>f(\vct x_0)\) for every feasible \(\vct x \ne \vct
  x_0\), and so conclude that \(\vct x_0\) is the unique optimum
  of~\eqref{eq:equal-removed}.

  \((\Rightarrow)\) Suppose \(\vct x_0\) is the unique
  optimum of~\eqref{eq:equal-removed}.  Let \(\vct \delta \) be some
  vector in the intersection \(\Fcone(f,\vct x_0)\cap \bigl(-\mtx Q
  \Fcone(g,\vct y_0)\bigr)\).  We must show \(\vct \delta =
  \mathbf{0}\).
  
  As \(\vct \delta \in\Fcone(f,\vct x_0)\),
  Proposition~\ref{prop:interval} implies that \(f(\vct x_0) \ge f(\vct x_0
  + \lambda \vct \delta) \) for all sufficiently small \(\lambda >0\).
  Applying Proposition~\ref{prop:interval} to the fact \(-\mtx Q^\adj
  \vct \delta \in \Fcone(g,\vct y_0)\) yields
  \begin{equation*}
    g(\vct y_0) 
    \ge g( \vct y_0 - \lambda\mtx Q^* \vct \delta)
    =   g\Bigl(\mtx Q^*\bigl(\vct z_0 - (\vct x_0 + \lambda \vct \delta)\bigr)\Bigr)  
  \end{equation*}
  for all sufficiently small \(\lambda >0\).  We have used the
  relation \(\vct z_0 = \vct x_0 + \mtx Q\vct y_0\) again here.  In
  summary, for some small enough \(\lambda>0\), the perturbed point
  \(\vct x_0 + \lambda \vct \delta\) is feasible
  for~\eqref{eq:equal-removed}, and its objective value is no larger
  than \(f(\vct x_0)\).  In other words, \(\vct x_0 + \lambda \vct
  \delta\) is an optimal point of~\eqref{eq:equal-removed}.  But
  \(\vct x_0\) is the unique optimal point of~\eqref{eq:equal-removed}
  by assumption, so we must have \(\vct \delta = \vct 0\).  This is
  the claim.
\end{proof}

\section{Background from integral geometry}
\label{sec:analys-via-integr}

By coupling the random basis model of Section~\ref{sec:random-basis-model} with the optimality condition of Lemma~\ref{lem:triv-intersect}, the optimality condition for the demixing method~\eqref{eq:genconv} boils down to a geometric question:
{\emph{When does a randomly oriented cone strike a fixed cone?}}

This section provides a background in spherical integral geometry, a subfield of integral geometry that studies random configurations of cones and quantities related to these configurations~\cite[Section~6.5]{Schneider2008}.  This theory provides an \emph{exact} expression, called the spherical kinematic formula, for the probability that the \CDMq~\eqref{eq:genconv} succeeds under our random model.  Unfortunately, the quantities involved in the spherical kinematic formula are typically difficult to compute.  To ease this burden, Section~\ref{sec:asympt-thresh-phase} defines geometric summary parameters that greatly simplify the application of the spherical kinematic formula. 

We use subspaces and the nonnegative orthant as running examples to illustrate the concepts from integral geometry.  These are not simply toy examples. A subspace plays an important role in the linear inverse problems in Section~\ref{sec:line-inverse-probl-1}, while the orthant, which is congruent to the feasible cone of the \(\ell_\infty\) norm at a sign vector, appears at several points in our examples in Section~\ref{sec:applications}.

\subsection{Spherical intrinsic volumes}
\label{sec:backgr-from-spher}

\begin{figure}[t]
  \centering
  \includegraphics[width=0.45\columnwidth]{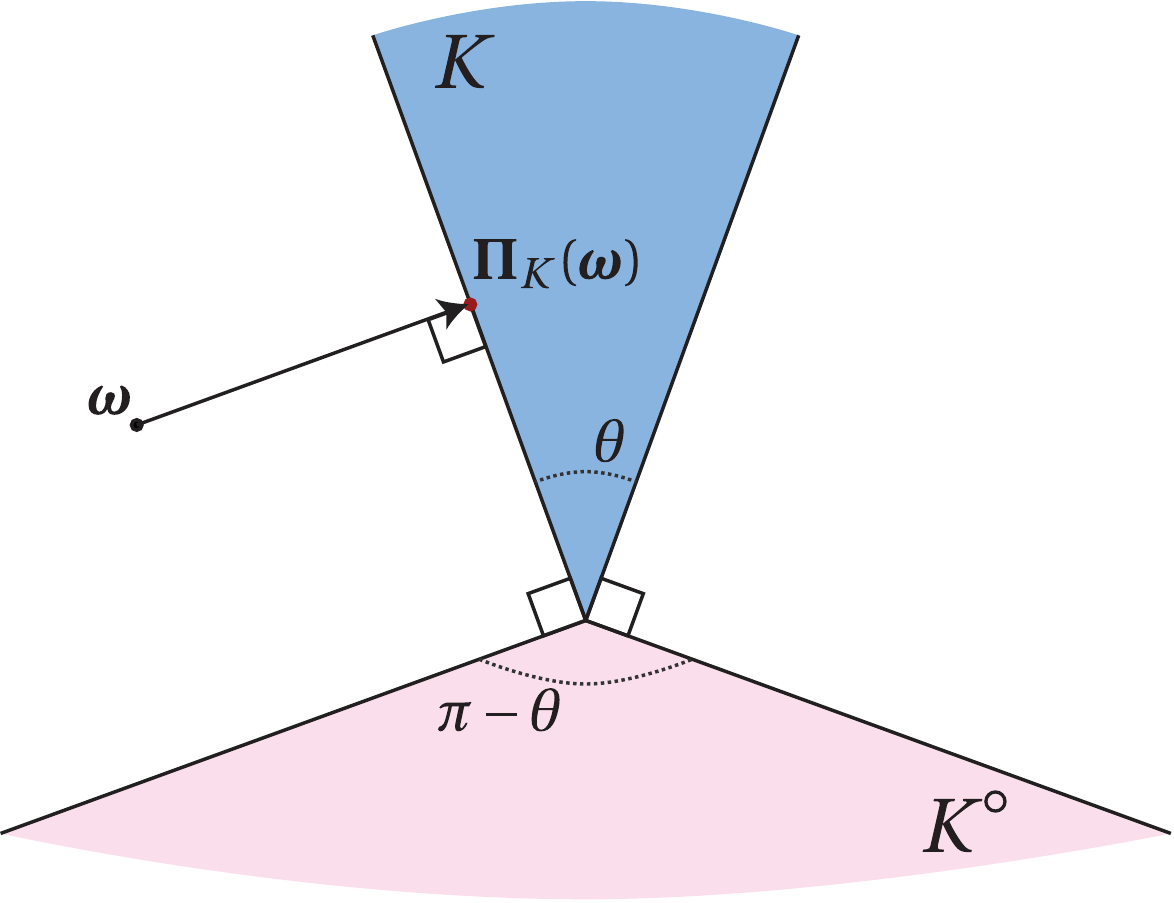}
  \caption{\textsl{Spherical intrinsic volumes in \(\R^2\).} A convex cone \(K\subset \R^2\) of solid angle \(\theta\) has four faces: one 2-dimensional face \textsl{(light blue)}, two 1-dimensional faces \textsl{(heavy blue lines)}, and one 0-dimensional face \textsl{(blue dot)}.  The projection \(\Proj_K(\vct w)\) of a Gaussian vector \(\vct \omega\) onto \(K\) lies in the 2-dimensional face when \(\vct \omega\) is in the blue region; a 1-dimensional face when \(\vct \omega\) is in the white region; and the 0-dimensional face when \(\vct \omega\) is in the red region.  By Definition~\ref{def:sphere-intrinsic-vols} of the spherical intrinsic volumes, we have \(v_1(K) = \theta/(2\pi)\), \(v_0(K) = 1/2\), and \(v_{-1}(K) = (\pi-\theta)/(2\pi)\). }
  \label{fig:conefig}
\end{figure}

We begin our introduction to integral geometry with fundamental geometric parameters known as spherical intrinsic volumes~\cite{McMullen1975}. Spherical intrinsic volumes quantify geometric properties of convex cones such as the fraction of space a cone consumes (a type of volume), the fraction of space taken by the corresponding dual cone, and quantities akin to surface area.  The following characterization~\cite[Proposition~4.4.6]{Amelunxen2011} is convenient.  \begin{definition}[Spherical intrinsic volumes]\label{def:sphere-intrinsic-vols}
  Let \(K\subset \R^d\) be a polyhedral convex cone, and define the
  Euclidean projection onto \(K\) by
  \begin{equation*}
    \Pi_K(\vct x) \defeq \argmin_{\vct y \in K} \norm{ \smash{\vct x - \vct y }}_{\ell_2}.
  \end{equation*}
  For \(i = -1, 0,\dotsc, d-1\), we define the \(i\)th \emph{spherical
    intrinsic volume of \(K\)} by
  \begin{equation*}
    v_i(K) \defeq \mathbb{P}\left\{ \begin{array}{l}\Pi_K(\vct \omega) \text{ lies in the
      relative interior} \\ \text{of an
      \((i+1)\)-dimensional face of } K \end{array}\right\},
  \end{equation*}
  where the vector \(\vct \omega\) is drawn from the standard Gaussian
  distribution on \(\R^d\).
\end{definition}

Although computing spherical intrinsic volumes is often challenging, a
direct application of Definition~\ref{def:sphere-intrinsic-vols} can
bear fruit in some situations. In Figure~\ref{fig:conefig}, we demonstrate how to compute spherical intrinsic volumes for cones in \(\R^2\).   A subspace of \(\R^d\) provides another simple example.

\begin{proposition}[Spherical intrinsic volumes of a subspace]\label{prop:subspace-iv}
  Suppose \(L\subset \R^d\) is a linear subspace of dimension \(n\). Then
  \(v_i(L) = \delta_{i,n-1}\), where \(\delta_{i,j}\) is the
  Kronecker \(\delta\) function.
\end{proposition}
\begin{proof}
  A subspace \(L\) of dimension \(n\) is a polyhedral cone with a
  single face, with dimension \(n\).  The projection of any point onto
  \(L\) lies in the relative interior of this face. The claim follows
  immediately from Definition~\ref{def:sphere-intrinsic-vols}.
\end{proof}

The computation of the spherical intrinsic volumes for the nonnegative
orthant requires only some elementary probability theory.    We illustrate the following
result in Figure~\ref{fig:orth-vols}.

\begin{figure}[t!]
  \centering
  \includegraphics[width=0.55\columnwidth]{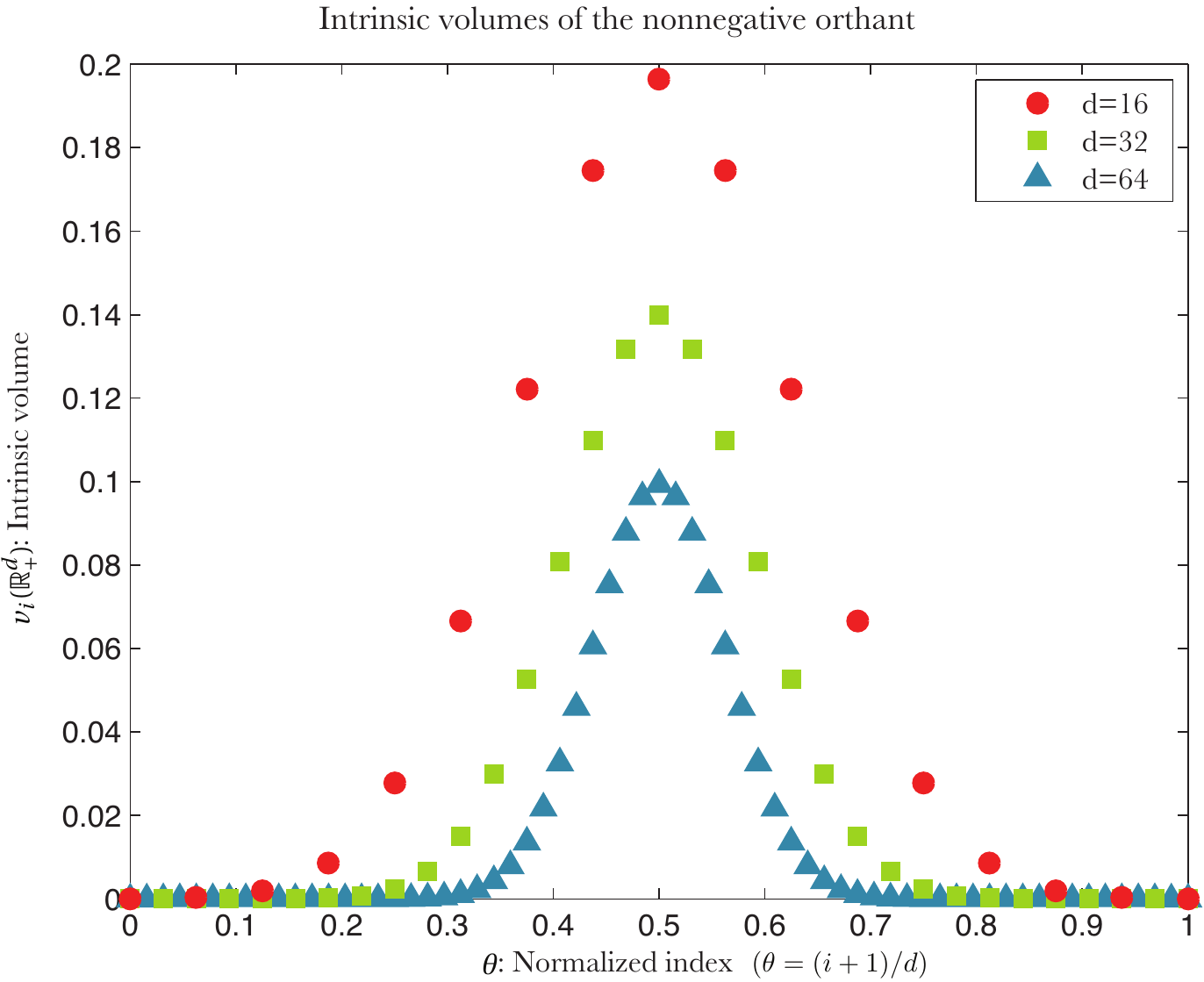}
  \caption{\textsl{Spherical intrinsic volumes of the orthant.} The
    spherical intrinsic volumes of the nonnegative orthant are given by
    scaled binomial coefficients~\eqref{eq:pos-orth-intrinsic-vol}.
    The scale on the horizontal axis is the normalized index \(\theta
    = (i+1)/d\).}
  \label{fig:orth-vols}
\end{figure}

\begin{proposition}[Spherical intrinsic volumes of the
  orthant~\protect{\cite[Example~4.4.7]{Amelunxen2011}}]
  \label{prop:orthant-intrinsic-volumes}
  The spherical intrinsic volumes of the nonnegative orthant ~\(\R^d_+\)
  are given by the binomial sequence
  \begin{equation}\label{eq:pos-orth-intrinsic-vol}
    v_i(\R^d_+) = 2^{-d}\binom{d}{i+1}
  \end{equation}
  for \(i=-1,0,\dotsc,d-1\).
\end{proposition}
\begin{proof}
  The Euclidean projection onto the nonnegative orthant is given by the
  componentwise threshold operation
  \begin{equation*}
    \left(\Pi_{\R^d_+}(\vct \omega)\right)_i = \max\{\omega_i, 0\}.
  \end{equation*}
  Therefore, the projection \({\Pi_{\R^d_+}(\vct \omega)}\) lies
  in the relative interior of an \((i+1)\)-dimensional face of
  \(\R^d_{+}\) if and only if \(\vct \omega\) has exactly \(i+1\)
  strictly positive values.

  When \(\vct \omega\) is drawn from a standard Gaussian distribution, the
  number of positive entries is distributed as a binomial random
  variable.  Hence, \(\smash{\Pi_{\R^d_+}(\vct \omega ) }\) lies in an
  \((i+1)\)-dimensional face of \(\R^d_+\) with probability
  \(2^{-d}\binom{d}{i+1}\).  This is precisely the definition of
  the \(i\)th spherical intrinsic volume \(v_i(\R^d_+)\).
\end{proof}

\begin{remark}[Extension to nonpolyhedral cones] The definition of spherical intrinsic volumes extends
  to all closed convex cones by approximation with polyhedral cones,
  but note that the probabilistic characterization above \emph{does
    not hold} for nonpolyhedral cones. The technical details involve
  continuity properties of the spherical intrinsic volumes under the
  spherical Hausdorff metric. This theory is developed
  in~\cite{Glasauer1995}, \cite{Glasauer1996}.
  See~\cite[Ch.~6.5]{Schneider2008} for a self-contained overview of
  spherical integral geometry developed via polyhedral approximation,
  or see~\cite[Sec.~4]{Amelunxen2011} for a development using tools
  from differential geometry.
 \end{remark}

\subsection{Key facts from integral geometry}
\label{sec:spher-kinem-form}

We now collect some of the properties of spherical intrinsic volumes that are required for our development.  We start with some elementary facts.
\begin{fact} \label{fact:obvious} Let \(K \subset \R^d\) be a closed convex
  cone. Then 
  \begin{enumerate}\setlength{\itemsep}{0pt}\setlength{\topsep}{0pt}\setlength{\parsep}{0pt}
  \item (Positivity) \(v_i(K) \ge 0 \) for each \(i =
    -1,0,\dotsc,d-1\), \label{item:obv-positivity}
  \item (Unit-sum) \(\sum_{i=-1}^{d-1} v_i(K) =
    1\), and \label{item:obv-unity}
  \item (Basis invariance) For any basis \(\mtx U\in \mathsf{O}^d\)
    and index \(i=-1,0,\dotsc,d-1\),\label{item:obv-invariant} we have
    \(v_i(\mtx U K) = v_i(K)\).
  \end{enumerate}
\end{fact}
\begin{proof}[Proof sketch]
  First, assume \(K\) is a polyhedral cone.  The positivity of
  spherical intrinsic volumes follows from the positivity of
  probability.  The unit-sum rule follows immediately from the
  fact that the projection \(\Pi_{K}(\vct x)\) lies in the relative
  interior of a unique face of \(K\).  %
  Finally, the basis invariance of spherical intrinsic volumes is
  immediate from the corresponding invariance of the Gaussian
  distribution.  Continuity of the spherical intrinsic volumes under
  the spherical Hausdorff metric implies that these facts must hold for all
  closed convex cones by approximation with polyhedral cones.
\end{proof}

We conclude our background discussion with the following remarkable
formula for the probability that a randomly oriented cone strikes a
fixed cone.  In view of Lemma~\ref{lem:main-geom-introduction}, this
result is fundamental to our understanding of the probability that the
\CDMq~\eqref{eq:genconv} succeeds under the random basis model.
\begin{fact}[Spherical kinematic formula~\protect{\cite[p.~261]{Schneider2008}}]\label{fact:sphere-kin}
  Let \(K\) and \(\tilde K\) be closed convex cones in \(\R^d\),
  at least one of which is not a subspace, and let \(\mtx Q\) be a
  random basis.  Then
  \begin{equation}\label{eq:sphere-kin}
    \mathbb{P}\bigl\{ K \cap \mtx{Q} \tilde K\ne \{\vct 0\} \bigr\} =
      \sum_{k=0}^{d-1} \bigl(1 + (-1)^k\bigr)\sum_{i=k}^{d-1} v_i (K)\cdot
      v_{d-1-i+k}(\tilde K).
  \end{equation}
\end{fact}
\begin{remark}
  Fact~\ref{fact:sphere-kin} is usually stated for random rotations
  drawn according to the Haar measure on the special orthogonal group
  \(\mathsf{SO}_d\), but the unitary invariance given by
  Fact~\ref{fact:obvious}.\ref{item:obv-invariant} readily implies that
  the spherical kinematic formula holds for \(\mtx Q\) drawn from the
  Haar measure on the orthogonal group \(\mathsf{O}_d\).
\end{remark}
\begin{remark}\label{rem:symmetry}
  The spherical Gauss--Bonnet formula
  (Fact~\ref{fact:spherical-gauss-bonnet} in
  Appendix~\ref{sec:region-fail-proof}) can be used to eliminate the
  apparent asymmetry between \(K\) and \(\tilde K\)
  in~\eqref{eq:sphere-kin}.  In particular, the identity \[
  \mathbb{P}\bigl\{ K \cap \mtx{Q} \tilde K\ne \{\vct 0\} \bigr\} =
  \mathbb{P}\bigl\{ \mtx{Q}K \cap \tilde K\ne \{\vct 0\} \bigr\}\]
  holds for any convex cones \(K\) and \(\tilde K\).
\end{remark}

\subsection{High-dimensional decay of spherical intrinsic volumes}
\label{sec:asympt-thresh-phase}

The spherical kinematic formula~\eqref{eq:sphere-kin}, coupled with
the geometric optimality conditions of
Lemma~\ref{lem:main-geom-introduction}, provides an exact expression
for the probability that the demixing problem~\eqref{eq:genconv}
succeeds.  Nevertheless, the formula involves the spherical intrinsic
volumes of two cones, and it is challenging to determine these quantities
directly from the definition except in simple situations.

To confront this challenge, we seek summary statistics for the intrinsic volumes as the dimension \(d\to \infty\). We motivate our approach with the orthant. In Figure~\ref{fig:orth-vols}, we see that the spherical intrinsic volumes of the orthant \(v_i(\R^d_+)\) decay rapidly as \(d\to \infty\) when the index \(i\) falls outside of the region \(i \approx d/2\).  The intrinsic volumes with indices \(i\) far away from  \(d/2\) will contribute very little to the sum~\eqref{eq:sphere-kin} appearing in the kinematic formula.  This observation simplifies the application of the kinematic formula when one of the cones is congruent to the orthant.

In general, we might hope that for some sequence of cones \(K^{(d)}\in \R^d\), the spherical intrinsic volumes  \(v_i(K^{(d)})\) decay rapidly as \(d\to \infty\) for indices \(i\) outside of some interval \([\kappa_\star d,\theta_\star d]\). We codify this behavior with \emph{decay thresholds} that indicate which intrinsic volumes are very small as the ambient dimension \(d\) grows.  
\begin{definition}[Decay threshold]\label{def:aeub}
  Let \(\Index \subset \mathbb{N}\) be an infinite set of indices, and
  suppose \(\{K^{(d)}\mid d \in \Index \}\) is an ensemble of closed
  convex cones with \(K^{(d)} \subset \R^d\) for each \(d\in \Index\).
  We say that \(\theta_\star \in [0,1]\) is an \emph{upper decay
    threshold} for \(\{K^{(d)}\}\) if, for every \(\theta >
  \theta_\star\), there exists an \(\varepsilon > 0\) such that, for
  all sufficiently large \(d\in \Index\), we have
  \begin{equation}\label{eq:exp-upper-bd}
    v_i(K^{(d)}) \le \mathrm{e}^{-\varepsilon d }
    \;\;\;\text{for every \;\(i \ge \lceil \theta d\rceil\).}
  \end{equation}
 On the other hand, we
  say that \(\kappa_\star \in [0,1]\) is a \emph{lower decay
    threshold} for \(\{K^{(d)}\}\) if, for every
  \(\kappa < \kappa_\star\), there exists an \(\eps >0\) such that,
  for all \(d\in \Index\) sufficiently large, we have
  \begin{equation*}
    v_i(K^{(d)}) \le \mathrm{e}^{-\varepsilon d} \;\;\;\text{  for every \;\(i \le \lceil \kappa d\rceil\).  }
  \end{equation*}
  We extend these definitions to non-closed cones by taking the
  closure. We say \(\theta_\star \) is an upper decay threshold for
  \(\{K^{(d)}\}\) if and only if it is an upper decay threshold for
  \(\bigl\{\overline{K}{}^{(d)}\bigr\}\).  Similarly, \(\kappa_\star \)
  is a lower decay threshold for \(\{K^{(d)}\}\) if and only if it is
  a lower decay threshold for \(\bigl\{\overline{K}{}^{(d)}\bigr\}\).
\end{definition}
\noindent When it will not cause confusion, we omit the index set \(\Index\).
\begin{remark}[Nonuniqueness of decay thresholds]
  The upper decay threshold \(\theta_\star\) for an ensemble \(\{K^{(d)}\}\) defined above is not unique since any \(\theta' > \theta_\star\) is also a decay threshold for \(\{K^{(d)}\}\). %
 An analogous comment holds for the lower decay threshold. 
\end{remark}
\subsubsection{Examples of decay thresholds}
\label{sec:return-two-exampl}

Since the intrinsic volumes are positive and sum to one (Fact~\ref{fact:obvious}), not every
intrinsic volume is exponentially small in the ambient dimension \(d\).  In particular, the inequality
\(\kappa_\star \le \theta_\star\) between lower and upper decay thresholds always holds. In practice, however, we can often find decay thresholds that satisfy the equality \(\kappa_\star = \theta_\star\), as the following examples demonstrate.

\begin{proposition}[Upper decay threshold for subspaces]
  \label{prop:subspace-aeub}
  Let \(\Index\subset \mathbb{N}\) be an infinite set of indices, and
  let \(\{L^{(d)}\mid d \in \Index\}\) be an ensemble of linear
  subspaces with \(L^{(d)}\subset \R^d\) for each \(d \in \Index\).
  Suppose there exists a parameter \(\sigma\in [0,1]\) such that
  \(\mathrm{dim}(L^{(d)}) = \lceil \sigma d\rceil\).  Then
  \(\theta_\star=\sigma\) is an upper decay threshold for the ensemble
  \(\{L^{(d)}\}\), and \(\kappa_\star=\sigma\) is a lower decay
  threshold for the ensemble \(\{L^{(d)}\}\).
\end{proposition}
\begin{proof}
  Let \(n = n^{(d)} = \dim(L^{(d)})\).  By
  Proposition~\ref{prop:subspace-iv}, we have \(v_{i}(L^{(d)}) =
  \delta_{i,n-1}\).  Then for any \(\theta > \sigma\), we have
  \(\lceil \theta d\rceil> \lceil \sigma d \rceil\) for all large
  enough \(d\in \Index\), so that
  \begin{equation}\label{eq:subspace-exp-bound}
    v_i(L^{(d)})  = 0 < \econst^{-  d}
  \end{equation}
  for all \(i \ge \lceil \theta d\rceil\) and \(d\) sufficiently
  large.  By definition, \(\theta_\star = \sigma\) is an upper decay
  threshold for the ensemble of subspaces \(\{L^{(d)}\mid d \in \Index\}\).  The
  demonstration that \(\kappa_\star = \sigma\) is a lower decay
  threshold for \(\{L^{(d)}\mid d \in \Index\}\) follows in the same way.
\end{proof}

\begin{figure}[t!]
  \centering
  \includegraphics[width=\figwidth\columnwidth]{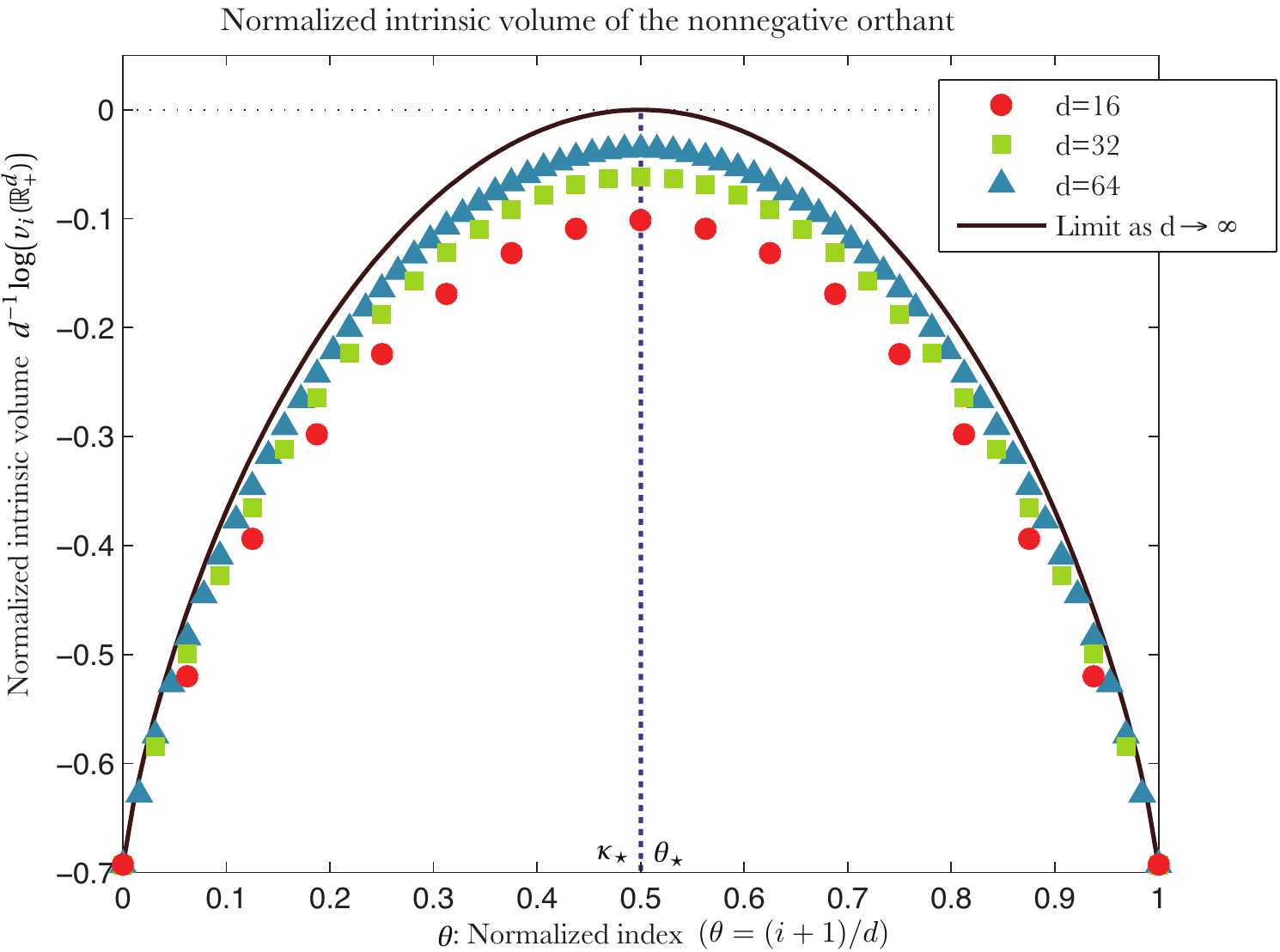}
  \caption{
    \textsl{Normalized spherical intrinsic volumes of the orthant.}  The graph illustrates the computation in Proposition~\ref{prop:orth-thresh}.  The normalized spherical intrinsic volumes \(d^{-1} \log\bigl(v_i(\R^d_+)\bigr)\) appear against the normalized index \(\theta = (i+1)/d\) for
    several values of dimension \(d\). The solid curve is the uniform
    limit \(H(\theta) - \log(2)\) of these rescaled volumes, where
    \(H(\theta)\) is the bit entropy~\eqref{eq:entropy}.  The
    smallest possible upper decay threshold is the rightmost point
    where the solid curve crosses the zero level line (dotted), and
    the largest possible lower decay threshold is the leftmost point
    where the zero level line crosses the solid curve. The solid curve
    takes its unique maximum value of zero at \(\theta = \frac{1}{2}\)
    so that \(\theta_\star=\frac{1}{2}\) is an upper decay threshold
    and \(\kappa_\star = \frac{1}{2}\) is a lower decay threshold for
    the ensemble \(\{\R^d\mid d =1,2,\dotsc \}\) of nonnegative
    orthants. }
  \label{fig:orth-vols-norm}
\end{figure}
\begin{proposition}[Upper decay threshold for the nonnegative
  orthant]\label{prop:orth-thresh}
  Let \(\Index\subset \mathbb{N}\) be an infinite index set.  The
  value \(\theta_\star = \frac{1}{2}\) is an upper decay threshold for
  the ensemble~\(\{\R^d_+\mid d\in \Index\}\) of nonnegative orthants,
  and \(\kappa_{\star}= \frac{1}{2}\) is a lower decay threshold for
  \(\{\R^d_+\mid d\in \Index\}\).
\end{proposition}
\begin{proof}
  It is well known (see,
  e.g.,~\cite[Eq.~(3.4)]{Donoho2006}) that, for any \(\theta \in
  [0,1]\), we have
  \begin{equation}\label{eq:binom-entropy-limit}
    \frac{1}{d} \log\binom{d}{\lceil \theta d \rceil} \to H(\theta) \text{
      uniformly in \(\theta\) as } d \to \infty,
  \end{equation}
  where 
  \begin{equation}
    H(\theta) \defeq -\theta \log(\theta) - (1-\theta)
    \log(1-\theta)\label{eq:entropy}
  \end{equation}
  is the natural entropy; be aware that the logarithms are
  base-\(\econst\) rather than the customary base-\(2\) used in
  information theory.  Basic calculus shows that \(H(\theta)\)
  achieves its unique maximum at \(\theta_\star = \frac{1}{2}\), where
  it has maximum value \(H\bigl(\frac{1}{2}\bigr)=\log(2)\).
  
  Let \(\theta > \frac{1}{2}\).  By continuity of \(H\), there is
  an \(\varepsilon >0\) such that \(H(\tilde\theta)<\log(2)-\eps\)
  for all \(\tilde \theta \ge \theta > \theta_\star\).  Continuity of
  the exponential and the uniform convergence
  in~\eqref{eq:binom-entropy-limit} together imply that, for all
  sufficiently large \(d\in \Index\), and any \(i \ge \lceil \theta d\rceil\),
  we have
  \begin{equation*}
    2^{-d} \binom{d}{i} \le \exp\bigl( - d \log(2) +
      d (\log(2) - \varepsilon)\bigr) = \mathrm{e}^{-\varepsilon d}.
  \end{equation*}
  The left-hand side is equal to the spherical intrinsic volume
  \(v_{i}(\R^d_+)\) by
  Proposition~\ref{prop:orthant-intrinsic-volumes}, so we see that
  \(\theta_\star = \frac{1}{2}\) is an upper decay threshold for
  \(\{\R^d_+\mid d \in \Index\}\).

  The proof that \(\kappa_\star = \frac{1}{2}\) is a lower decay
  threshold for \(\{\R^d_+\mid d \in \Index\}\) follows along similar
  lines.  Briefly, for any \(\kappa< \frac{1}{2}\), there is an \(\eps
  > 0\) such that \(H(\tilde \kappa)< \log(2)-\eps\) for every
  \(\tilde \kappa \le \kappa\).  For the same reasons as before, when
  \(d\) is large enough, we have
  \begin{equation*}
    2^{-d} \binom{d}{i} \le \exp\bigl( - d \log(2) +
      d (\log(2) - \varepsilon)\bigr) = \mathrm{e}^{-\varepsilon d}
  \end{equation*}
  for all \(i\le \lceil \kappa d\rceil\).  We conclude that
  \(\kappa_\star = \frac{1}{2}\) is a lower decay threshold for the
    ensemble of nonnegative orthants \(\{\R^d_+\mid d \in \Index\}\).
\end{proof}

Figure~\ref{fig:orth-vols-norm} illustrates the computation above.  We discuss other approaches for finding decay thresholds in
Section~\ref{sec:line-inverse-probl-1}. Table~\ref{tab:aeubs-from-linear-inverse} summarizes the decay thresholds determined in this work.

\begin{table}[t!]
  \renewcommand{\arraystretch}{1.3}
  \centering
  \caption{\textsl{Decay of intrinsic volumes.}  The decay thresholds
    developed in this work appear below.  The computations for the orthant and subspace appear in Section~\ref{sec:return-two-exampl}, while the computations for the feasible cone of the Schatten 1-nrom and spectral norm appear in Section~\ref{sec:using-gaussian-width}. Other computations appear where indicated.
  }
  \label{tab:aeubs-from-linear-inverse}
  \begin{tabular}{p{4.3cm}ccc}
    \toprule
    \multirow{2}{*}{\textbf{Cone ensemble}} & \multicolumn{2}{c}{\textbf{Level zero}}
    & \textbf{Level \(\psi\)  }
    \\ 
    \cmidrule(l{1em}r){2-4}
    & Upper threshold \(\theta_\star\) & Lower threshold \(\kappa_\star\) &    \(\theta_\star(\psi)\) 
    \\    
    \midrule
    Subspaces of dimension~\(\lceil \sigma d \rceil\) & 
    \(\sigma\)&
    \(\sigma\) & 
    \(\sigma\)
    \\
    Orthants \(\R^d_+\) & \(\frac{1}{2}\) 
    & \(\frac{1}{2}\) 
    & (See
    Prop.~\ref{prop:upper-bd-psi-orth})
    \\
  \begin{tabular}{@{}p{4.3cm}@{}}
    Feasible cones of \(\ell_1\) norm at \(\tau\)-sparse vectors
  \end{tabular}
  & %
    &
    \makebox[0pt][c]{    \begin{tabular}{@{}c@{}}
      (See Prop.~\ref{prop:l1-decay} in Appendix~\ref{sec:asympt-thresh-calc})
    \end{tabular}
  }
    \\

    \begin{tabular}{@{}p{4.3cm}@{}}
    Feasible cones of the Schatten 1\nobreakdash-norm  at  \(n\times n\) 
    rank  \(\lceil \rho n\rceil\) matrices
  \end{tabular}
    &
    \begin{tabular}{@{}c@{}}
      \(6\rho-3\rho^2\) \\[-4pt]
     (See also Rem.~\ref{rem:OH10}.)
    \end{tabular}
    & & 
    \\
    \begin{tabular}{@{}p{4.3cm}@{}}
      Feasible cones of the spectral norm at an orthogonal matrices 
    \end{tabular}
    & 
    \(\frac{3}{4}\)
    & & 
    \\
    \bottomrule %
\end{tabular}
\end{table}

\section{Success and failure}
\label{sec:find-weak-thresh}

This section synthesizes the material from
Sections~\ref{sec:generic-setup-main} and~\ref{sec:analys-via-integr}
to determine whether the \CDMq~\eqref{eq:genconv} succeeds, or fails,
with high probability.  Section~\ref{sec:asymptotic-regime} introduces
the concept of a demixing ensemble.  Our main results arrive in
Section~\ref{sec:main-results}, where we find that the success and
failure of the \CDMq~\eqref{eq:genconv} are characterized by decay
thresholds.  Section~\ref{sec:strong-phase-trans} extends our methods
to achieve uniform guarantees on the success of
method~\eqref{eq:genconv}.

\subsection{Ensembles of demixing problems}
\label{sec:asymptotic-regime}
A demixing ensemble is a collection of demixing problems
that is indexed by the ambient dimension of the observation.  We
explain this idea in the context of MCA, and we develop the abstract
definition in Section~\ref{sec:deconv-ensembl-gener}.

\subsubsection{Example: The MCA demixing ensemble}
\label{sec:mca-high-dimensions}

Recall from Section~\ref{sec:first-appl-deconv} that MCA seeks to
demix a superposition of two sparse vectors.  Let us fix sparsity
levels \(\tau_{\vct x}\) and \(\tau_{\vct y}\) in \([0,1]\).  For each
pair \((\tau_{\vct x},\tau_{\vct y})\), we construct an ensemble of
demixing problems with one problem per dimension.  For
each \(d \in \mathbb{N}\), let \(\vct x_0^{(d)}\) and \(\vct
y_0^{(d)}\) be vectors in \(\R^d\) with 
\begin{equation*}
\nnz(\vct x_0^{(d)}) = \lceil \tau_{\vct x} d\rceil \quad \text{and}
\quad \nnz(\vct y_0^{(d)})\ = \lceil \tau_{\vct y}d\rceil.  
\end{equation*}
In other words, the sparsity of each vector is proportional to the
ambient dimension.  Draw a random basis \(\mtx Q^{(d)}\) from
\(\mathsf{O}_d\).  We observe the vector \(\vct z_0^{(d)} = \vct
x_0^{(d)} + \mtx Q^{(d)} \vct y_0^{(d)}\). To set up a convex
demixing method, we need to introduce appropriate complexity
measures for \(\vct x_0^{(d)}\) and \(\vct y_0^{(d)}\).  For both
vectors, the \(\ell_1\) norm on \(\R^d\) is the natural choice for inducing sparsity. Assume
we have access to the side information \(\alpha^{(d)} =
\lone{\smash{\vct y_0^{(d)}}}\).

Together, these data define a demixing ensemble for MCA. We want
to study when the MCA problem~\eqref{eq:genconv} succeeds with high
probability for all members of the ensemble with \(d\) sufficiently
large.

\subsubsection{Abstract demixing ensembles}
\label{sec:deconv-ensembl-gener}

It is straightforward to extend this idea to other demixing
problems.  Let \(\Index \subset \mathbb{N}\) be an infinite set of
indices. A \emph{demixing ensemble} consists of one problem per
index.  For each \(d \in \Index\), the data are
\begin{itemize}\renewcommand{\itemsep}{0pt}
\item Vectors \(\vct x_0^{(d)},\vct y_0^{(d)} \in \R^d\),
\item A random basis \(\mtx Q^{(d)}\in \mathsf{O}_d\) that is statistically
  independent of the other ensemble data,
\item The observation
\begin{math}
  \vct z_0^{(d)} = \vct x_0^{(d)} + \mtx Q^{(d)} \vct y_0^{(d)} \in \R^d,
\end{math}
\item Complexity measures \(f^{(d)}\) and \(g^{(d)}\) defined on
  \(\R^d\), and
\item The side information \(\alpha^{(d)} = g^{(d)}(\vct y_0^{(d)})\).
\end{itemize}
 Given such a demixing ensemble, we seek to determine
conditions for which the \CDMq
\begin{equation}\label{eq:genconv-by-dimension}
  \minprog{}{f^{(d)}(\vct x)}{ g^{(d)}(\vct y)
    \le \alpha^{(d)}\;\;\text{and}\;\; \vct x + \mtx Q^{(d)} \vct y
    = \vct z_0^{(d)},} 
\end{equation}
succeeds with high probability when the dimension \(d\) is large. When it does not cause confusion, we  omit the superscript \(d\).

Our goal in this paper is to describe regions where the demixing program~\eqref{eq:genconv-by-dimension} succeeds, or fails, with high probability as the dimension \(d\to \infty\).  In order to avoid cumbersome repetition in our theorems, we make a shorthand definition.
\begin{definition}[Overwhelming probability in high dimensions]
  Given a demixing ensemble as in
  Section~\ref{sec:deconv-ensembl-gener}, we say
  that~\eqref{eq:genconv-by-dimension} \emph{succeeds with
    overwhelming probability in high dimensions} if there exists an
  \(\varepsilon > 0\) such that, for every sufficiently large
  dimension \(d\in \Index\), program~\eqref{eq:genconv-by-dimension}
  succeeds with probability at least \(1-\mathrm{e}^{-\varepsilon d}\)
  over the randomness in \(\mtx Q^{(d)}\).  Similarly, we say
  that~\eqref{eq:genconv-by-dimension} \emph{fails with overwhelming
    probability in high dimensions} if there exists an \(\varepsilon
  >0\) such that, for every sufficiently large dimension \(d \in
  \Index\), program~\eqref{eq:genconv-by-dimension} succeeds with
  probability at most \(\mathrm{e}^{-\varepsilon d}\).
\end{definition}
Following common practice in asymptotic analysis, the definition above masks the dependence between the probability decay rate \(\eps\) and the ``sufficiently large'' dimension \(d_0\).  This approach represents a tradeoff. We will find that it provides demixing guarantees in terms of only decay thresholds, but it does not provide explicit probabilistic bounds for finite \(d\).  Nevertheless, our asymptotic analysis proves quite accurate at predicting the behavior in our experiments.  See Remark~\ref{rem:explicit-prob-bds} for further discussion.

\subsection{The main results}
\label{sec:main-results}

We are now in a position to state our main results. The first result
shows that the upper decay threshold provides guarantees for the
success of demixing under the model of
Section~\ref{sec:asymptotic-regime}.
\begin{theorem}[Success of demixing]
  \label{thm:sharp-thresh}
  Consider a demixing ensemble as in
  Section~\ref{sec:deconv-ensembl-gener}. Suppose the ensembles
  \(\bigl\{\Fcone(f^{(d)},\vct x_0^{(d)})\mid d \in \Index\bigr\}\) and
  \(\bigl\{\Fcone(g^{(d)},\vct y_0^{(d)})\mid d\in \Index\bigr\}\) of feasible cones
  have upper decay thresholds \(\theta_{\vct x}\) and \(\theta_{\vct
    y}\). If
  \[\theta_{\vct x} + \theta_{\vct y} < 1,\] 
  then program~\eqref{eq:genconv-by-dimension} succeeds with
  overwhelming probability in high dimensions.
\end{theorem}
The proof of this result appears in
Section~\ref{sec:proof-theo-sharp}.  This next result shows that the
failure of demixing is characterized by the lower decay
threshold.
\begin{theorem}[Failure of demixing]
  \label{thm:failure-thresh}
  Consider a demixing ensemble as in
  Section~\ref{sec:deconv-ensembl-gener}. Suppose the ensembles
  \(\bigl\{\Fcone(f^{(d)},\vct x_0^{(d)}) \mid d \in \Index\bigr\}\) and
  \(\bigl\{\Fcone(g^{(d)},\vct y_0^{(d)})\mid d \in \Index\bigr\}\) of
  feasible cones have lower decay thresholds \(\kappa_{\vct x}\) and
  \(\kappa_{\vct y}\).  If
  \[\kappa_{\vct x} + \kappa_{\vct y} >  1,\]
  then program~\eqref{eq:genconv-by-dimension} fails with overwhelming
  probability in high dimensions.
\end{theorem}
The proof of this second result is similar in spirit to the proof of
Theorem~\ref{thm:sharp-thresh}, but it requires additional technical
finesse; we defer the details to Appendix~\ref{sec:region-fail-proof}.

\subsubsection{Consequences for the choice of complexity functions}
\label{sec:cons-design-deconv}
The complementary nature of 
Theorems~\ref{thm:sharp-thresh} and~\ref{thm:failure-thresh} has striking
implications. The success, or failure, of the \CDMq~\eqref{eq:genconv-by-dimension}
in high dimensions depends only on the sum of the decay thresholds.
As a consequence, the upper and lower decay thresholds assess the
quality of the complexity measures \(f^{(d)}\) and \(g^{(d)}\) in high
dimensions.

Since the decay thresholds are independent of any interrelationship
between the structured vectors \(\vct x_0^{(d)}\) and \(\vct y_0^{(d)}\) in the
superimposed observation \(\vct z_0^{(d)}\), the quality of the complexity
measure \(f^{(d)}\) is independent of the choice \(g^{(d)}\).  This explains, for
instance, the ubiquity of the use of the \(\ell_1\) norm for inducing
sparsity and the Schatten 1-norm as a complexity measure for
rank. Simply put, when a complexity measure is good for one
incoherent demixing problem, it is good for another.

\subsubsection{Sharp phase transitions}
\label{sec:thresh-phen-conj}

Theorems~\ref{thm:sharp-thresh}
and~\ref{thm:failure-thresh} have an important consequence for phase
transitions in the \CDMq~\eqref{eq:genconv-by-dimension}.  Equality
between the lower and upper decay thresholds  holds in
cases where we have access to exact formulas for the spherical
intrinsic volumes.  Proposition~\ref{prop:subspace-aeub} shows
that the upper and lower decay thresholds are equal for subspaces
whose dimension is proportional to the ambient space, and
Proposition~\ref{prop:orth-thresh} implies that the upper decay
threshold is equal to the lower decay threshold for the ensemble of
nonnegative orthants.  Moreover, our computations in
Appendix~\ref{sec:asympt-thresh-calc} suggest that equality between
the upper and lower decay thresholds also holds for the ensemble of
feasible cones of the \(\ell_1\) norm at vectors with a fixed
proportion of nonzero elements.

The equality of the upper and lower decay thresholds explains the
close agreement between our theoretical bounds and the empirical
experiments. For instance, consider
Figure~\ref{fig:l1-l1-demixing}. For sparsity levels below the
green curve, the sum of the upper decay thresholds is less than one,
so Theorem~\ref{thm:sharp-thresh} implies that demixing succeeds
with overwhelming probability in high dimensions.  On the other hand,
with sparsity levels above the green curve, the sum of the lower decay
thresholds exceeds one, so demixing fails with overwhelming
probability in high dimensions by Theorem~\ref{thm:failure-thresh}.
(See Section~\ref{sec:deconv-sparse-vect} for the details of this
calculation.)

One may wonder whether the transition between success and failure is
sharp in general. In work undertaken after the submission of this article, with collaborators, we have determined that the answer to this question is \emph{yes} for a large class of demixing ensembles.  In essence, the result~\cite[Thm.~6.1]{AmeLotMcC:13} indicates that the upper- and lower-decay thresholds for a sufficiently regular ensemble \(\{K^{(d)}\mid d \in \mathcal{D}\}\) are equal.
For example, the fact that  \(d^{-1}W(K^{(d)}\cap \mathsf{S}^{d-1})^2 \to \rho \in (0,1)\) as \(d\to \infty\) is sufficient to guarantee that the upper- and lower-decay thresholds \(\theta_\star\) and \(\kappa_\star\) of the ensemble \(\{K^{(d)}\mid {d\in \mathcal{D}}\}\) satisfy \(\theta_\star = \kappa_\star=\rho\).  (The function \(W\) is the Gaussian width defined in Section~\ref{sec:using-gaussian-width}.) We refer the reader to this newer work for details.

\subsubsection{Proof of Theorem~\ref{thm:sharp-thresh}}
\label{sec:proof-theo-sharp}

The proof of Theorem~\ref{thm:sharp-thresh} follows readily from a
geometric statement concerning the probability that a random cone
strikes a fixed cone as the dimension \(d\) becomes large.
\begin{theorem}\label{thm:asympt-thresh-gen}
  Suppose \(\Index\) is an infinite set of indices, and let \(\{K^{(d)}\subset
  \R^d\mid d \in \Index\}\) and \(\{\tilde K^{(d)} \subset \R^d\mid d \in \Index\}\)
  be two ensembles of closed convex cones with upper decay thresholds~
  \(\theta_\star\) and ~\(\tilde \theta_\star\).  If~ \(\theta_\star +
  \tilde \theta_\star < 1\), then there exists an \(
  \varepsilon >0\) such that, for all sufficiently large \(d\), we have \( \mathbb{P}\bigl\{ K^{(d)} \cap \mtx Q
  \tilde K^{(d)} \ne \{\vct 0\}\bigr\} \le \mathrm{e}^{- \varepsilon
    d}\).
\end{theorem}

Before proceeding to the proof, let us explain how
Theorem~\ref{thm:sharp-thresh} follows from
Theorem~\ref{thm:asympt-thresh-gen} and the geometric optimality
conditions of Lemma~\ref{lem:triv-intersect}.

\begin{proof}[Proof of Theorem~\ref{thm:sharp-thresh} from
  Theorem~\ref{thm:asympt-thresh-gen}]
  By the assumptions in Theorem~\ref{thm:sharp-thresh}, the
  ensembles \(\bigl\{{\overline{\Fcone}(f^{(d)}, \vct x_0^{(d)})}\mid
  d\in \Index\bigr\}\) and \(\bigl\{-\overline{ \Fcone}(g^{(d)},\vct
  y_0^{(d)})\mid d \in \Index\bigr\}\) of closed cones satisfy the
  hypothesis of Theorem~\ref{thm:asympt-thresh-gen}.  Thus, there is
  an \(\eps >0\) for which
  \begin{equation*}
    \overline{\Fcone}(f^{(d)}, \vct
    x_0^{(d)}) \bigcap\bigl( - \mtx Q \overline{\Fcone}(g^{(d)}, \vct
    y_0^{(d)})\bigr)=\{\vct 0\}
  \end{equation*}
  except with probability \(\mathrm{e}^{- \varepsilon d}\), for all
  sufficiently large \(d\).

  Since cones are contained in their closure, the two feasible cones
  have a trivial intersection at least as frequently as their closures
  (but see the remark below). Applying our geometric optimality
  condition, Lemma~\ref{lem:triv-intersect}, immediately implies
  that~\eqref{eq:genconv-by-dimension} succeeds with probability at
  least \(1-\mathrm{e}^{-\varepsilon d}\) for every sufficiently
  large \(d\).
\end{proof}

\begin{remark}\label{rem:touch-prob}
  In fact, the probability that randomly oriented convex cones
  strike is \emph{equal} to the probability that their closures strike.  This
  seemingly innocuous claim appears to have no simple proof from first
  principles.  However, this fact readily follows from the discussion
  of touching probabilities in~\cite[pp.~258--259]{Schneider2008}.
\end{remark}

Figure~\ref{fig:computation-diag} illustrates the main idea behind
the following proof.
\begin{figure}[t!]
  \centering
  \includegraphics[width=0.49\columnwidth]{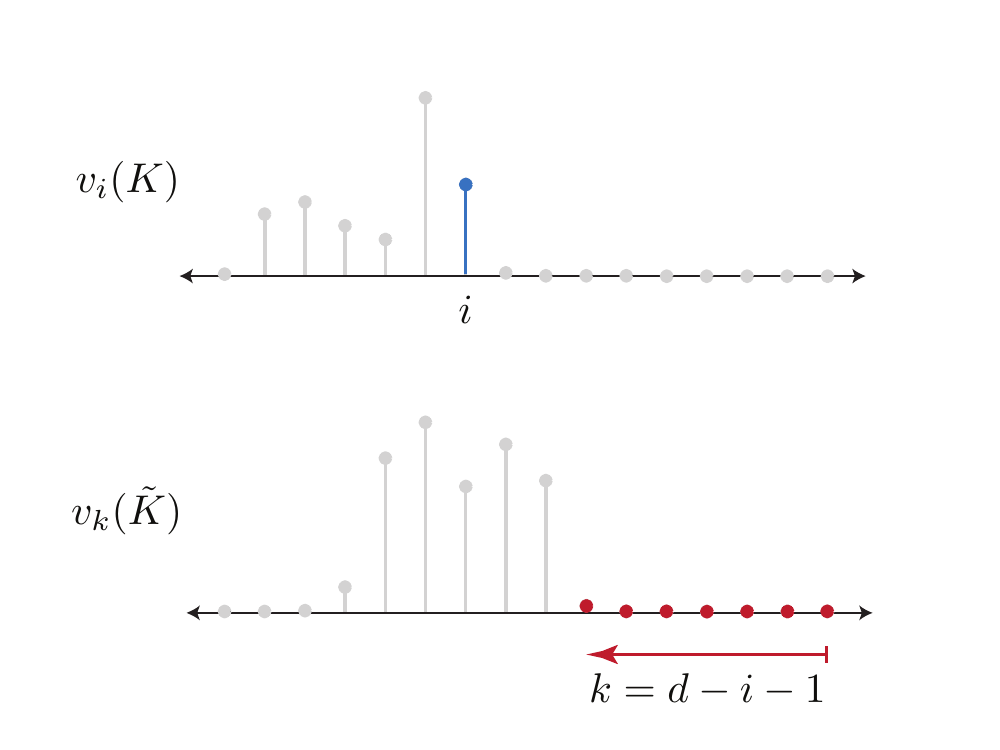}
  \includegraphics[width=0.49\columnwidth]{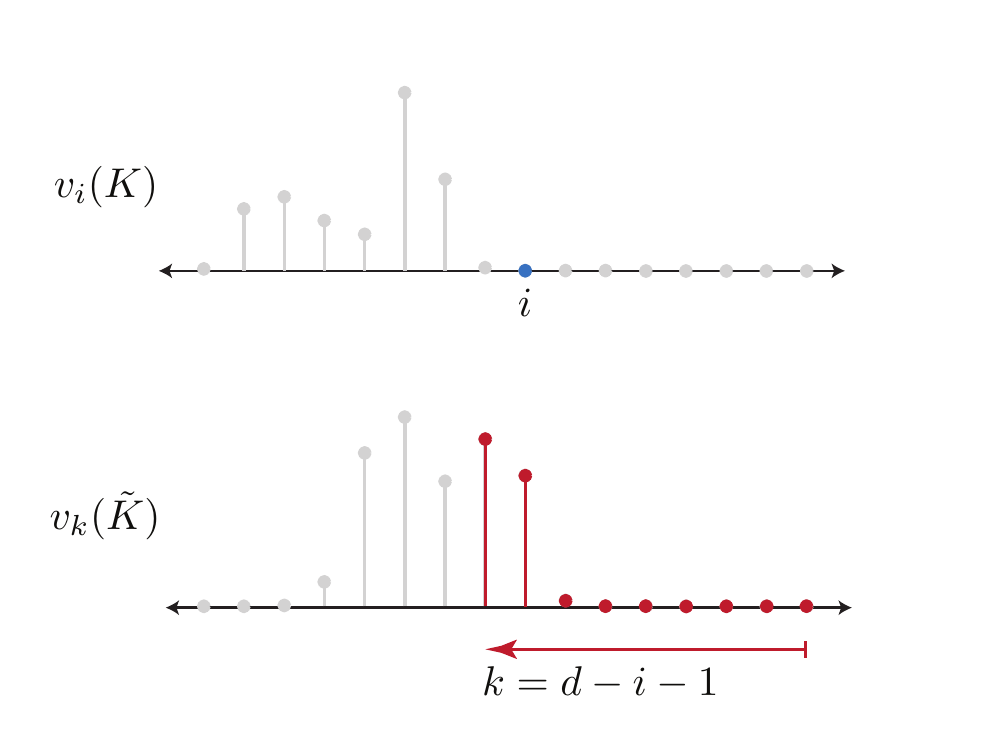}
  \caption{\textsl{Main idea behind the proof of
      Theorem~\ref{thm:asympt-thresh-gen}.}  When the spherical
    intrinsic volumes are very small for large \(i\), the product of the
    outer (top) and inner (bottom) terms in
    equation~\eqref{eq:bottom-sum-proof-of-main-claim1} is always
    small.  In the left panel, the large value of \(v_i(K)\) is offset
    by the small inner sum \(\sum_{k=d-i-1}v_k(\tilde K)\).  On the
    right, the small value of \(v_i(K)\) counteracts the large sum
    \(\sum_{k=d-i-1}v_k(\tilde K)\).  This situation always
    occurs when the upper decay thresholds satisfy \(\theta_\star+
    \tilde \theta_\star<1\) and the ambient dimension \(d\) is large.
  } \label{fig:computation-diag}
\end{figure}

\begin{proof}[Proof of Theorem~\ref{thm:asympt-thresh-gen}]
  The spherical kinematic formula~\eqref{eq:sphere-kin} only applies
  when at least one cone is not a subspace, so we split the
  demonstration into two cases.  First, we assume that at least one of
  the ensembles \(\{K^{(d)}\}\) or \(\{\tilde K^{(d)}\}\) of cones
  does not contain any subspace. Then, we consider the case where both
  ensembles of cones contain only subspaces.  The argument readily
  extends to the general case by considering subsequences where one of
  the two cases above holds.  Throughout the proof, we drop the explicit
  dependence of \(K^{(d)}\) and \(\tilde K^{(d)}\) on the dimension
  \(d\) for notational clarity.

  We start with the first case: assume that either \(K\) or \(\tilde
  K\) is not a subspace.  Our main tool is the spherical kinematic
  formula, Fact~\ref{fact:sphere-kin}.  The positivity of the
  spherical intrinsic volumes and the bound \((1+(-1)^k)\le
  2\) imply that  the probability \(P\) of interest satisfies
 \begin{align}
    P \defeq \mathbb{P}\bigl\{ K \cap \mtx Q \tilde K \ne
    \{\vct 0\}\bigr\} &\le  2 \sum_{k=0}^{d-1} \sum_{i=k}^{d-1} v_i(K)\cdot
    v_{d-1-i+k}(\tilde K)\notag \\ &=  2 \sum_{i=0}^{d-1} v_i(K) \sum_{k =
      d-i-1}^{d-1}  v_{k}(\tilde K).
    \label{eq:bottom-sum-proof-of-main-claim1}
  \end{align}
  The equality follows by a change in the order of summation and a
  change of the summation index.  Since \(\theta_\star+\tilde
  \theta_\star<1\), there exist parameters \(\theta>\theta_\star\) and
  \(\tilde \theta > \tilde \theta_\star\) for which \( \theta + \tilde
  \theta < 1\).  We expand the right-hand sum
  of~\eqref{eq:bottom-sum-proof-of-main-claim1} into four terms,
  say \(\Sigma_1\), \(\Sigma_2\), \(\Sigma_3\), and \(\Sigma_4\):
  \begin{multline}
    \frac{1}{2} P \le  \underbrace{\sum_{i=0}^{\lceil \theta d \rceil} v_i(K)
      \sum_{k = d- i-1}^{\lceil \tilde \theta d\rceil } v_k(\tilde K)
    }_{\eqdef \Sigma_1}
    +
    \underbrace{\sum_{i = \lceil \theta d \rceil + 1}^{d-1} v_i(K) \sum_{k = d -
      i-1}^{\lceil \tilde \theta d\rceil} v_k(\tilde K)}_{\eqdef
    \Sigma_2} 
    \\  + \underbrace{\sum_{i=0}^{\lceil \theta d \rceil} v_i(K) \sum_{k =
       \lceil \tilde \theta d \rceil + 1}^{d-1} v_k(\tilde K)}_{\eqdef \Sigma_3}
   +\underbrace{\sum_{i = \lceil \theta d \rceil + 1}^{d-1} v_i(K) \sum_{k =
       \lceil \tilde \theta d \rceil + 1}^{d-1} v_k(\tilde
     K)}_{\eqdef \Sigma_4}. 
   \label{eq:sigmas}
 \end{multline}
 We bound each summand separately.  First, the fact \(\theta +
 \tilde \theta < 1\) implies that, for all sufficiently large \(d\), we
 have \(\lceil \theta d \rceil + \lceil \tilde \theta d \rceil < d-1\).  We
 conclude that \(\Sigma_1=0\) when \(d\) is large enough: The inner
 sum is empty.

 To bound \(\Sigma_2\), apply
 Facts~\ref{fact:obvious}.\ref{item:obv-positivity}
 and~\ref{fact:obvious}.\ref{item:obv-unity} to the inner sum to find
 \begin{equation*}
   \Sigma_2 \le \sum_{i = \lceil \theta d \rceil + 1}^{d-1} v_i(K)
   \le (d-1) \mathrm{e}^{-\varepsilon' d},
 \end{equation*}
 where the second inequality holds for some \(\varepsilon'>0\) and all
 sufficiently large \(d\) owing to the definition of the upper decay
 threshold \(\theta_\star\).  Through analogous reasoning, the definition
 of \(\tilde \theta_\star\) gives the exponential bounds
 \begin{equation*}
   \Sigma_3 \le (d-1)\mathrm{e}^{-\varepsilon'' d}\quad
 \text{and} \quad \Sigma_4 \le (d-1) \mathrm{e}^{-\varepsilon''' d},
 \end{equation*}
 for some \(\varepsilon'',\varepsilon'''>0\) and all sufficiently
 large \(d\).  Taking \( \varepsilon \) sufficiently small (say
 \(\varepsilon = \frac{1}{2}
 \min\{\varepsilon',\varepsilon'',\varepsilon'''\}\)) and \(d\)
 sufficiently large gives the result for the first case. 
 
 For the second case, suppose that both \(K\) and \(\tilde K\) are
  subspaces.  Set \(n \defeq \dim(K)\) and \(\tilde n \defeq
 \dim(\tilde K)\).  Choose parameters \(\theta>\theta_\star\) and
 \(\tilde \theta >\tilde \theta_\star\) such that \(\theta+\tilde
 \theta < 1\). The definition of the upper decay threshold and
 Proposition~\ref{prop:subspace-iv} imply that
 \begin{equation*}
 v_{i}(K) = \delta_{i,n-1}\le \mathrm{e}^{-\eps' d}
\end{equation*}
for some \(\eps'>0\), all sufficiently large \(d\), and every
\(i\ge\lceil \theta d\rceil\).  In particular, this inequality implies
\(n-1 < \lceil \theta d \rceil\) for all \(d\) large enough.
Similarly, we find \(\tilde n-1 < \lceil \tilde \theta d \rceil \) for
all \(d\) sufficiently large.

Since \(\theta + \tilde \theta < 1\), we have
\begin{equation*}
n + \tilde n
\le \lceil \theta d \rceil +\lceil \tilde \theta d\rceil+2 < d
\end{equation*}
whenever \(d\) is large enough.  That is, the sum of the dimensions of
the subspaces is less than the ambient dimension.  Since two randomly
oriented subspaces are almost surely in general position, we see that
\(K\cap \mtx Q \tilde K = \{\zerovct\}\) with probability one, whenever
\(d\) is large enough.  This completes the second case, so we are
done.
\end{proof}
 \begin{remark}[Explicit dimensional dependence] \label{rem:explicit-prob-bds}
The methods above can provide explicit  dependence between the sum of the decay thresholds \(\theta_x +\theta_y\), the ``sufficiently large'' dimension and the  probability decay rate \(\eps\) when  detailed information about the intrinsic volumes is available.   In the case of the orthant, for example, Stirling's formula with remainder~\cite[Sec.~5.6.1]{OlvLozBoi:10} may  provide enough information.  For the descent cones of the \(\ell_1\) norm, it may be possible to achieve such explicit dependence using the approach of~\cite{DonTan:10}.  %
 \end{remark}
\subsection{Uniform demixing guarantees}
\label{sec:strong-phase-trans}

Suppose that \(\mtx Q\) is drawn at random and fixed. For example, in the MCA problem from Section~\ref{sec:first-appl-deconv}, we may observe a number of different images, but the structures that we expect to find (encoded by \(\mtx Q\)) are the same in each image. In this case, we ask whether~\eqref{eq:l1-const-1} demix \emph{every}  sparse pair \((\vct x_0,\vct y_0)\) from the associated observation \(\vct z_0 = \vct x_0 +\mtx Q \vct y_0\), and thus successfully identify the structures in a large family of images. 

More generally, we can study the probability that the generic demixing program~\eqref{eq:genconv} will demix every structured pair \((\vct x_0,\vct y_0)\) with a random---but fixed---basis \(\mtx Q\).  (These uniform guarantees go by the name of \emph{strong bounds}~\cite{Donoho2005a}, \cite{Donoho2006}, \cite{Donoho2009}, \cite{Donoho2010a}.)  By the geometric optimality condition of Lemma~\ref{lem:triv-intersect}, this probability is equal to the probability of the event
\begin{equation}\label{eq:strong-condition}
  \Fcone(f,\vct x_0) \bigcap \bigl(-\mtx Q\Fcone(g,\vct y_0)\bigr) =
  \{\zerovct\}
  \text{ for \emph{every} structured pair \((\vct x_0,\vct y_0)\).}
\end{equation}

In this section, we control the probability of~\eqref{eq:strong-condition} by coupling  the argument leading to Theorem~\ref{thm:sharp-thresh} with a union bound.   This approach does not necessarily limit our methods to a finite number of structured pairs \((\vct x_0,\vct y_0)\), however.  In the case of a sparse vectors, for example,  the set of feasible cones
\begin{equation}\label{eq:feas-cone-l1-strong}
  \bigl\{ \Fcone(\lone{\cdot},\vct x) \mid \vct x\in \R^d, \;\nnz(\vct x) = k\bigr\}
\end{equation}
consists of  \(\binom{d}{k}2^k\) cones because the feasible cone  \(\Fcone(\lone{\cdot},\vct x)\) depends only on the sparsity and sign pattern, and not the magnitude, of the elements of \(\vct x\).  (See Section~\ref{sec:l1-l1-strong-transitions} for more details.)  Thus, even the simple union bound can offer insight into the behavior of constrained MCA~\eqref{eq:l1-const-1} for an \emph{infinite} family of images.

Applying a union bound to a finite, but rather large, set of cones such as~\eqref{eq:feas-cone-l1-strong} requires stronger probabilistic information than provided by the decay thresholds.   Therefore, we  define an extension of the upper decay threshold that provides detailed information on the rate of decay of spherical
intrinsic volumes. 
\begin{definition}[Decay at level \(\psi\)]\label{def:decay-at-psi}
  Let \(\Index \subset \mathbb{N}\) be an infinite set of indices, 
  let \(\{K^{(d)}\mid d \in \Index\}\) be an ensemble of closed convex
  cones with \(K^{(d)}\in \R^d\), and suppose \(\psi \ge 0\).  We say that
  \(\theta_\star\) is an upper decay threshold \emph{at level
    \(\psi\)} for the ensemble \(\{K^{(d)}\}\) if, for every \(\theta
  > \theta_\star\), there exists an \(\varepsilon > 0\) such
  that, for all sufficiently large \(d\), the inequality
  \begin{equation*}\label{eq:exp-upper-bd-level-psi}
    v_i(K^{(d)}) \le \mathrm{e}^{-d(\psi + \varepsilon) }
  \end{equation*}
  holds for all \(i \ge \lceil \theta d\rceil\).  When no level is
  specified, we take \(\psi = 0\) for compatibility with
  Definition~\ref{def:aeub}.  As in Definition~\ref{def:aeub}, this
  definition extends to non-closed cones by taking the closure.
\end{definition}
Subspaces and orthants again provide  useful examples.  
\begin{proposition}[Decay at level \(\psi\) for an ensemble of subspaces]
  Let \(\{L^{(d)}\}\) be an infinite ensemble of linear subspaces with
  \(L^{(d)}\subset \R^d\) and
  \(\dim(L^{(d)}) = \lceil \sigma d\rceil\), and suppose \(\psi \ge
  0\).  Then \(\theta_{\star}=\sigma\) is an upper decay threshold for
  \(\{L^{(d)}\}\) at level~\(\psi\).
\end{proposition}
\begin{proof} The proof is substantially similar to the proof of
  Proposition~\ref{prop:subspace-aeub}. Let \(n = \dim(L^{(d)})\), so
  that \(v_i(L^{(d)}) = \delta_{i,n-1}\) by
  Proposition~\ref{prop:subspace-iv}.  For any \(\theta > \sigma\), we
  have \(\lceil \theta d \rceil > \lceil \sigma d \rceil\) for all
  large enough \(d\).  Therefore,
  \begin{equation*}
    v_i(L^{(d)}) = 0 < \econst^{-(\psi+1) d}
  \end{equation*}
  for all \(i\ge \lceil \theta d \rceil \) and \(d\) sufficiently
  large.  By definition, \(\theta_\star=\sigma\) is an upper decay
  threshold at level \(\psi\) for \(\{L^{(d)}\}\).
\end{proof}

\begin{figure}[t!]
  \centering
  \includegraphics[width=\figwidth\columnwidth]{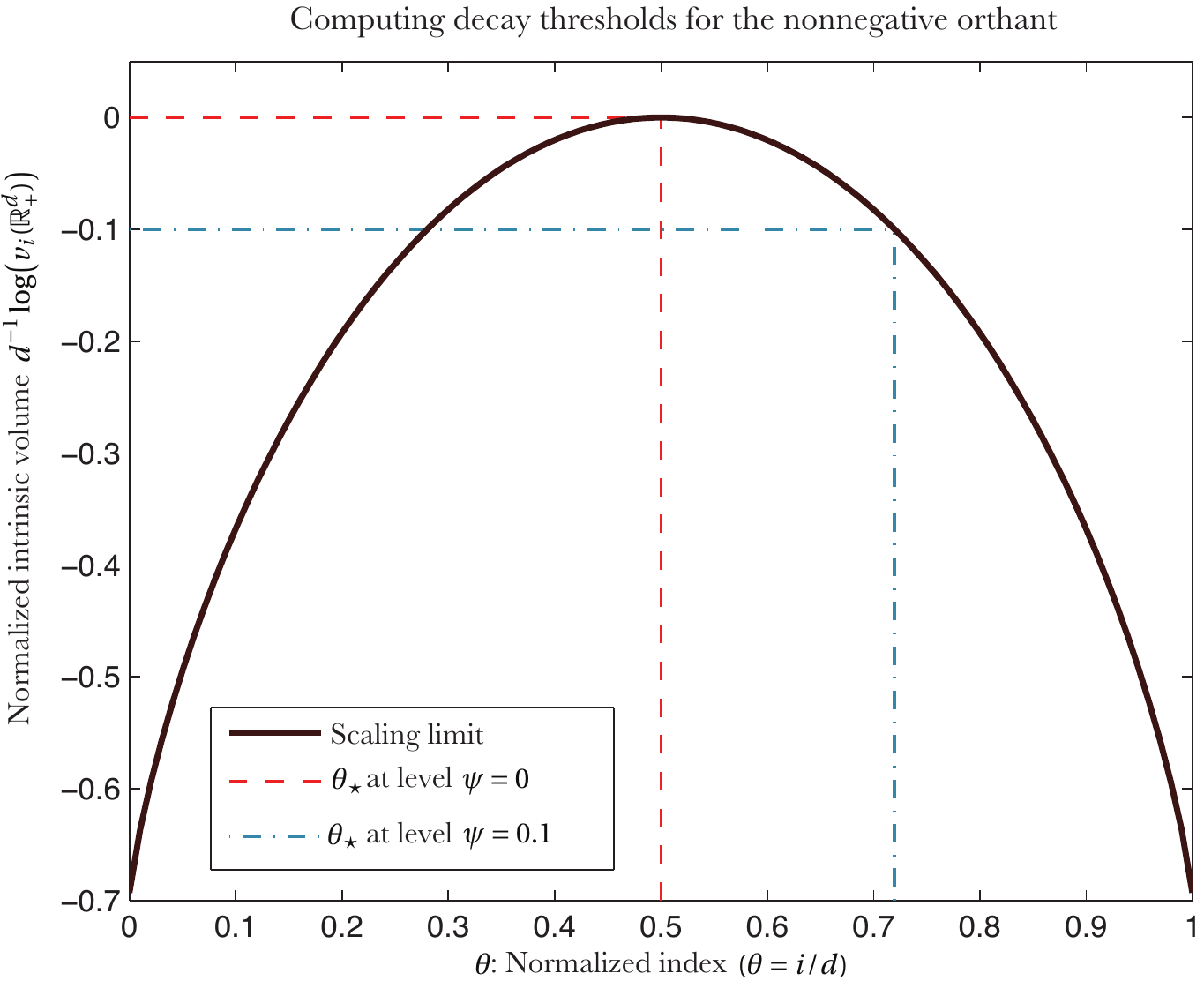}
  \caption{\textsl{Computation of decay thresholds for the orthant.}
    The solid curve is the limiting behavior of the spherical
    intrinsic volumes of the orthant \(\R^d_+\) from 
    Figure~\ref{fig:orth-vols-norm}.  The rightmost point where the
    horizontal line at level \(-\psi\) crosses this curve is an upper
    decay threshold at level \(\psi\).  From the diagram, we see that 
    \(\theta_\star=\frac{1}{2}\) is an upper decay threshold at level
    zero, while \(\theta_\star\approx 0.72\) is an upper decay
    threshold at level \(\psi= 0.1\).}\label{fig:orth-exponent}
\end{figure}

The computation for the orthant is only slightly more involved.  The
following result is illustrated in Figure~\ref{fig:orth-exponent}.
\begin{proposition}[Decay at level \(\psi\) for \(\{\R^d_+\}\)] 
  \label{prop:upper-bd-psi-orth}
  Suppose \(0\le \psi \le \log(2)\), and let \(\Index \subset
  \mathbb{N}\) be an infinite set of indices.  Define
  \begin{equation}\label{eq:upper-bd-psi-orth}
    \smash{\theta_{\R^d_+}(\psi)} \defeq 
    \sup{} \bigl\{\theta\mid H(\theta) \ge \log(2) -\psi \bigr\},
  \end{equation}
  where \(H(\theta)\) is the entropy defined in~\eqref{eq:entropy}.
  Then \(\smash{\theta_{\R^d_+}(\psi)}\) is an upper decay threshold at level
  \(\psi\) for the ensemble \(\{\R^d_+\}\) of orthants.
\end{proposition}
\begin{proof}[Proof sketch]
  Let \(\theta > \smash{\theta_{\R_+^d}(\psi)}\).  As in the proof of
  Proposition~\ref{prop:orth-thresh}, there exists an \(\varepsilon
  >0\), such that, for all \(d\) sufficiently large, we have
  \begin{equation*}
    v_{i}(\R^d_+) \le \exp\Bigl(
    d\bigl(H(\theta_\psi) - \log(2) + \varepsilon\bigr)\Bigr) =
    \exp\bigl(-d(\psi + \eps)\bigr)
  \end{equation*}
  for every \(i \ge \lceil \theta d\rceil\). By definition,
  \(\smash{\theta_{\R^d_+}}(\psi)\) is an upper decay threshold at level
  \(\psi\) for the ensemble \(\{\R^d_+\}\).  
\end{proof}

We now extend the decay threshold to cover the family of cones
considered in condition~\eqref{eq:strong-condition}.  In our
applications, the feasible cones appearing
in~\eqref{eq:strong-condition} are  congruent, so we restrict
our attention to this case.
\begin{definition}[Decay threshold for an ensemble of sets of cones]
  For an infinite set of indices \(\Index \subset \mathbb{N}\), let
  \(\{\mathcal{K}^{(d)}\mid d \in \Index\}\) be an ensemble of
  \emph{sets} of  congruent cones, indexed by the ambient
  dimension, and let \(\{K^{(d)}\}\) be an ensemble of exemplars, that
  is, \(K^{(d)} \in \mathcal{K}^{(d)}\) for every \(d\in \Index\). We say that
  \(\{\mathcal{K}^{(d)}\}\) has upper decay threshold \(\theta\) at
  level \(\psi\) if the sequence of exemplars \(\{K^{(d)}\}\) has
  upper decay threshold \(\theta\) at level \(\psi\).
\end{definition}
By Fact~\ref{fact:obvious}.\ref{item:obv-invariant}, the decay
threshold for an ensemble \(\{\mathcal{K}^{(d)}\}\) of sets of
 congruent cones is independent of the choice of exemplars,
so this nomenclature is well defined.  We now state an analog to
Theorem~\ref{thm:asympt-thresh-gen} that bounds the probability that a
number of cones strike, provided there are not too many cones.
\begin{theorem}\label{thm:strong-thm}
  Let \(\{\mathcal{K}^{(d)} \}\) and \(\{
  \tilde{\mathcal{ K}}^{(d)}\}\) be two ensembles of
 {sets} of  congruent closed convex cones, indexed by
  ambient dimension \(d\).  Suppose the cardinality of
  \(\mathcal{K}^{(d)} \) and \(\tilde{\mathcal{K}}^{(d)}\) grows no
  faster than exponentially: there exist \(\psi,\tilde \psi\) such
  that, for every \(\eta >0\) and all \(d\) sufficiently large, we
  have the inequalities \(|\mathcal{K}^{(d)}| \le \mathrm{e}^{d (\psi
    + \eta)}\), and \(|\tilde{\mathcal{K}}^{(d)}| \le \mathrm{e}^{d
    (\tilde \psi +\eta)}\).  Suppose further \(\{\mathcal{K}^{(d)}
  \}\) and \(\{ \tilde{\mathcal{
      K}}^{(d)}\}\) have respective upper decay
  thresholds \(\theta_\star\) and \(\tilde \theta_\star\), each at
  level \(\psi+ \tilde \psi\).
  
  If~ \(\theta_\star + \tilde \theta_\star < 1\), then there exists an
  \(\varepsilon>0\) such that for every sufficiently large \(d\),
    \begin{equation*}
      \mathbb{P}\left\{  K \cap \mtx Q \tilde
        K \ne \{\vct 0\} \,\text{ for any }\,  K \in \mathcal{K}^{(d)},\;
        \tilde K \in \tilde{\mathcal{K}}^{(d)}\right\} \le
      \mathrm{e}^{-\varepsilon d},
  \end{equation*}
  where the probability is taken over the random basis \(\mtx Q \in \mathsf{O}_d\). 
\end{theorem}
The proof of Theorem~\ref{thm:strong-thm} simply couples a union bound
to the proof of Theorem~\ref{thm:asympt-thresh-gen}, so we defer the
demonstration to Appendix~\ref{sec:proof-strong-theorem}.  Note that
the statement of Theorem~\ref{thm:strong-thm} is equivalent to that of
Theorem~\ref{thm:asympt-thresh-gen} when \(\mathcal{K}^{(d)}\) and
\(\tilde{\mathcal{K}}^{(d)}\) are singletons, because we may take
\(\psi = \tilde \psi = 0\) in this case.

Theorem~\ref{thm:strong-thm} can be used to verify that
event~\eqref{eq:strong-condition} holds with high probability when the
dimension \(d\) becomes large.  In
Section~\ref{sec:l1-l1-strong-transitions}, we use this approach to
verify that the MCA formulation~\eqref{eq:l1-const-1} can demix
{all} sufficiently sparse vectors, and in
Section~\ref{sec:linf-strong-thresholds}, we use
Theorem~\ref{thm:strong-thm} to show that the channel coding
method~\eqref{eq:channel-code-1} is robust to adversarial sparse
corruptions.

\section{Computing decay thresholds}
\label{sec:line-inverse-probl-1}

Section~\ref{sec:main-results} demonstrates that the decay thresholds
provide a simple way to analyze demixing under the random basis
model.  This section describes several methods for computing decay
thresholds. We begin by considering direct approaches, where precise
formulas for the spherical intrinsic volumes give correspondingly
precise thresholds.

The direct method is powerful, but its application is limited to
regimes where we have access to formulas for spherical intrinsic
volumes.  In Section~\ref{sec:relat-line-inverse}, we observe that
known results on linear inverse problems~\cite{ChaRecPar:12}
imply bounds on decay thresholds.  This observation allows us to study
upper decay thresholds for several structural classes,  including
low-rank matrices.

\subsection{Direct approach}
\label{sec:direct-approach}

There are several situations where the direct approach for calculating
decay thresholds is feasible.  Propositions~\ref{prop:subspace-aeub}
and~\ref{prop:orth-thresh} compute upper and lower decay thresholds
for ensembles of subspaces and orthants directly from the definition
of spherical intrinsic volumes.  In
Appendix~\ref{sec:asympt-thresh-calc}, we use the asymptotic polytope
angle computations of~\cite{Donoho2006} to compute the decay threshold
for ensembles of feasible cones of the \(\ell_1\) norm at sparse
vectors.  The approach follows roughly the same lines as
Propositions~\ref{prop:orth-thresh} and~\ref{prop:upper-bd-psi-orth},
but the argument requires a good deal of background information that
is tangential to this work.

\subsection{Relationship to linear inverse problems}
\label{sec:relat-line-inverse}

There is a useful link between the number of random linear measurements
required to identify a structured signal with a convex complexity measure
and the upper decay threshold of the associated feasible
cone. Roughly speaking, the upper decay threshold is the
ratio between the number of linear measurements required to identify a
structured signal and the ambient dimension.  This observation
provides a powerful method for determining decay thresholds.

Linear inverse problems are closely related to demixing problems.  Suppose we observe the linear
image \(\vct z_0 = \mtx A\vct x_0\), where \(\mtx A\) is a known
matrix and \(\vct x_0\) is a structured vector.  Given an associated convex complexity measure \(f\), Chandrasekaran et al.~\cite{ChaRecPar:12} study the
convex optimization program
\begin{equation}
  \label{eq:generic-lin-inv}
  \minprog{}{f(\vct x)}{\mtx A\vct x = \vct z_0.}
\end{equation}
These authors consider the question ``Given the data \(\vct z_0 =\mtx
A \vct x_0\), when is \(\vct x_0\) the unique optimal point
of~\eqref{eq:generic-lin-inv}?''  The answer to this question is closely related to our demixing problem. 

To place the linear inverse problem in our asymptotic framework, we consider an ensemble of problems indexed by the ambient dimension \(d\). Fix an undersampling
parameter \(\sigma \in [0,1]\).  For each \(d\) in some infinite set \(\Index\subset \mathbb{N}\) of indices, assume we are given a structured vector \(\vct x_0^{(d)}\in \R^d\), a Gaussian measurement matrix \(\mtx \Omega^{(d)}\in \R^{\lceil \sigma d\rceil \times d}\), the observation \(\vct z_0^{(d)} = \mtx{ \Omega}^{(d)} \vct x_0^{(d)} \in \R^d\), and a complexity measure \(f^{(d)}\) associated with the structure of \(\vct x_0^{(d)}\).  We attempt to identify \(\vct x_0^{(d)}\) by solving the optimization
problem
\begin{equation}
  \label{eq:lin-inv-by-dim}
  \minprog{}{f^{(d)}(\vct x)}{\mtx \Omega^{(d)}\vct x = \vct z_0^{(d)}}
\end{equation}
with decision variable \(\vct x\in \R^d\).  This method succeeds when
\(\vct x_0^{(d)}\) is the unique optimal point
of~\eqref{eq:lin-inv-by-dim}.  The following result shows that
the problem of computing the number of random linear measurements
needed to identify a structured vector is equivalent to determining an
upper decay threshold.
\begin{lemma}\label{lem:lin-inv-and-decay}
 Consider the ensemble described above.
  \begin{enumerate}\renewcommand{\itemsep}{0pt}
  \item \label{item: thresh-implies-success} Suppose the ensemble
    \(\bigl\{\Fcone(f^{(d)},\vct x_0^{(d)})\bigr\}\) of feasible cones
    has an upper decay threshold \(\theta_\star < \sigma \).
    Then~\eqref{eq:lin-inv-by-dim} succeeds with overwhelming
    probability in high dimensions.
  \item \label{item:success-implies-thresh} On the other hand, suppose
    the linear inverse program~\eqref{eq:lin-inv-by-dim} succeeds with
    overwhelming probability in high dimensions. Then the ensemble
    \(\bigl\{\Fcone(f^{(d)},\vct x_0^{(d)})\bigr\}\) of feasible cones
    has an upper decay threshold \(\theta_\star = \sigma\).
  \end{enumerate}
\end{lemma}
The proof of Lemma~\ref{lem:lin-inv-and-decay} appears in
Appendix~\ref{sec:proof-thoer-coroll}, but we sketch the main ideas here.  The program~\eqref{eq:generic-lin-inv} identifies \(\vct x_0\)
precisely when the null space of \(\mtx A\) intersects the feasible
cone \(\Fcone(f,\vct x_0)\)
trivially~\cite[Prop.~2.1]{ChaRecPar:12}.  When \(\mtx A =\mtx \Omega\) is a
Gaussian matrix, its null space is a randomly oriented
subspace. Combining this fact with the kinematic formula~\eqref{eq:sphere-kin}  lets us compute the decay threshold using the number of observations required to recover a vector  with~\eqref{eq:generic-lin-inv} under a Gaussian measurement model.

In Appendix~\ref{sec:reconciliation}, we describe how our computation
of the upper decay threshold for the feasible cone of the \(\ell_1\)
norm at sparse vectors relates to the recovery guarantees for basis
pursuit explored in the series of
papers~\cite{Donoho2004}, \cite{Donoho2006}, \cite{Donoho2009}, \cite{Donoho2009a}, \cite{Donoho2010}.
Our approach also yields a sharp transition between success and
failure regimes for basis pursuit. See
Appendix~\ref{sec:matching-upper-bound} for the details.

\begin{remark}
  A counterpart to Lemma~\ref{lem:lin-inv-and-decay} that links the
  failure of linear inverse problems to the lower decay threshold
  \(\kappa_\star\) is readily derivable with the techniques used in
  this work.  This may enable the computation of lower decay
  thresholds through information-theoretic arguments~\cite{MR2597190}.
\end{remark}

\subsubsection{The upper decay threshold from the
  Gaussian width}
\label{sec:using-gaussian-width}
Lemma~\ref{lem:lin-inv-and-decay} provides a powerful tool for
computing upper decay thresholds.  Define the \emph{Gaussian width} of
a cone \(K\subset \R^d\) by
the expression
\begin{equation*}
  W(K\cap \mathsf{S}^{d-1}) \defeq 
  \mathbb{E}\left[\sup_{\vct x \in K\cap \mathsf{S}^{d-1}} 
    \langle \vct \omega, \vct x\rangle\right],
\end{equation*}
where the random vector \(\vct
\omega\) is drawn from the Gaussian distribution on \(\R^d\).  The
following corollary lets us determine upper decay thresholds from
Gaussian width bounds.  As usual, \(\Index\) is an infinite
subset of the natural numbers.
\begin{corollary}\label{cor:gauss-width-bd}
  Consider an ensemble \(\bigl\{\Fcone(f^{(d)},\vct x_0^{(d)})\subset
  \R^d\mid d \in \Index\bigr\}\) of feasible cones.  Suppose that
  \begin{equation}\label{eq:limsup-width}
    \limsup_{d\to \infty} \frac{1}{d} 
    W\bigl(\Fcone(f^{(d)},\vct x_0^{(d)})\cap \mathsf{S}^{d-1}\bigr)^2 \le \theta_\star.
  \end{equation}
  Then \(\theta_\star\) is an upper decay threshold (at level zero) for
  \(\bigl\{\Fcone(f^{(d)},\vct x_0^{(d)})\bigr\}\).
\end{corollary}
\noindent 
The proof relies the result~\cite[Cor.~3.3(1)]{ChaRecPar:12}, which asserts that the linear inverse problem~\eqref{eq:lin-inv-by-dim} succeeds with high probability when the number of measurements \(n^{(d)}\gtrsim W\bigl(\Fcone(f^{(d)},\vct x_0^{(d)})\cap \mathsf{S}^{d-1}\bigr)^2\).  Since \(\lceil\theta_\star d\rceil\) exceeds the width for large \(d\) by assumption~\eqref{eq:limsup-width},  taking \(n^{(d)}=\lceil \theta_\star d\rceil\)  results in high-probability success in~\eqref{eq:lin-inv-by-dim}.  Corollary~\ref{cor:gauss-width-bd} then follows from the second part of Lemma~\ref{lem:lin-inv-and-decay}.  The details appear   at the end of Appendix~\ref{sec:proof-thoer-coroll}.

This machinery allows us to compute upper decay thresholds in
situations involving matrix observations. For simplicity, we consider
only the space \(\R^{n\times n}\) of square \(n\times n\) matrices,
where \(n \in \mathbb{N}\).  The ambient dimension of this vector
space is \(d=n^2\), but we will index the observations by the
parameter \(n\).  This poses no difficulty, because it is equivalent
to indexing over the set \(\mathcal{D}=\{1^2,2^2,3^2,\dotsc\}\).
\begin{proposition}\label{prop:schatten1-bound}
  Fix \(\rho \in [0,1]\).  For each \(n\in \mathbb{N}\), let \(\mtx
  X_0^{(n)} \in \R^{n\times n}\) be a matrix with \(\rank(\mtx
  X_0^{(n)}) = \lceil \rho n\rceil\).  Then the corresponding ensemble
  \(\bigl\{\Fcone(\sone{\cdot},\mtx X_0^{(n)})\bigr\}\) of feasible
  cones has upper decay threshold \(\theta_\star = 6 \rho - 3\rho^2\).
\end{proposition}
\begin{proof}
  Let \(r=r^{(n)}= \rank(\mtx X_0^{(n)})\).
  From~\cite[Prop.~3.11]{ChaRecPar:12}, we have
  \begin{equation}\label{eq:S1-wdth-bd}
    W\Bigl(\Fcone(\sone{\cdot},\mtx X_0^{(n)})\cap \mathsf{S}^{n^2-1}\Bigr)^2 
    \le 6 r n- 3r^2 = \bigl(6\rho  - 3\rho^2\bigr) n^2.
  \end{equation}
  Dividing both sides by the ambient dimension \(n^2\) and taking   limits, we see that the conditions of   Corollary~\ref{cor:gauss-width-bd} hold.  This gives the result.
\end{proof}
\begin{remark}\label{rem:OH10}
  An asymptotically sharp upper bound for the Gaussian width in~\eqref{eq:S1-wdth-bd} is given by the solution to an implicit equation  in~\cite[Eq.~(94)]{OymHas:10}.   We use this asymptotically precise formula %
 for computing the location of the green curve in  Figure~\ref{fig:l1-nuc-phase}.  See Section~\ref{sec:low-rank-matrices} for more details.
\end{remark}

\begin{proposition}\label{prop:orthogonal-bound}
  For each \(n\in \mathbb{N}\), let \(\vct X_0^{(n)} \) be an
  orthogonal matrix.  Then the corresponding ensemble
  \(\bigl\{\Fcone(\opnorm{\cdot},\mtx X_0^{(n)})\bigr\} \) of feasible
  cones has upper decay threshold \(\theta_\star = \frac{3}{4}\).
\end{proposition}
\begin{proof}
  From~\cite[Prop.~3.13]{ChaRecPar:12}, we have
  \begin{equation*}
    W\Bigl(\Fcone(\opnorm{\cdot},\mtx X_0^{(n)})\cap
    \mathsf{S}^{n^2-1}\Bigr)^2
    \le \frac{3 n^2 - n}{4} = \frac{3}{4}n^2 + O(n).
  \end{equation*}
  The result follows upon dividing by \(n^2\), taking limits, and
  applying Corollary~\ref{cor:gauss-width-bd}.
\end{proof}

We list all of the bounds computed in our work in
Table~\ref{tab:aeubs-from-linear-inverse} on page~\pageref{tab:aeubs-from-linear-inverse}, but note that several more
bounds are readily derivable from the Gaussian width calculations
in~\cite[Sec.~3.4]{ChaRecPar:12}.

\section{Experiments}
\label{sec:applications}

We can tackle a variety of scenarios using the theory developed in
Section~\ref{sec:find-weak-thresh} and the decay threshold
calculations of Section~\ref{sec:line-inverse-probl-1}.

\subsection{Morphological component analysis}
\label{sec:deconv-sparse-vect}
\begin{figure}[t!]
  \centering
\includegraphics[width=0.48\columnwidth]{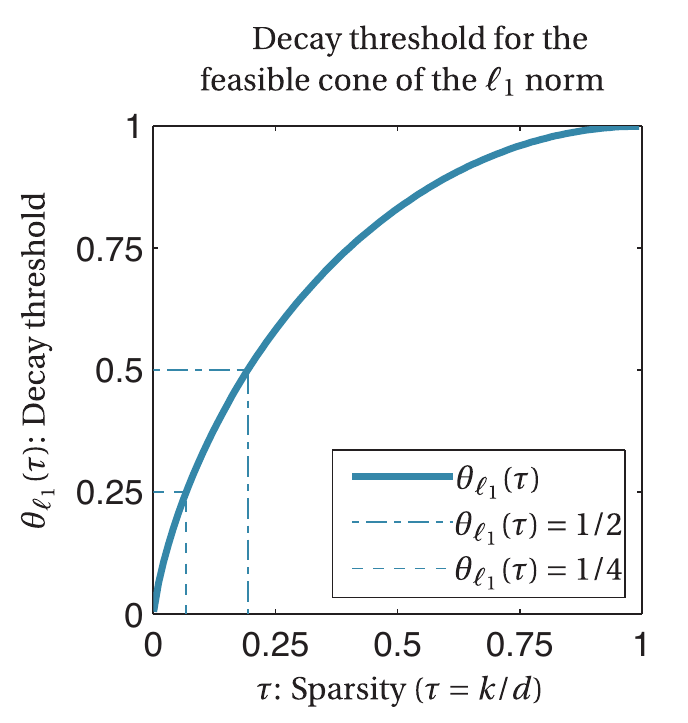}
\hspace{6pt}
    \includegraphics[width=0.48\columnwidth]{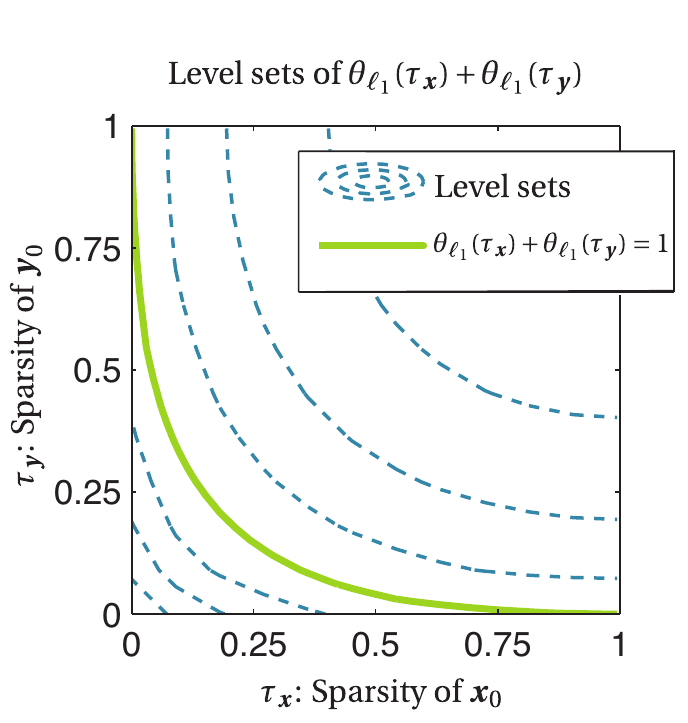}
  \caption{\textsl{Decay threshold for the \(\ell_1\) norm.}  The left
    panel shows \(\theta_{\ell_1}(\tau)\), the upper decay threshold
    for the sequence of feasible cones of the \(\ell_1\) norm at
    \(\lceil \tau d \rceil\)-sparse vectors as a function of the
    sparsity \(\tau\).  The narrow lines show that
    \(\theta_{\ell_1}(\tau) < \frac{1}{2}\) for \(\tau<0.19\), while
    \(\theta_{\ell_1}(\tau) < \frac{1}{4}\) for \(\tau< 0.06\).  The
    right panel displays level sets of the function
    \(\theta_{\ell_1}(\tau_{\vct x}) + \theta_{\ell_1}(\tau_{\vct
      y})\).  The thick curve marks the level set
    \(\theta_{\ell_1}(\tau_{\vct x}) + \theta_{\ell_1}(\tau_{\vct
      y})=1\); it corresponds to the green curve in
    Figure~\ref{fig:l1-l1-demixing}.  }
  \label{fig:l1-decay-levelsets}
\end{figure}

We return to the MCA model of Section~\ref{sec:first-appl-deconv}.
Our goal is to analyze when we can demix two signals that are
sparse in incoherent bases. To apply our theoretical results, we
consider the demixing ensemble from
Section~\ref{sec:mca-high-dimensions}.

Fix sparsity levels \(\tau_{\vct x}\) and \(\tau_{\vct y}\) in
\([0,1]\).  For each dimension \(d\in \mathbb{N}\), we construct
signals \(\vct x_0\) and \(\vct y_0\) in \(\R^d\) that satisfy 
\begin{equation*}
  \nnz(\vct x_0)=\lceil \tau_{\vct x} d\rceil\quad\text{and}\quad\nnz(\vct y_0)=
  \lceil\tau_{\vct y} d\rceil.
\end{equation*}
Draw a random basis \(\mtx Q\),  and suppose that we observe \(\vct
z_0 = \vct x_0 + \mtx Q \vct y_0\).  

The \(\ell_1\) norm is a natural complexity measure for sparse
vectors.  Given the side information \(\alpha = \lone{\smash{\vct
    y_0}}\), we pose the constrained MCA problem
\begin{equation}\label{eq:l1-const-2}%
  \minprog{}{\lone{\vct x}}{\lone{\smash{\vct y}} \le\alpha
    \;\;\text{and}\;\;
    \vct x + \mtx Q \vct y = \vct z_0.}
\end{equation}
The theory of Section~\ref{sec:find-weak-thresh} describes
when~\eqref{eq:l1-const-2} identifies \((\vct x_0,\vct y_0)\) with
overwhelming probability in high dimensions in terms of decay
thresholds for the ensembles \(\{\Fcone(\lone{\cdot},\vct x_0)\}\) and
\(\{\Fcone(\lone{\cdot},\vct y_0)\}\) indexed by the dimension \(d\).

Observe that, up to rotations, the geometry of the feasible cone
\(\Fcone(\lone{\cdot},\vct w)\) depends only on the number of nonzero
entries in \(\vct w\), but not on the positions or magnitudes of the
entries. For a fixed sparsity level \(\tau \in [0,1]\), consider an
ensemble \(\{\vct w^{(d)}\}\) with \(\vct w^{(d)} \in \R^d\) and
\(\nnz(\vct w^{(d)}) = \lceil \tau d\rceil\) for each \(d\in
\mathbb{N}\).  In Appendix~\ref{sec:asympt-thresh-calc}, we compute
the optimal upper decay threshold for the ensemble
\(\{\Fcone(\lone{\cdot},\vct w^{(d)})\}\).  We denote this value by
\(\theta_{\ell_1}(\tau)\) (see Figure~\ref{fig:l1-decay-levelsets},
left panel).  Therefore, the ensembles \(\{\Fcone(\lone{\cdot},\vct
x_0)\}\) and \(\{\Fcone(\lone{\cdot},\vct y_0)\}\) have decay
thresholds \(\theta_{\ell_1}(\tau_{\vct x})\) and
\(\theta_{\ell_1}(\tau_{\vct y})\).

Theorem~\ref{thm:sharp-thresh} implies that our demixing
method~\eqref{eq:l1-const-2} for sparse vectors succeeds with
overwhelming probability in high dimensions so long as
\(\theta_{\ell_1}(\tau_{\vct x}) + \theta_{\ell_1}(\tau_{\vct y})<
1\). The right panel of Figure~\ref{fig:l1-decay-levelsets} shows the
level sets of the function \(\theta_{\ell_1}(\tau_{\vct x})+
\theta_{\ell_1}(\tau_{\vct y})\) for \((\tau_{\vct x},\tau_{\vct y})\)
on the unit square \([0,1]^2\).  The green curve is the level set
\begin{equation*}
\bigl\{(\tau_{\vct x},\tau_{\vct y})\mid \theta_{\ell_1}(\tau_{\vct x})
+\theta_{\ell_1}(\tau_{\vct y})=1\bigr\}.
\end{equation*}
Theorem~\ref{thm:sharp-thresh} implies that
program~\eqref{eq:l1-const-2} succeeds with overwhelming probability
when the joint sparsity \((\tau_{\vct x},\tau_{\vct y})\) lies below
the green curve.

On the other hand, our computations show that the upper decay
threshold \(\theta_{\ell_1}(\tau)\) is numerically equal to the lower
decay threshold \(\kappa_{\ell_1}(\tau)\)---see the discussion in
Appendix~\ref{sec:comp-asympt-expon}.
Theorem~\ref{thm:failure-thresh} implies that the sparse demixing
method~\eqref{eq:l1-const-2} fails with overwhelming probability for
sparsity levels \((\tau_{\vct x},\tau_{\vct y})\) in the region above
the green curve.  In other words, the green curve on the right panel
of Figure~\ref{fig:l1-decay-levelsets} delineates a sharp transition
between success and failure for constrained MCA~\eqref{eq:l1-const-2}.
The green curve in Figure~\ref{fig:l1-decay-levelsets}(b) is the same as the
green curve in Figure~\ref{fig:l1-l1-demixing}.

\subsubsection{Strong guarantees}
\label{sec:l1-l1-strong-transitions}
The theory of Section~\ref{sec:strong-phase-trans} allows us to
provide a uniform recovery guarantee. For a fixed draw of the random
basis, constrained  MCA can demix all sufficiently sparse pairs of
vectors with overwhelming probability in high dimensions.

Fix the sparsity \((\tau_{\vct x},\tau_{\vct y})\) and the ambient
dimension \(d\).  Suppose \(\mtx Q\in \mathsf{O}_d\) is a random basis.  By
Lemma~\ref{lem:main-geom-introduction}, constrained MCA can identify 
every \((\tau_{\vct x},\tau_{\vct y})\)-sparse pair \((\vct x_0,\vct
y_0)\) given the observation \(\vct z_0 = \vct x_0 +\mtx Q\vct y_0 \),
provided that the event
\begin{equation}\label{eq:l1-strong-event}
  \Fcone(\lone{\cdot},\vct x_0) \bigcap\bigl( - \mtx Q
  \Fcone(\lone{\cdot},\vct y_0) \bigr)= \{\mathbf{0}\}  \text{ for each
    \((\tau_{\vct x},\tau_{\vct y})\)-sparse pair \(\bigl(\vct x_0,\vct y_0\bigr)\)}
\end{equation}
holds.  Theorem~\ref{thm:strong-thm} guarantees that the probability
of event~\eqref{eq:l1-strong-event} is large when some associated
decay thresholds are small enough; let us describe how to verify the required
technical assumptions.

First, the results in Appendix~\ref{sec:asympt-thresh-calc} allow us
to compute upper decay threshold at levels \(\psi \ge 0\) for the
ensemble \(\{\Fcone(\lone{\cdot},\vct w^{(d)})\}\) from
Section~\ref{sec:deconv-sparse-vect} consisting of feasible cones for
the \(\ell_1\) norm at \(\tau\)-sparse vectors.  We extend our earlier
notation by writing this quantity as \(\theta_{\ell_1}(\tau,\psi)\).
This is the first element required to check the hypotheses of
Theorem~\ref{thm:strong-thm}.

We also require information on the total number of feasible cones under
consideration. Let
\begin{equation*}
  \mathcal{K}^{(d)}_{\vct x} = \bigcup\bigl\{
  \Fcone(\lone{\cdot},\vct x_0)\bigr\}, \quad \text{and} \quad  
  \mathcal{K}^{(d)}_{\vct y} =
  \bigcup\bigl\{\Fcone(\lone{\cdot},\vct y_0)\bigr\},
\end{equation*}
where the unions take place over all \(\lceil \tau_{\vct x} d\rceil\)-sparse
\(\vct x_0 \in \R^d\) and all \(\lceil \tau_{\vct y} d\rceil\)-sparse \(\vct y_0 \in
\R^d\).  There are exactly
\begin{math}
  2^{k}\binom{d}{k}
\end{math}
different---but congruent---feasible cones of the
\(\ell_1\) norm at vectors in \(\R^d\) with \(k\) nonzero entries, one
for each sign/sparsity pattern.  This corresponds to the number of
\((k-1)\)-dimensional faces of the crosspolytope; see,
e.g.,~\cite[Sec.~3.3]{Donoho2006}.  With \(k=\lceil \tau d \rceil\),
it follows from the proof of Proposition~\ref{prop:orth-thresh} that
\begin{equation}\label{eq:binom-limit}
  d^{-1}\log\left(2^{k} 
    \binom{d}{k}\right) \to\tau \log(2) + H(\tau) \eqdef E(\tau),
\end{equation}
uniformly as \(d \to \infty\).  By continuity of the exponential, for
every \(\eta>0\) and all sufficiently large \(d\), the number of
feasible cones is bounded above by
\begin{equation*}
|\mathcal{K}^{(d)}_{\vct x}|=   2^{\lceil \tau_{\vct x}
  d\rceil}\binom{d}{\lceil \tau_{\vct x} d\rceil} \le \mathrm{e}^{d
  (E(\tau_{\vct x}) +\eta)},
\end{equation*}
for large enough \(d\).  Similarly, \(|\mathcal{K}^{(d)}_{\vct y}| \le
\mathrm{e}^{d (E(\tau_{\vct y}) +\eta)}\) for all sufficiently large
\(d\).

We have now collected enough information to apply our theory.  By
Theorem~\ref{thm:strong-thm}, the event~\eqref{eq:l1-strong-event} holds
with overwhelming probability in high dimensions so long as
\begin{equation}\label{eq:l1-strong-ineq}
  \theta_{\ell_1}\bigl(\tau_{\vct x},E(\tau_{\vct x}) + E(\tau_{\vct y}) \bigr)
  +
  \theta_{\ell_1}\bigl(\tau_{\vct y},E(\tau_{\vct x}) + E(\tau_{\vct y}) \bigr) <1.
\end{equation}
The blue curve appearing in Figure~\ref{fig:l1-l1-demixing} on page~\pageref{fig:l1-l1-demixing} shows the level set
\begin{equation*}
  \Bigl\{(\tau_{\vct x},\tau_{\vct y}) \mid
  \theta_{\ell_1}\bigl(\tau_{\vct x},E(\tau_{\vct x}) + E(\tau_{\vct y}) \bigr)
  +
  \theta_{\ell_1}\bigl(\tau_{\vct y},E(\tau_{\vct x}) + E(\tau_{\vct y}) \bigr) = 1\Bigr\}.
\end{equation*}
When the sparsity level \((\tau_{\vct x},\tau_{\vct y})\) lies below
the blue curve, inequality~\eqref{eq:l1-strong-ineq} holds, so that the
demixing method~\eqref{eq:l1-const-2} succeeds at demixing
\emph{every} pair \((\vct x_0,\vct y_0)\) of \((\tau_{\vct
  x},\tau_{\vct y})\)-sparse vectors with overwhelming probability in
high dimensions.

\subsubsection{Numerical experiment for constrained MCA}
\label{sec:MCA-numerical-experiment}

The following numerical experiment illustrates the accuracy of these theoretical
results.  We fix the dimension \(d=100\).  For each pair \((k_{\vct
  x},k_{\vct y})\in \{0,1,\dotsc,100\}^2\), we repeat the following
procedure \(25\) times.
\begin{enumerate}\setlength{\itemsep}{0pt}
\item Draw \(\vct x_0, \vct y_0 \in \R^d\) with \(k_{\vct x}\) or
  \(k_{\vct y}\) nonzero elements, respectively.  The locations of the
  nonzero elements are chosen at random, and these elements are
  equally likely to be \(+ 1\) or \(-1\).
\item Generate a random basis \(\mtx Q \in \mathsf{O}_d\); see
  Remark~\ref{rem:random-basis} below.
\item Solve~(\ref{eq:l1-const-2}) for the optimal point \((\vct
  x_\star,\vct y_\star)\) with the numerical optimization software
  \texttt{CVX}~\cite{Grant2008}, \cite{Grant2010}.
\item Declare success if \(\linf{\vct x_\star - \vct x_0}<10^{-4}\). 
\end{enumerate}
The background of Figure~\ref{fig:l1-l1-demixing} shows the
results of this experiment as a function of \(\tau_{\vct x} = k_{\vct
  x}/d\) and \(\tau_{\vct y}=k_{\vct y}/d\).  The yellow curve marks
the empirical \(50\%\) success line, and we note that it tracks the
green theoretical curve closely.  It emerges that \(d=100\) is already
large enough to see the high dimensional behavior described by our
theoretical results.

\begin{remark}\label{rem:random-basis}
  A random basis in \(\mtx Q\in \mathsf{O}_d\) is often defined
  through a conceptually simple two-step operation.  First, draw a
  square \(d\times d\) Gaussian matrix; second, orthogonalize the
  columns of this matrix via the Gram--Schmidt procedure.  Although
  this definition has the flavor of a numerical algorithm, the conceptual process is not numerically  stable. Moreover, standard procedures for stabilizing the
  orthogonalization do not preserve the Haar measure.  For a
  straightforward, numerically stable approach to generating random
  bases, see~\cite{Mezzadri2007}.
\end{remark}

\subsection{Secure and robust channel coding}
\label{sec:channel-coding-1}

Next, we study the secure channel coding scheme of
Section~\ref{sec:chann-coding-scheme}.  We want to analyze when the
receiver can decode a transmitted message that is subject to a sparse
corruption.  The difficulty depends on the sparsity level \(\tau\) of
the corruption, where \(\tau\in [0,1]\). Let us introduce an ensemble
of demixing problems.  In each dimension \(d\in \mathbb{N}\),
choose a \(d\)-bit message \(\vct m_0 \in \{\pm 1\}^d\) and a
corruption \(\vct c_0\in\R^d \) with \(\nnz(\vct c_0) = \lceil \tau
d\rceil\).  Draw a random basis \(\mtx Q\) known to both the receiver
and transmitter.  The receiver observes \(\vct z_0 = \mtx Q\vct
m_0+\vct c_0\), the encoded message plus the sparse interference.

The natural complexity measure for the sparse corruption \(\vct c_0\)
is the \(\ell_1\) norm, while the \(\ell_\infty\) norm is the
appropriate complexity measure for the sign vector \(\vct m_0\).  Note
that \(\linf{\vct m_0}=1\).  To recover the original message, the
receiver solves the problem
\begin{equation}
  \label{eq:channel-coding}
  \minprog{}{\lone{\vct c}}{
      \linf{\vct m} \le 1 \;\;\text{and}\;\;
    \vct c + \mtx Q \vct m  = \vct z_0.}
\end{equation}
This approach succeeds when \((\vct m_0,\vct c_0)\) is the unique optimal point
of~\eqref{eq:channel-coding}.

We consider two types of corruptions.  \emph{Benign} corruptions are
taken in any manner that is independent of \(\mtx Q\)---this would
happen, for instance, when the corruption is generated by an adversary
with no knowledge of \(\mtx Q\) or the transmission \(\mtx Q \vct
m_0\).  \emph{Adversarial} corruptions are worst-case sparse
corruptions; these corruptions may model malicious interference or an
erasure channel that sets some of the coordinates of \(\vct z_0\) to
zero.  We first consider with the benign corruption, and then consider
the adversarial case in Section~\ref{sec:linf-strong-thresholds}.

\subsubsection{Benign corruptions}
\label{sec:benign-corr-case}

Since the corruption is chosen independently of the basis \(\mtx Q\),
the \CDMq~\eqref{eq:channel-coding} succeeds if it can identify the
single pair \((\vct c_0,\vct m_0)\).  Theorems~\ref{thm:sharp-thresh}
and~\ref{thm:failure-thresh} describe where procedure succeeds, or
fails, with overwhelming probability in high dimensions, but we must
first determine some decay thresholds.

The feasible cone of the \(\ell_\infty\) norm at a sign vector is
 congruent to an orthant (Example~\ref{ex:linf-feas}).  From
Table~\ref{tab:aeubs-from-linear-inverse}, we see that the sequence
\(\{\R^d_+\mid d\in \mathbb{N}\}\) of orthants has an upper decay threshold of
\(\smash{\theta_{\smash{\R^d_+}}} =\frac{1}{2}\) and a matching lower decay
threshold \(\kappa_{\R^d_+}=\frac{1}{2}\).
Appendix~\ref{sec:asympt-thresh-calc} defines an upper decay threshold
\(\theta_{\ell_1}(\tau)\) for the sequence of feasible cones of the
\(\ell_1\) norm at \(\tau\)-sparse vectors.

With these thresholds in hand, Theorem~\ref{thm:sharp-thresh}
guarantees that~\eqref{eq:channel-coding} succeeds with overwhelming
probability in high dimensions provided that \(
\theta_{\R^d_+}+\theta_{\ell_1}(\tau) < 1\), or, equivalently, provided that
\begin{math}
 \theta_{\ell_1}(\tau) < \frac{1}{2}.
\end{math}
On the other hand, Theorem~\ref{thm:failure-thresh} shows
that~\eqref{eq:channel-coding} fails with overwhelming probability in
high dimensions if the corresponding upper decay threshold satisfies
\begin{math}
  \kappa_{\ell_1}(\tau) + \kappa_{\R^d_+} > 1.
\end{math}
This condition holds when \(\kappa_{\ell_1}(\tau)> \frac{1}{2}\). 

Numerically, we find that \(\theta_{\ell_1}(\tau)<\frac{1}{2}\) when
the sparsity level \(\tau \lesssim 0.193\); see the left panel of
Figure~\ref{fig:l1-decay-levelsets}.  The fact that
\(\kappa_{\ell_1}(\tau)=\theta_{\ell_1}(\tau)\) to numerical precision
implies that \(\kappa_{\ell_1}(\tau)> \frac{1}{2}\) for \(\tau\gtrsim
0.193\).  In other words, if fewer than \(19\%\) of the entries of
\(\vct c_0\) are nonzero, the scheme~\eqref{eq:channel-coding}
succeeds with overwhelming probability in high dimensions; otherwise,
it fails with overwhelming probability in high dimensions.  This sharp
transition corresponds to the location of the dashed line in
Figure~\ref{fig:channel-coding}.

\subsubsection{The adversarial case}
\label{sec:linf-strong-thresholds}
\begin{figure}[t!]
  \centering
  \includegraphics[width=.6\columnwidth]{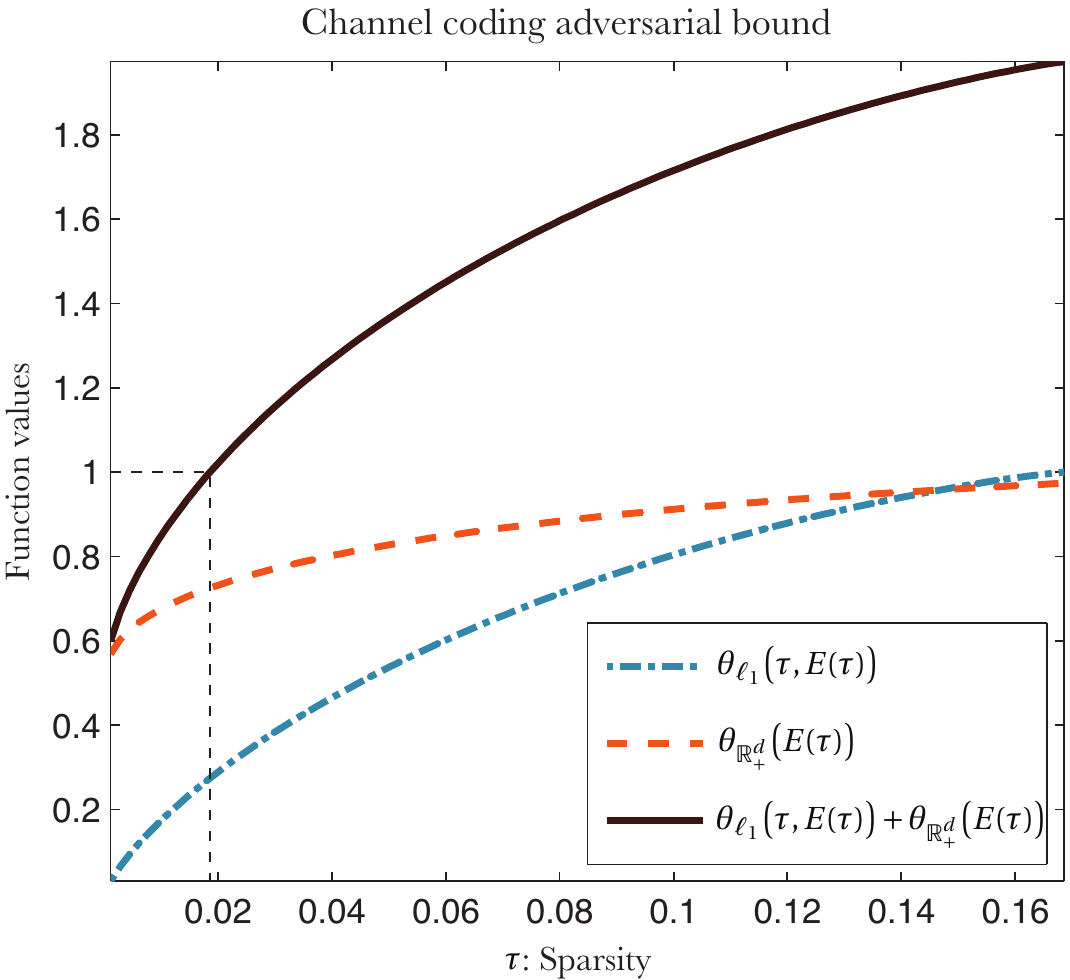}
  \caption{\textsl{Calculating the  adversarial guarantees for
    channel coding.} The lower two curves in the figure
    correspond to \(\theta_{\ell_1}(\tau,E(\tau))\) and
    \(\theta_{\smash{\R^d_+}}(E(\tau))\) defined by~\eqref{eq:theta-l1} and~\eqref{eq:upper-bd-psi-orth}. The upper curve indicates the sum
    \(\theta_{\ell_1}(\tau,E(\tau))+\theta_{\smash{\R^d_+}}(E(\tau))\).  For
    \(\tau<0.018\), the sum lies below one, so
    Theorem~\ref{thm:strong-thm} implies that
    event~\eqref{eq:adv-corrupt-event} holds with overwhelming
    probability.  }
  \label{fig:adv-recov-guar}
\end{figure}

In the adversarial case, the corruptions are sparse but may depend on
the basis \(\mtx Q\) and the message \(\vct m_0\). To ensure that no
corruption with \(\lceil \tau d\rceil\) nonzero entries
can cause~\eqref{eq:channel-coding} to fail, we must verify
that~\eqref{eq:channel-coding} succeeds at identifying \((\vct
c_0,\vct m_0)\) for \emph{every} \(\lceil \tau d \rceil\)-sparse
vector \(\vct c_0\).  From Lemma~\ref{lem:triv-intersect}, this is
equivalent to the event
\begin{equation}
  \label{eq:adv-corrupt-event}
  \Fcone(\lone{\cdot},\vct c_0)\cap\bigl( -\mtx Q \Fcone(\linf{\cdot},\vct
  m_0) \bigr)= \{\zerovct\} \text{ for every \(\lceil \tau d\rceil\)-sparse vector \(\vct c_0\)}.
\end{equation}
We can use Theorem~\ref{thm:strong-thm} to verify that
event~\eqref{eq:adv-corrupt-event} holds with overwhelming probability
in high dimensions. Let us collect the additional information required
to verify the technical assumptions of this theorem.

For \(\vct m_0\in \{\pm 1\}^d\), define the singleton
\(\mathcal{K}_{\vct m}^{(d)} = \{\Fcone(\linf{\cdot},\vct m_0)\}\).
The feasible cone \(\Fcone(\linf{\cdot},\vct m_0)\) is 
congruent to the orthant \(\R^d_+\), so the sequence
\(\mathcal{K}_{\vct m}^{(d)}\) of sets has an upper decay threshold at
level \(\psi\) of \(\theta_{\R^d_+}(\psi)\) defined in
Proposition~\ref{prop:upper-bd-psi-orth}.

Define the set \(\mathcal{K}^{(d)}_{\vct c} = \bigcup
\{\Fcone(\lone{\cdot},\vct c_0)\}\), where the union occurs over all
vectors \(\vct c_0\) with \(\lceil \tau d\rceil\) nonzero elements.
By~\eqref{eq:binom-limit}, the size of
\(\mathcal{K}^{(d)}_{\vct c} \) is bounded by
\begin{equation*}
    \abs{\mathcal{K}^{(d)}_{\vct c}} = 2^{\lceil \tau d\rceil}
      \binom{d}{\lceil \tau d\rceil} \le \econst^{d(E(\tau)+\eta)}
\end{equation*}
for any \(\eta>0\) and all sufficiently large \(d\).  

Since \(\mathcal{K}^{(d)}_{\vct m}\) is a singleton,
Theorem~\ref{thm:strong-thm} implies that
event~\eqref{eq:adv-corrupt-event} holds with overwhelming probability
in high dimensions whenever
\begin{equation}\label{eq:l1-linf-strong-ineq}
\theta_{\R^d_+}\bigl(E(\tau)\bigr) + \theta_{\ell_1}\bigl(\tau,E(\tau)\bigr) < 1.
\end{equation}
Computing \(\theta_{\ell_1}(\tau,\psi)\) and \(\theta_{\R^d_+}(\psi)\)
numerically, we find that for all \(\tau \lesssim 0.0186\),
inequality~\eqref{eq:l1-linf-strong-ineq} holds.  We conclude that our
channel coding scheme is robust to all adversarial corruptions
so long as the corruptions have no more than about \(1.8\% \) nonzero
entries.  This computation is illustrated in
Figure~\ref{fig:adv-recov-guar}.

\subsubsection{Numerical experiment}
\label{sec:numerical-experiment-channel-coding}

We perform two numerical experiments to complement our theory, one for
the benign corruptions and the other for a specific type of malicious erasure.

For dimensions \(d = 100\) and \(d=300\) and for each of \(70\)
equally spaced values of \(\tau\in [0,0.35]\), we test the benign
corruption case by repeating the following procedure \(200\) times:
\begin{enumerate}
  \setlength{\itemsep}{0pt}
\item Draw a binary vector \(\vct m_0 \in \{\pm1\}^d\) at random. 
\item Choose a corruption \(\vct c_0\) with \(k= [\tau d]\) nonzero
  elements; the support of \(\vct c_0\) is random, and the nonzero
  elements are taken to be \(\pm 1\) with equal probability.
\item Generate a random basis \(\mtx Q\in \mathsf{O}_d\); see
  Remark~\ref{rem:random-basis}.
\item Solve~(\ref{eq:channel-coding}) with the observation \(\vct z_0
  = \mtx Q \vct m_0 + \vct c_0\) with the numerical optimization
  software \texttt{CVX}; call \((\vct m_\star,\vct c_\star)\) the
  optimal point.
\item Declare success if \(\linf{\vct m_\star - \vct m_0}<10^{-4}\).
\end{enumerate}

The second experiment incorporates a malicious erasure.  As in the
benign case, the experiment  is run for dimensions \(d = 100\) and
\(d=300\) and for \(70\) equally spaced values of \(\tau\) between
zero and one.  For each of these parameters, we repeat the following
\(200\) times:
\begin{enumerate}\setlength{\itemsep}{0pt}
\item Draw a message \(\vct m_0 \in \{\pm1\}^d\) at random, and
  generate a random basis \(\mtx Q \in \mathsf{O}_d\).
\item Set the observation \(\vct z_0 = \mathrm{erase}(\mtx Q\vct
  m_0,[\tau d])\), where \(\mathrm{erase}(\vct x,k)\) sets the \(k\)
  largest-magnitude elements of \(\vct x\) to zero.
\item Solve~\eqref{eq:channel-coding} with \texttt{CVX} for the
  optimal point \((\vct m_\star,\vct c_\star)\), and
\item Declare success if \(\linf{\vct m_\star -\vct m_0}< 10^{-4}\).  
\end{enumerate}

The curves in Figure~\ref{fig:channel-coding} show the results of
these experiments.  For benign corruptions, the  empirical \(50\%\) success rate occurs very near the
predicted sparsity value \(\tau = 0.193\), and the transition region
is more narrow for larger \(d\).  Thus, the experiment closely match our prediction for the location of the asymptotic phase transition.

The empirical evidence suggests that our
adversarial guarantees are conservative for the type of
malicious corruption used in the experiment.  This is expected,
as we have no reason to believe that such erasures correspond to the
worst-case corruption.  However, the empirical transition between
success and failure near \(\tau\approx 0.05\)  suggests that 
our adversarial bound lies within a factor of two or three of the
best possible guarantee.

\subsection{Low-rank matrices under sparse corruptions}
\label{sec:low-rank-matrices}
Consider the problem of separating a low-rank square matrix from a
corruption that is sparse in a random basis.  The parameters that
determine the difficulty are the proportional rank \(\rho\in [0,1]\)
and the sparsity level \(\tau\in [0,1]\).  Let us introduce a
demixing ensemble.  For each side length \(n\in \mathbb{N}\),
choose a low-rank matrix \(\mtx X_0\in \R^{n\times n}\) with
\(\rank(\mtx X_0) = \lceil \rho n\rceil\) and a sparse matrix \(\mtx
Y_0 \in \R^{n\times n}\) with \(\nnz(\mtx Y_0) = \lceil \tau
n^2\rceil\).  Draw a random basis \(\mathcal{Q}\) for \(\R^{n\times
  n}\), and suppose we observe \(\mtx Z_0 = \mtx X_0 +
\mathcal{Q}(\mtx Y_0)\).

To promote low-rank, we use the Schatten 1-norm, and to promote
sparsity, we use the matrix \(\ell_1\) norm. Given the side information
\(\alpha=\lone{\mtx Y_0}\), we pose the \CDMq\
\begin{equation}
  \label{eq:l1-nuc-recover}
  \minprog{}{\sone{\mtx X} }{ \lone{ \mtx Y} \le\alpha \;\;\text{and}\;\;
    \mtx X + \mathcal{Q}(\mtx Y) = \mtx Z_0.}
\end{equation}
We study when \((\mtx X_0,\mtx Y_0)\) is the unique solution
to~\eqref{eq:l1-nuc-recover} with overwhelming probability in high
dimensions.

The feasible cone of the \(\ell_1\)-matrix norm at a matrix \(\mtx Y_0
\in \R^{n \times n}\) with \(k\) nonzero elements is isomorphic to the feasible
cone of the \(\ell_1\) norm at a sparse vector \(\vct y_0 \defeq
\vec(\mtx Y_0)\in \R^{n^2}\) with \(k\) nonzero entries.  It follows that the
value \(\theta_{\ell_1}(\tau)\) from~\eqref{eq:theta-l1} is an upper
decay threshold for the ensemble \(\{\Fcone(\lone{\cdot},\vct Y_0)\}
\) of feasible cones indexed by the ambient dimension \(d=n^2\) of the
matrix space \(\R^{n\times n}\).

By Proposition~\ref{prop:schatten1-bound}, we see that
\(\theta_{S_1}(\rho) \defeq 6 \rho - 3 \rho^2\) is an upper decay
threshold for the ensemble \(\{\Fcone(\sone{\cdot},\mtx X_0)\}\) of
feasible cones, indexed by the ambient dimension.  However, a smaller upper decay threshold \(\tilde\theta_{S_1}(\rho)\) is available using the results of~\cite{OymHas:10}---see Remark~\ref{rem:OH10}.  The green line in  Figure~\ref{fig:l1-nuc-phase} is the level set
\begin{equation*}
  \bigl\{(\rho,\tau)\mid \theta_{\ell_1}(\tau) +
  \tilde{\theta}_{S_1}(\rho) = 1\bigr\},
\end{equation*}
where the \(\tilde\theta_{S_1}(\rho)\) is given by the asymptotic upper bound on the Gaussian width given implicitly in~\cite[Eq.~(94)]{OymHas:10}.\footnote{Our actual computation uses the simpler, but equivalent, formula given in~\cite[Prop.~4.9]{AmeLotMcC:13}.}
For \((\rho,\tau)\) pairs lying below the curve, 
Theorem~\ref{thm:sharp-thresh} implies that our demixing
method~(\ref{eq:l1-nuc-recover}) succeeds with overwhelming
probability in high dimensions.

\subsubsection{Numerical experiment}
\label{sec:numerical-experiment-1}

Let us summarize the experiment in Figure~\ref{fig:l1-nuc-phase}.  The
matrix side length \(n = 35\) is fixed, and for each pair \((\rho,\tau)\) in
the set
\(\smash{\bigl\{\frac{1}{n},\frac{2}{n},\dotsc,1\bigr\}}^2\), we
repeat the following procedure \(25\) times:
\begin{enumerate}\setlength{\itemsep}{0pt}
\item Draw a matrix \(\mtx X_0 = \mtx Q_L \mtx \Lambda \mtx Q_R \in
  \R^{n\times n}\) with rank \(r = [\rho n]\), where \(\mtx \Lambda\)
  is a diagonal matrix that satisfies \(\Lambda_{ii} = 1\) for \(i=1,\dotsc
  r\) and \(\Lambda_{ii}=0\) otherwise and \(\mtx Q_L\), \(\mtx Q_R\)
  are independent random bases in \(\mathsf{O}_n\).
\item Generate a random matrix \(\mtx Y_0\in \R^{n\times n}\) with
  \([\tau n^2]\) nonzero entries; the nonzero entries in \(\mtx Y_0\)
  take the values \(+ 1\) or \(-1\) with equal probability.
\item Generate a random basis \(\mathcal{Q}\) for \(\R^{n\times n}\).
\item Solve~(\ref{eq:l1-nuc-recover}) with the observation \(\mtx Z_0
  = \mtx X_0+\mathcal{Q}(\mtx Y_0)  \) with \texttt{CVX}, and set
  \((\mtx X_\star,\vct Y_\star) \) to the optimal point.
\item Declare success if \(\linf{\mtx X_\star - \mtx X_0}<10^{-4}\). 
\end{enumerate}
From Figure~\ref{fig:l1-nuc-phase}, we see that the theoretical bound
closely matches the empirical success curve throughout the entire regime.

\subsection{Assorted matrix demixing problems}
\label{sec:hodgepodge}

We conclude this section with some other combinations of structured
square matrices that we can separate using the
\CDMq~\eqref{eq:genconv}.  In each of these applications, we observe a
superposition of the form \(\mtx Z_0 = \mtx X_0 +\mathcal{Q} (\mtx
Y_0) \in \R^{n\times n} \), where \(\mathcal{Q}\) is a random basis
for the matrix space \(\R^{n\times n}\).    We consider various structures for \(\mtx X_0\)
and \(\mtx Y_0\)---either low rank, orthogonal, sparse, or sign
matrices---and we show that our theory quickly identifies a regime
where an appropriate \CDMq\ succeeds with overwhelming probability in
high dimensions.  While we know of no concrete applications for these particular demixing programs, the analysis below illustrates the ease with which our theory extends to new settings.

\subsubsection{Orthogonal and sparse matrices}
\label{sec:orthogonal-sparse}

Fix a sparsity level \(\tau \) in \([0,1]\).  For each side length
\(n\in \mathbb{N}\), choose an orthogonal matrix \(\mtx X_0 \in
\mathsf{O}_{n}\) and a sparse matrix \(\mtx Y_0\in \R^{n\times n}\)
with \(\nnz(\mtx Y_0) = \lceil \tau n^2\rceil\).  We use the operator
norm as a complexity measure for \(\mtx X_0\) and the matrix
\(\ell_1\) norm as a complexity measure for \(\mtx Y_0\).  Since
\(\opnorm{\mtx X_0}=1\), we pose the \CDMq
\begin{equation}
  \label{eq:orth-sparse}
  \minprog{}{\lone{\mtx Y}}{\opnorm{\vct X} \le 1\;\;\text{and}\;\; \mtx X+
    \mathcal{Q}(\mtx Y) = \vct Z_0.}
\end{equation}
The interchange of the objective and constraint as compared
with~\eqref{eq:genconv}  poses no difficulty because the
optimality conditions of Lemma~\ref{lem:triv-intersect} are 
symmetric with respect to the objective and constraint.  

By Theorem~\ref{thm:sharp-thresh}, the \CDMq~\eqref{eq:orth-sparse}
succeeds with overwhelming probability in high dimensions so long as
\begin{equation*}
  \theta_{\ell_1}(\tau) + \theta_{\mathrm{Op}} < 1,  
\end{equation*}
where \(\theta_{\ell_1}(\tau)\), defined in
Appendix~\ref{sec:asympt-thresh-calc}, is an upper decay threshold for the
ensemble \(\{\Fcone(\lone{\cdot},\vct Y_0)\}\) and
\(\theta_{\mathrm{Op}}=\frac{3}{4}\) is an upper decay threshold for
the ensemble \(\{\Fcone(\opnorm{\cdot},\vct X_0)\}\) by
Proposition~\ref{prop:orthogonal-bound}.  Therefore,
program~\eqref{eq:orth-sparse} succeeds with overwhelming probability
in high dimensions whenever
\begin{equation*}
    \theta_{\ell_1}(\tau) < \frac{1}{4}.
\end{equation*}
This occurs for \(\tau<0.06\); see the left panel of
Figure~\ref{fig:l1-decay-levelsets}.  We conclude
that~\eqref{eq:orth-sparse} demixes an orthogonal matrix \(\mtx
X_0\) from a matrix sparse in a random basis \(\mathcal{Q}(\mtx Y_0)\)
with high probability when no more than about \(6\%\) of the elements
of \(\mtx Y_0\) are nonzero.

\subsubsection{Low-rank and sign matrices}
\label{sec:low-rank-sign}

Fix the proportional rank \(\rho \) in \([0,1]\).  We abbreviate the ambient
dimension \(d=n^2\). For each side
length \(n\in \mathbb{N}\), choose a low-rank matrix \(\mtx X_0\in
\R^{n\times n}\) with \(\rank(\mtx X_0) = \lceil \rho n\rceil\) and a
sign matrix \(\mtx Y_0 \in \{\pm 1\}^{n\times n}\).  We use the
Schatten 1-norm as a complexity measure for rank and the matrix
\(\ell_\infty\) norm as a complexity measure for sign matrices.  Given
that \( \linf{\mtx Y_0}=1\), we consider the \CDMq\
\begin{equation}\label{eq:lowrk-sign}
  \minprog{}{\sone{\mtx X}}{\linf{\mtx Y}\le 1\;\;\text{and}\;\; \mtx X +
    \mathcal{Q}(\mtx Y) = \mtx Z_0.}
\end{equation}
We invoke Theorem~\ref{thm:sharp-thresh} to see
that~\eqref{eq:lowrk-sign} succeeds with overwhelming probability in
high dimensions whenever \(\theta_{S_1}(\rho)+\theta_{\R^{d}_+}<1\).
Here, \(\theta_{S_1}(\rho)\) is an upper decay threshold for the
ensemble \(\{\Fcone(\sone{\cdot},\mtx X_0)\} \) of feasible cones, and
\(\theta_{\R^{d}_+}\) is an upper decay threshold for the ensemble of
nonnegative orthants.  By Proposition~\ref{prop:orth-thresh}, we have
{\(\theta_{\smash{\R^{d}_+}}=\frac{1}{2}\)}, while
Proposition~\ref{prop:schatten1-bound} gives \(\theta_{S_1}(\rho) =
6\rho- 5\rho^2\).  Therefore, the \CDMq~\eqref{eq:lowrk-sign} succeeds with
overwhelming probability in high dimensions so long as
\begin{equation*}
  6 \rho - 5\rho^2 < \frac{1}{2}.
\end{equation*}
This bound is valid when \(\rho \le 0.09\).  We
conclude that~\eqref{eq:lowrk-sign} can demix a low-rank matrix
from a sign matrix in a random basis with overwhelming probability if
\(\rank(\mtx X_0) \le 0.09n\), where \(n\) is side length of \(\mtx X_0\).

\subsubsection{Low-rank and orthogonal matrices}
\label{sec:low-rank-orthogonal}

Let \(\rho\in [0,1]\) be a proportional rank parameter.  For each side
length \(n\in \mathbb{N}\), choose a low-rank matrix \(\mtx X_0 \in
\R^{n\times n}\) with \(\rank(\mtx X_0) = \lceil \rho n\rceil\) and an
orthogonal matrix \(\mtx Y_0 \in \mathsf{O}_n\).  With the usual
choice of complexity measures, the \CDMq~is
\begin{equation}
  \label{eq:lowrank-orth}
  \minprog{}{\sone{\mtx X}}{\opnorm{\mtx Y} \le 1\;\;\text{and}\;\; \mtx X +
    \mathcal{Q}(\mtx Y) = \mtx Z_0.}
\end{equation}
By Theorem~\ref{thm:sharp-thresh}, program~\eqref{eq:lowrank-orth}
succeeds with overwhelming probability in high dimensions so long as
\(\theta_{S_1}(\rho)+ \theta_{\mathrm{Op}} < 1\).
Propositions~\ref{prop:schatten1-bound}
and~\ref{prop:orthogonal-bound} imply that this occurs whenever
\begin{equation*}
  6 \rho - 3 \rho^2 < \frac{1}{4}.
\end{equation*}
For instance, it suffices that \(\rho \le 0.04\).  Therefore, the
\CDMq~\eqref{eq:lowrank-orth} can identify a superposition of a
low-rank matrix and an orthogonal matrix with overwhelming probability
in high dimensions so long as \(\rank(\mtx X_0) \le 0.04 n\),
where \(n\) is the side length of \(\mtx X_0\).

\section{Prior art and future directions}
\label{sec:past-and-future}

This work occupies a unique place in the literature on demixing.
The analysis is highly general and applies to many problems.  At the
same time, the results are sharp or nearly sharp. This final section
offers a wide-angle view of the field of demixing, from
applications to analytical techniques, with a focus on methods based
on convex optimization.  We conclude by discussing some extensions of
our current approach in the hope of encouraging further
development in this field.

\subsection{A short history of convex demixing and incoherence}
\label{sec:prior-art}

The use of convex optimization for signal demixing has a long
history. Early predecessors to morphological component analysis come
from the work of Claerbout \& Muir~\cite{Claerbout1973} and
Taylor~et~al.~\cite{Taylor1979}, where \(\ell_1\) minimization is used
to identify sparse spike trains from an observed seismic trace.

Demixing methods based on \(\ell_1\) minimization were put on a
rigorous footing in the 1980s with the work of Santosa \&
Symes~\cite{MR857796} and Donoho \& Stark~\cite{Donoho1989}.  These
results, either implicitly or explicitly, rely on incoherence in the
form of an uncertainty principle.  The work of Donoho \&
Huo~\cite{Donoho2001} formalizes the notion of incoherence. Incoherent
models, both random and deterministic, now pervade the sparse
demixing
literature~\cite{Starck2003}, \cite{Starck2005}, \cite{Elad2005}, \cite{HegBar:12}, \cite{Christoph2012}.

In the last decade, new classes of convex regularizers have been
introduced for solving inverse problems in signal processing. In
particular, the Schatten 1-norm is used for problems involving
low-rank matrices~\cite{Fazel2002}, \cite{MR2680543}.  Demixing methods
that involve the Schatten 1-norm include robust principal component
analysis~\cite{Candes2009}, \cite{Xu2010a}, \cite{XuCarSan:12}, \cite{McCoy2011} and latent
variable selection~\cite{Chandrasekaran2010a}.  Rigorous theoretical
results for these techniques typically involve a spectral
incoherence assumption, but no previous work in this area identifies phase transition behavior.

\subsection{The neighborhood of this work}
\label{sec:analys-deconv}

We take much of our inspiration from the geometric analysis of linear
inverse problems in~\cite{ChaRecPar:12}.  Indeed, the geometric
optimality condition (Lemma~\ref{lem:triv-intersect}) is a direct
generalization of a geometric
result~\cite[Prop.~2.1]{ChaRecPar:12} for linear inverse
problems.  Moreover, the Gaussian width bounds from that work prove
useful for computing the decay thresholds in this research.

A related line of work, due to Negahban
et~al.~\cite{NIPS2009_1070}, \cite{NegWai:12}, is based on the concept of
restricted strong convexity.  The results in these papers are sharp
within constant factors, but they do not yield bounds as precise as
ours.  Another general approach to demixing appears in
of~\cite{HegBar:12}, where a deterministic incoherence condition leads
to recovery guarantees, even in nonconvex settings.  The recovery
bounds available through this method are not competitive with
the guarantees we provide.

Several works also consider demixing two sparse vectors. The works~\cite{WriMa:10}, \cite{NguTra:13}  show that a nearly dense vector could be demixed from a sufficiently sparse vector, but they do not identify phase transition behavior.  Recent work~\cite{PopBraStu:13} also offers demixing guarantees for  demixing sparse vectors when the
sparsity is mildly sublinear in the dimension. Their model is similar
to our MCA formulation in Section~\ref{sec:first-appl-deconv}, but again the
results do not identify the phase transition between success and
failure. %

\subsubsection{Random geometry and convex optimization}
We now trace the use of methods from integral geometry for understanding randomized convex optimization programs.  Vershik \& Sporyshev~\cite{MR850459} use an asymptotic analysis of polytope angles to analyze the average-case behavior of the simplex method for linear programming.  The underlying formulas have their roots in the results of Ruben~\cite{MR0121713}, although some of the ideas apparently go back to the work Schl\"afli from the mid-nineteenth century---see Ruben's paper for a discussion.

The analysis of Vershik \& Sporyshev fed a line of investigation
on the expected face counts of randomly projected
polytopes~\cite{Affentranger1992a}, \cite{Brocozky1999}, a topic of
theoretical interest in combinatorial geometry.  These computations
resurfaced in convex optimization in the line of work of
Donoho \& Tanner~\cite{Donoho2004}, \cite{Donoho2005a}, \cite{Donoho2006}, \cite{Donoho2009},
  \cite{Donoho2009a}, \cite{Donoho2010a}, \cite{Donoho2010}.  These articles characterize the behavior of convex optimization methods for solving several linear inverse problems under a random measurement model.  In Appendix~\ref{sec:asympt-thresh-calc}, we leverage the asymptotic polytope angle calculations of Donoho \& Tanner to compute decay threshold for the \(\ell_1\) norm at sparse vectors.

This asymptotic polytope angle approach also yields stability
guarantees for basis pursuit~\cite{6034753}.  Furthermore, it has been
used to establish that iteratively reweighted basis pursuit can
provide strictly stronger guarantees than standard basis
pursuit~\cite{5495210}, \cite{Khajehnejad2010}, \cite{Khajehnejad2011}.

Our approach to random geometry differs from these earlier works
because it starts with the modern theory of spherical integral
geometry.  Previous research was based on an older theory of polytope
angles~\cite{Grunbaum1967}, \cite{Grunbaum1968}, \cite{McMullen1975}.  Spherical
integral geometry reached its current state of development in the
dissertation~\cite{Glasauer1995}, \cite{Glasauer1996}.  Chapter 6.5
of~\cite{Schneider2008} and the notes therein summarize this research.
We also draw on insights from the thesis~\cite{Amelunxen2011}.

\subsection{Conclusions and future directions}
\label{sec:concl-future-direct}

The results in this work demonstrate the power of spherical integral
geometry in the context of demixing.  Our bounds are often tight,
and they are broadly applicable.  This approach raises many questions
worth further attention.  We conclude with a list of directions for future
work. During the period that our original manuscript was under review, several of these areas have seen significant progress.  We augment our original list of open problems with a summary of progress made in the interim.

\begin{description}\setlength{\itemsep}{0pt}
\item[Tight results for Lagrangian demixing.]  The Lagrange
  penalized demixing method~\eqref{eq:gen-lagrange} is important
  because it requires less knowledge about the unobserved vectors
  \((\vct x_0,\vct y_0)\) than the corresponding constrained
  method~\eqref{eq:genconv}.  The results in this work give
  information regarding the potential for, and the limits of, the
  penalized demixing approach~\eqref{eq:genconv}.  Nevertheless, a
  precise analysis of the penalized problem~\eqref{eq:gen-lagrange}
  and its dependence on the penalty parameter \(\lambda\) would have
  real practical value.

  \emph{While a sharp phase transition characterization for the Lagrange demixing problem~\eqref{eq:gen-lagrange} remains open, the recent work~\cite{FoyMac:13} offers theoretical guarantees and explicit choices of Lagrange parameters.  }
\item[Multiple demixing.] It would be interesting to study
  demixing problems involving more than two structured vectors.

  \emph{A study of this problem appears in~\cite{WriGanMin:13}.  The present authors provide sharp phase transition characterizations for demixing an arbitrary number of signals in~\cite{McCTro:13a}.}
\item[Spherical intrinsic volumes for more cones.] Computation of
  additional decay thresholds will provide new bounds for
  \CDMq s.  The sharpest decay thresholds appear to require
  formulas for spherical intrinsic volumes.  For instance, an
  asymptotic analysis of the spherical intrinsic volumes for feasible
  cones of the Schatten 1-norm would provide sharp recovery results
  for low-rank matrix demixing problems.  Amelunxen \& B\"urgisser have made some
  recent progress in this direction by developing a formula for the
  spherical intrinsic volumes for the semidefinite
  cone~\cite{AmeBur:12a}.%
\item[Log-concavity of spherical intrinsic volumes.] 
  B\"urgisser \&   Amelunxen~\cite[Conj.~2.19]{Burgisser2010} conjecture that
  the sequence of spherical intrinsic volumes is log-concave.  This  conjecture is closely related to the question of whether the upper  and lower decay thresholds match.  %
  
  \emph{In recent work by the present authors and collaborators,  the intrinsic volumes are shown to have a nontrivial log-concave upper bound~\cite[Sec.~6.1]{AmeLotMcC:13}. While this result implies that the upper and lower decay thresholds are often equal (Section~\ref{sec:thresh-phen-conj}), the log-concavity conjecture remains open.}
\item[Extensions to more general probability measures.] The
  analysis in this work focuses on a specific random model.  It
  would be interesting to incorporate more general probability
  measures into our framework.  This may be a difficult problem; by
  the results of Section~\ref{sec:relat-line-inverse}, this question
  is closely related to the observed universality phenomenon in basis
  pursuit~\cite{Donoho2009a}.

  \emph{Bayati et al.~\cite{BayLelMon:12} provide a rigorous version the universality property for basis pursuit observed in~\cite{Donoho2009a}. It remains unclear whether  their methods adapt to demixing problems considered here.}
\end{description}

\appendix

\section{Equivalence of the constrained and penalized
  methods}
\label{sec:lagrange-parameter}

This appendix provides a geometric account of the equivalence between
the constrained~\eqref{eq:genconv} and
penalized~\eqref{eq:gen-lagrange} \CDMq s.  The
results in this section let us interpret our conditions for the
success of the constrained demixing method~\eqref{eq:genconv} as
limits on, and opportunities for, the Lagrange demixing
method~\eqref{eq:gen-lagrange}.

We begin with the following well-known result; it holds without any
technical restrictions. We omit the demonstration, which is an easy
exercise in proof by contradiction. (See
also~\cite[Cor.~28.1.1]{Rockafellar1970}.)
\begin{proposition}\label{prop:lag-to-const}
  Suppose the Lagrange problem~\eqref{eq:gen-lagrange} succeeds for
  some value \mbox{\(\lambda >0\)}.  Then~\eqref{eq:genconv} succeeds.
\end{proposition}

Before stating a partial converse to
Proposition~\ref{prop:lag-to-const}, we require a technical
definition.  We say that a proper convex function \(f\) is
\emph{typical at \(\vct x\)} if \(f\) is subdifferentiable at \(\vct
x\) but does not achieve its minimum at \(\vct x\).  With this
technical condition in place, we have the following complement to
Proposition~\ref{prop:lag-to-const}.

\begin{proposition}\label{lem:const-to-lag}
  Suppose \(f\) is typical at \(\vct x_0\) and \(g\) is typical at
  \(\vct y_0\).  If the constrained method~\eqref{eq:genconv}
  succeeds, then there exists a parameter \(\lambda > 0\) such that \((\vct x_0,\vct y_0)\) is an optimal point for the Lagrange 
  method~\eqref{eq:gen-lagrange}.
\end{proposition}
\noindent
Note that there is a subtlety here: the Lagrange program may have strictly more optimal points than the corresponding constrained problem even for the best choice of \(\lambda\), so that we cannot guarantee that \((\vct x_0,\vct y_0)\) is the unique optimum.  See~\cite[Sec.~28]{Rockafellar1970} for more details.
\begin{proof}[Proof of Proposition~\ref{lem:const-to-lag}]
  The key idea is the construction of a subgradient that
  certifies the optimality of the pair \((\vct x_0,\vct y_0)\) for the Lagrange
  penalized problem~\eqref{eq:gen-lagrange} for an appropriate choice
  of parameter \(\lambda\).  As with many results in
  convex analysis, a separating hyperplane plays an important role.

  By Lemma~\ref{lem:triv-intersect}, the constrained
  problem~\eqref{eq:genconv} succeeds if and only if \(\Fcone(f,\vct
  x_0) \cap -\mtx Q \Fcone(g,\vct y_0) = \{\mathbf{0}\}\).  The
  trivial intersection of the feasible cones implies that there exists
  a hyperplane that separates these cones. (This fact is a special
  case of the Hahn--Banach separation theorem for convex cones due to
  Klee~\cite{KleeJr1955}.)  In other words, there exists some vector
  \(\vct u \ne \mathbf{0}\) such that
  \begin{equation*}
    \langle \vct u, \vct x \rangle \le 0 \text{ for all } \vct x \in
    \Fcone(f,\vct x_0),
  \end{equation*}
  and moreover
  \begin{equation*}
    \langle \vct u, \vct y \rangle \ge 0 \text{ for all } \vct y \in
    -\mtx Q \Fcone(g,\vct y_0).
  \end{equation*}
  In the language of polar cones, the first separation inequality is
 simply the statement that \(\vct u \in \Fcone(f,\vct x_0)^\circ\),
  while the second inequality is equivalent to \(\mtx Q^\adj \vct u \in
  \Fcone(g,\vct y_0)^\circ\).

  We will now show that \(\vct u\) generates a subgradient optimality
  certificate for the point \((\vct x_0,\vct y_0)\) in
  problem~\eqref{eq:gen-lagrange} for an appropriate choice of parameter
  \(\lambda >0\). We denote the subdifferential map by \(\partial\).  

  At this point, we invoke our technical assumption.  Since \(f\) is
  typical at \(\vct x_0\), the polar to the feasible cone is generated
  by the subdifferential of \(f\) at \(\vct
  x_0\)~\cite[Thm~23.7]{Rockafellar1970}.  In particular, there exists
  a number \(\lambda_f\ge 0\) such that \(\vct u \in \lambda_f \partial
  f(\vct x_0)\).  In fact, the stronger inequality \(\lambda_f>0\)
  holds because \(\vct u\ne \zerovct\).  For the same reason, there
  exists a number \(\lambda_g >0\) such that \(\mtx Q^\adj \vct u \in
  \lambda_g \partial g(\vct y_0)\).

  Define \(h(\vct x) \defeq \lambda_f f(\vct x) + \lambda_g g(\mtx
  Q^\adj(\vct z_0 - x))\).  By standard transformation rules for
  subdifferentials~\cite[Thms.~23.8,~23.9]{Rockafellar1970}, we have 
  \begin{equation*}
    \partial h(\vct x_0) \supset 
    \lambda_f \partial f(\vct x_0) -\lambda_g \mtx Q \partial
    g(\vct y_0),
  \end{equation*}
  where \(A - B \defeq A + (-B)\) is the Minkowski sum of the sets
  \(A\) and \(-B\).  Since \(\vct u\in \lambda_f \partial f(\vct
  x_0)\) and \(\vct u \in \lambda_g \mtx Q \partial g(\vct y_0)\), we
  see \(\mathbf{0} \in \partial h(\vct x_0)\). By the definition of
  subgradients, \(\vct x_0\) is a global minimizer of \(h\).
  Introducing the variable \(\vct y = \mtx Q^\adj(\vct z_0-\vct x)\),
  it follows that \((\vct x_0,\vct y_0)\) is a global minimizer of
  \begin{equation*}
    \minprog{}{f(\vct x) +\frac{\lambda_g}{\lambda_f} g(\vct y)}{\vct
      x + \mtx Q \vct y = \vct z_0.} 
  \end{equation*}
  This is Lagrange problem~\eqref{eq:gen-lagrange} with the parameter
  \(\lambda = \lambda_g/\lambda_f>0\), so we have the result.
\end{proof}

\section{Regions of failure and uniform guarantees}
\label{sec:auxilliary-results}
We now present the proofs of the results of
Section~\ref{sec:find-weak-thresh} concerning regions of failure
(Theorem~\ref{thm:failure-thresh}) and strong demixing guarantees
(Theorem~\ref{thm:strong-thm}) for the \CDMq . These demonstrations
closely follow the pattern laid down by the proof of
Theorem~\ref{thm:sharp-thresh}.

\subsection{Regions of failure: The proof of
  Theorem~\ref{thm:failure-thresh}}
\label{sec:region-fail-proof}
We first state an analog of Theorem~\ref{thm:asympt-thresh-gen}.  As
usual, \(\Index\) is an infinite set of indices.
\begin{theorem}\label{thm:asympt-lower-bounds}
  Let \(\{K^{(d)}\subset \R^d\mid d\in \Index\}\) and \(\{\tilde
  K^{(d)}\subset \R^d \mid d \in \Index\}\) be two ensembles of closed
  convex cones with lower decay thresholds \(\kappa_\star\) and
  \(\tilde \kappa_\star\).  If~ \(\kappa_\star + \tilde \kappa_\star >
  1\), then there exists an  \(\varepsilon >0\) such that~
  \(\mathbb{P}\bigl\{ K^{(d)} \cap \mtx Q \tilde K^{(d)} \ne
  \{\mathbf{0}\}\bigr\} \ge 1- \mathrm{e}^{-\varepsilon d} \) for all
  sufficiently large \(d\).
\end{theorem}
Theorem~\ref{thm:failure-thresh} follows  from
Theorem~\ref{thm:asympt-lower-bounds} in the same way that
Theorem~\ref{thm:sharp-thresh} follows from
Theorem~\ref{thm:asympt-thresh-gen} with one additional technical
point regarding closure conditions. 
\begin{proof}[Proof of Theorem~\ref{thm:failure-thresh} from
  Theorem~\ref{thm:asympt-lower-bounds}]
  By the assumptions in Theorem~\ref{thm:failure-thresh}, the
  ensembles \(\bigl\{\overline{\Fcone}(f^{(d)}, \vct
  x_0^{(d)})\bigr\}\) and \(\bigl\{-\overline{ \Fcone}(g^{(d)},\vct
  y_0^{(d)})\bigr\}\) of closed cones satisfy the hypotheses of
  Theorem~\ref{thm:asympt-lower-bounds}.  Therefore, there is an \(\eps >0\) such that the closure of the feasible cones
  have \emph{non}trivial intersection with probability at least
  \(1-\mathrm{e}^{- \varepsilon d}\), for all large enough \(d\).

  It follows from Remark~\ref{rem:touch-prob} that the probability of
  the event \(\overline{\Fcone}(f^{(d)}, \vct x_0^{(d)})\cap -\mtx Q
  \overline{\Fcone}(g^{(d)}, \vct y_0^{(d)}) \ne \{\zerovct\}\) is
  equal to the probability of the event \(\Fcone(f^{(d)}, \vct
  x_0^{(d)}) \cap - \mtx Q \Fcone(g^{(d)}, \vct y_0^{(d)}) \ne\{\zerovct\}\).
  Applying the geometric optimality condition of
  Lemma~\ref{lem:triv-intersect} immediately implies
  that~\eqref{eq:genconv} fails with probability at least
  \(1-\mathrm{e}^{-\varepsilon d}\).
\end{proof}

The proof of Theorem~\ref{thm:asympt-lower-bounds} requires an
additional fact  concerning spherical intrinsic volumes.
\begin{fact}[Spherical Gauss--Bonnet formula~\protect{\cite[P.~258]{Schneider2008}}] \label{fact:spherical-gauss-bonnet} For any
  closed convex cone \(K\subset \R^d\) that is not a subspace,
  we have
  \begin{equation*}
    \sum_{\substack{i=-1\\ i \text{ even}}}^{d-1} v_i(K) = 
    \sum_{\substack{i=-1\\ i \text{ odd}}}^{d-1} v_i(K) = \frac{1}{2}.
  \end{equation*}
\end{fact}
In the proof below, the Gauss--Bonnet formula is crucial for dealing
with the parity term \((1+(-1)^k)\) that arises in the spherical
kinematic formula~\eqref{eq:sphere-kin}.

\begin{proof}[Proof of Theorem~\ref{thm:asympt-lower-bounds}]
  Since the Gauss--Bonnet formula only applies to cones that are not
  subspaces, we split the demonstration into three cases: neither
  ensemble \(\{K^{(d)}\}\) nor \(\{\tilde K^{(d)}\}\) contains a
  subspace, one ensemble consists of subspaces, or both ensembles
  consist of subspaces.  We assume without loss that each case holds
  for every dimension \(d\in \Index\); the proof extends to the
  general case by considering subsequences where only a
  single case applies.

  We drop the superscript \(d\) for clarity.  Assume first that
  neither \(K\) nor \(\tilde K\) is a subspace.  Let \(\varpi(k) =
  (1+(-1)^k)\) be the parity term in the spherical kinematic
  formula~\eqref{eq:sphere-kin}.  Changing the order of summation in the
  spherical kinematic formula, we find
  \begin{equation*}
    P := \mathbb{P}\left\{K \cap \mtx Q \tilde K \ne
      \{\mathbf{0}\} \right\} = \sum_{i=0}^{d-1} v_i(K)
    \sum_{k=d-i-1}^{d-1} \varpi(k-d+1+i) v_k(\tilde K).
  \end{equation*}
  Let \(\kappa < \kappa_\star\) and \(\tilde \kappa < \tilde
  \kappa_\star\) with \(\kappa+ \tilde \kappa > 1\); such scalars
  exist because \(\kappa_\star + \tilde \kappa_\star >1\).  By
  positivity of the spherical intrinsic volumes
  (Fact~\ref{fact:obvious}.\ref{item:obv-positivity}), we have
  \begin{equation}\label{eq:double-sum}
    P \ge \sum_{i= \lceil \kappa d \rceil + 1}^{d-1} v_i(K)
    \sum_{k=d-i-1}^{d-1} \varpi(k-d+1+i) v_k(\tilde K)
 \end{equation}
 We will see that the inner sum above is very close to one.  Indeed,
 \begin{equation}\label{eq:inner-sum-discrep}
   \sum_{k=d-i-1}^{d-1} \varpi(k-d+1+i) v_k(\tilde K) = 
   \sum_{k=-1}^{d-1} \varpi(k-d+1+i) v_k(\tilde K)  - \tilde{\xi}_{i}=1-\tilde{\xi}_i,
 \end{equation}
 where \(\tilde{\xi}_i\) is a discrepancy term
 (see~\eqref{eq:discrepancy-obvious} below).  The second equality
 follows by the spherical Gauss--Bonnet formula
 (Fact~\ref{fact:spherical-gauss-bonnet}), and the assumption that
 \(\tilde K\) is not a subspace.

 We now bound the discrepancy term \(\tilde {\xi}_i\) uniformly over
 \(i \ge \lceil\kappa d\rceil+1\).  Since \(\kappa+\tilde \kappa>1\),
 for any \(i \ge \lceil \kappa d \rceil+1\) we have \(d-2-i \le
 \lceil\tilde \kappa d\rceil\). By definition of the lower decay
 threshold, we see that the discrepancy term must be small: for any
 \(i \ge \lceil \kappa d \rceil+1\),
 \begin{equation}\label{eq:discrepancy-obvious}
   \tilde{\xi}_i = \sum_{k=-1}^{d - 2-i} \varpi(k-d+1+i) v_k(\tilde K)
   \le 2 \sum_{k=-1}^{\lceil \tilde \kappa d\rceil } v_k(\tilde
   K)  \le 2(d-1)\mathrm{e}^{-\varepsilon' d},
 \end{equation}
 for some \(\varepsilon' >0\) and all sufficiently large
 \(d\). Applying~\eqref{eq:inner-sum-discrep}
 and~\eqref{eq:discrepancy-obvious} to~\eqref{eq:double-sum}, we find
 \begin{equation}\label{eq:prob-lower-2}
   P \ge \sum_{i=\lceil \kappa d \rceil + 1}^{d-1} v_i(K)\bigl( 1-
   2(d-1)\mathrm{e}^{-\varepsilon' d}\bigr) \ge \left(\sum_{i=\lceil \kappa d \rceil
     + 1}^{d-1} v_i(K) \right)- 2(d-1) \mathrm{e}^{-\varepsilon' d},
 \end{equation}
 where the second inequality follows from Fact~\ref{fact:obvious}: the
 spherical intrinsic volumes are positive and sum to one.
 We now reindex the sum on the right-hand side
 of~\eqref{eq:prob-lower-2} over \(i = -1,0,\dotsc, d-1\)
 with only exponentially small loss:
 \begin{equation*}
   \sum_{i=\lceil \kappa d \rceil + 1}^{d-1} v_i(K^{(d)})  =
   \sum_{i=-1}^{d-1} v_i(K) - \xi %
 \end{equation*}
 where the discrepancy \(\xi\) satisfies 
 \begin{equation*}
    \xi = \sum_{i = -1}^{\lceil \kappa d \rceil } v_k(K) \le
   (d-1) \mathrm{e}^{-\varepsilon'' d},
 \end{equation*}
 for some \(\eps''>0\) and all sufficiently large \(d\) by definition
 of the lower decay threshold.  Applying these observations
 to~\eqref{eq:prob-lower-2}, we deduce that
 \begin{equation*}
   P \ge  \sum_{i=-1}^{d-1} v_i(K)   - (d-1)\Bigl( \mathrm{e}^{-\varepsilon' d} +
   \mathrm{e}^{-\varepsilon''d}\Bigr)\ge  \sum_{i=-1}^{d-1} v_i(K) - \mathrm{e}^{-\varepsilon d}
 \end{equation*}
 for some \(\varepsilon >0\) and all sufficiently large \(d\).  Since
 \(\sum_{i=-1}^{d-1} v_i(K)= 1\) by
 Fact~\ref{fact:obvious}.\ref{item:obv-unity}, this shows the result
 when both \(K\) and \(\tilde K\) are not subspaces, completing the
 first case.

 For the second case, suppose that only one of the cones is a
 subspace.  Without loss, we may  assume \(\tilde K\) is the subspace
 by the symmetry of the spherical kinematic formula (see
 Remark~\ref{rem:symmetry}).  Denote the dimension of the subspace
 \(\tilde K\) by \(\tilde n\defeq \mathrm{dim}(\tilde K)\), and take
 parameters \(\kappa\) and \(\tilde \kappa\) as above.

 By Proposition~\ref{prop:subspace-iv}, the spherical intrinsic
 volumes of \(\tilde K\) are given by \(v_i(\tilde K) =
 \delta_{i,\tilde n-1}\).  Inserting this Kronecker~\(\delta\) into the
 spherical kinematic formula~\eqref{eq:sphere-kin} and simplifying the
 resulting expression, we find the probability of interest is given by
 \begin{equation*}\label{eq:one-subspace-copy}
   P \defeq \mathbb{P}\left\{K \cap \mtx Q \tilde K \ne
      \{\mathbf{0}\} \right\}= \sum_{k=d-\tilde n}^{d-1} \varpi(k+\tilde n-d) v_k(K).
 \end{equation*}
 Reindexing the sum over \(k=-1,0,\dotsc,d-1\), we see
 \begin{equation}\label{eq:P-one-subspace}
   P= \sum_{k=-1}^{d-1}\varpi(k+\tilde n-d) v_k(K)-
   \sum_{i=-1}^{d-\tilde n-1}\varpi(k+\tilde n-d) v_i(K)
   = 1- \sum_{i=-1}^{d-\tilde n-1}\varpi(k+\tilde n-d)v_i(K)
 \end{equation}
 where the second equality holds by the spherical Gauss--Bonnet
 formula (Fact~\ref{fact:spherical-gauss-bonnet}).

 We now show that \(\tilde n\) is relatively large. The definition of
 the lower decay threshold implies that there exists an \(\eps' >0\)
 such that
 \begin{equation*}
   v_i(\tilde K)= \delta_{i,\tilde n-1} \le \econst^{-\eps' d} \;\;\;\text{for all}\;\;
   i \le \lceil \tilde \kappa d\rceil
 \end{equation*}
 when \(d\) is sufficiently large.  This inequality cannot
 hold if \(\tilde n -1\le \lceil \tilde \kappa d\rceil \), so we deduce that
 \(\tilde n \ge \tilde \kappa d \) for all sufficiently large \(d\).

 Since \(\tilde n\ge \tilde\kappa d\) and \(\kappa +\tilde \kappa > 1\), we
 must have \(d-\tilde n-1 < \lceil \kappa d\rceil\) for all sufficiently
 large \(d\).
 Applying the definition of the lower decay threshold,
 we find the sum on the right-hand side of~\eqref{eq:P-one-subspace}
 is exponentially small: there exists an \(\eps''>0\) such
 that
 \begin{equation*}
\sum_{i=-1}^{d-\tilde n-1}\varpi(k+\tilde n - d)v_i(K)\le 2(d-1)\econst^{-\varepsilon'' d}
\end{equation*}
 for all sufficiently large \(d\).  The result for the second case
 follows immediately.

 Finally, we consider the case when both of the cones are subspaces.
 Suppose \(K\) has dimension \(n\), while \(\tilde K\) has dimension
 \(\tilde n\), and let \(\kappa\), \(\tilde \kappa\) be as above.  As
 in the second case, we find that \(n\ge \kappa d\) and \(\tilde n\ge
 \tilde \kappa d\) when \(d\) is sufficiently large.  Then the
 inequality \(\kappa_\star + \tilde \kappa_\star >1\) implies that
 \(n+\tilde n > d\), that is, the sum of the dimensions of the
 subspaces is larger than the ambient dimension.  A standard fact from
 linear algebra implies \(K \cap \mtx Q \tilde K \ne \{\vct 0\}\) for
 any unitary \(\mtx Q\)---in other words, for all \(d\) large enough,
 the probability of nontrivial intersection is one.  This
 completes the third case, and we are done.
\end{proof}

\subsection{Proof of the strong guarantees of
  Theorem~\ref{thm:strong-thm}}
\label{sec:proof-strong-theorem}

\begin{proof}[Proof of Theorem~\ref{thm:strong-thm}]
  For clarity, we drop the superscript \(d\) in this proof.
  We begin with the union bound: the probability of
  interest \(P\) is bounded above by
  \begin{equation}\label{eq:union-bd}
   P \defeq \mathbb{P}\left\{  K\cap \mtx Q 
      \tilde K \ne \{\vct 0\} \,\text{ for any }\, K \in
      \mathcal{K},\; \tilde{K} \in
      \tilde{\mathcal{K}}\right\}  \le |\mathcal{K}|
    \cdot |\tilde{\mathcal K}| \cdot \mathbb{P}\left\{ K
      \cap \mtx Q  \tilde K \ne \{\vct 0\}\right\}.
  \end{equation}
  From here, the proof closely parallels the proof of
  Theorem~\ref{thm:asympt-thresh-gen}, so we compress the
  demonstration.  We consider two cases, one where at least one cone
  is a subspace, the other where neither cone is a subspace; the
  result extends to the mixed case by considering subsequences.

  Suppose first that at least one cone is a subspace.  Let
  \(\theta>\theta_\star\) and \(\tilde \theta > \tilde \theta_\star\)
  with \(\theta+\tilde \theta <1\).  We bound the probability on the
  right-hand side of~\eqref{eq:union-bd} by
  \begin{equation*}
    \frac{1}{2}\mathbb{P}\left\{ K
      \cap \mtx Q  \tilde K \ne \{\vct 0\}\right\} 
    \le \Sigma_1+\Sigma_2+\Sigma_3+\Sigma_4,
  \end{equation*}
  where the \(\Sigma_i\) are given in~\eqref{eq:sigmas}.  The fact
  that \(\theta + \tilde \theta<1\) implies that \(\Sigma_1 = 0\) for
  sufficiently large \(d\), as in the proof of
  Theorem~\ref{thm:asympt-thresh-gen}.  Since \(\theta>
  \theta_\star\), the definition of the upper decay threshold at level
  \(\psi+\tilde \psi\) implies
  \begin{equation*}
    \Sigma_2 \le \sum_{i=\lceil \theta d \rceil + 1}^{d-1} v_i(K)
    \le (d-1) \mathrm{e}^{-d (\psi + \tilde \psi+ \varepsilon')}
  \end{equation*}
  for some \(\varepsilon'>0\) and all sufficiently large \(d\).  With
  analogous reasoning, we find similar exponential bounds for
  \(\Sigma_3\) and \(\Sigma_4\):
  \begin{equation*}
    \Sigma_3 \le (d-1) \mathrm{e}^{-d(\psi + \tilde \psi +
      \varepsilon'')} , \quad \Sigma_4 \le (d-1) \mathrm{e}^{-d( \psi
      + \tilde \psi +\varepsilon''')}
  \end{equation*}
  again for positive \(\varepsilon'',\varepsilon'''\) and all
  sufficiently large \(d\).  Summing these inequalities and taking
  \(d\) sufficiently large gives
  \begin{equation*}
       \mathbb{P}\left\{ K
      \cap \mtx Q  \tilde K \ne \{\vct 0\}\right\}  \le
    \mathrm{e}^{-d(\psi + \tilde \psi + \hat \varepsilon)}
  \end{equation*}
  for some \(\hat \varepsilon>0\).  The claim then follows from our
  exponential upper bound on the growth of \(|\mathcal{K}|\) and
  \(|\tilde{\mathcal{K}}|\) with \(\eta = \hat \varepsilon/2\):
  \begin{equation*}
    P \le 2\cdot |\mathcal{K}| \cdot |\tilde{\mathcal{K}}|\cdot \mathrm{e}^{ d (\psi + \tilde \psi + \hat \varepsilon/2) -
      d(\psi +\tilde \psi + \hat \varepsilon)} =
    2\mathrm{e}^{-\frac{\hat\varepsilon}{2}d}.
  \end{equation*}
  Taking \(\varepsilon = \hat \varepsilon/4\) and \(d\) sufficiently
  large gives the claim in the first case. 

  Now consider the case where both cones are subspaces, and let
  \(n\defeq\dim(K) \) and \(\tilde n\defeq \dim(\tilde K)\).  Take
  parameters \(\theta > \theta_\star\), and \(\tilde \theta>\tilde
  \theta_\star\), such that \(\theta+\tilde\theta <1\).  As in the
  proof of Theorem~\ref{thm:asympt-thresh-gen}, the
  Kronecker~\(\delta\) expression for the intrinsic volumes of the
  subspaces \(K\) and \(\tilde K\) given by
  Proposition~\ref{prop:subspace-iv}, combined with the definition of
  the upper decay threshold, reveals that \(n\le \lceil \theta
  d\rceil\) and \(\tilde n \le \lceil \tilde \theta d\rceil\) for all
  sufficiently large \(d\).  The fact that \(\theta+\tilde \theta <1\)
  implies \(n+\tilde n < d\) for all sufficiently large \(d\).  Since
  randomly oriented subspaces are almost always in general position,
  the probability that \(K\cap \mtx Q K\ne \{\zerovct\}\) is zero.
  This is the second case, so we are done.
\end{proof}

\section{Decay thresholds for feasible cones of the
  \protect{$\ell_1 $} norm}
\label{sec:asympt-thresh-calc}

This section describes how we compute decay thresholds for the
feasible cone of the \(\ell_1\) at sparse vectors. The 
polytope angle calculations appearing in~\cite{Donoho2006} form an
important part of this computation.  For convenient comparisons,
Table~\ref{tab:notation-translation} provides a map between our
notation and that of the reference.

Fix a sparsity parameter \(\tau\in [0,1]\), and let \(\Index\) be an
infinite set of indices. For each dimension \(d \in \Index\), we
define a vector \(\vct x^{(d)} \in \R^d\) such that \(\nnz(\vct x^{(d)}) = \lceil
\tau d\rceil\). The following results describes the behavior of the
spherical intrinsic volumes of the feasible cone
\(\Fcone(\lone{\cdot},\vct x^{(d)})\) in terms of the sparsity \(\tau\) and the
normalized index \(\theta = i/d\) when \(d\) is large.

\begin{lemma}\label{thm:l1-intrinsic-upper-bound}
  Consider the ensemble above.  There exists a function \(\Psi_{\mathrm{total}}
  \) such that, for every \(\varepsilon > 0\) and all sufficiently large
  \(d\in \mathcal{D}\), we have
  \begin{equation}\label{eq:l1-exponent-psis}
    \frac{1}{d}\log\Bigl(v_{\lceil\theta
      d\rceil}\bigl(\Fcone(\lone{\cdot},\vct x^{(d)})\bigr)\Bigr) \le
    \Psi_{\mathrm{total}}(\theta,\tau) + \varepsilon
 \end{equation}
  for all \(\theta \in [\tau,1]\), and 
  \begin{equation}\label{eq:l1-exponent-zeros}
    v_{\lceil\theta
      d\rceil}\bigl(\Fcone(\lone{\cdot},\vct x^{(d)})\bigr) = 0
  \end{equation}
  for \(\theta \in [0, \tau)\).
\end{lemma}

We discuss the definition and computation of the normalized exponent
\(\Psi_{\mathrm{total}}\) in Section~\ref{sec:computing-psis}.  This
function provides decay thresholds for the ensemble
\(\bigl\{\Fcone(\lone{\cdot},\vct x^{(d)})\mid d \in \Index\bigr\}\) in the same
way that the limit~\eqref{eq:binom-entropy-limit} provides decay
thresholds for the ensemble of orthants.  See
Section~\ref{sec:comp-asympt-expon} for details.

\begin{table}[t]
  \centering
  \renewcommand{\arraystretch}{1.1}
  \caption{\textsl{Notation translation between 
      this work
      and~\cite{Donoho2006}.} Note that in the reference, 
    \(\Psi_{\mathrm{net}}\) is defined for
    three arguments, but only depends on two parameters,
    namely \(\nu\) and \(\rho\delta\).} 
  \label{tab:notation-translation}
 \begin{tabular}{p{3.5cm}cc}
   \toprule
   \textbf{Parameter} & \textbf{Our notation} 
   & \textbf{Notation of~\cite{Donoho2006}} \\
   \midrule
   Sparsity ratio & \(\tau\) & \(\rho \delta\) \\ 
   \begin{tabular}{@{}p{3.5cm}@{}}
     Ratio of measurements to ambient dimension 
   \end{tabular}
   & \(\sigma\) &
   \(\delta\)
\\
   Undersampling ratio & \(\tau / \sigma\) & \( \rho\) \\
   Normalized index & \(\theta \)  &\(\nu\)\\
   Internal exponent & \(\Psi_{\mathrm{int}}(\theta,\tau)\) 
   &\(\Psi_{\mathrm{int}}(\nu,\rho \delta)\) \\
   External exponent 
   & \(\Psi_{\mathrm{ext}}(\theta,\tau)\) &
   \(\Psi_{\mathrm{ext}}(\nu,\rho \delta)\) \\
   Total exponent & \(\Psi_{\mathrm{total}}(\theta,\tau)\) &
   \(\Psi_{\mathrm{net}}(\nu,\rho,\delta) - \Psi_{\mathrm{face}}(\rho\delta)\)\\
   Net exponent & \(\Psi_{\mathrm{total}}(\theta,\tau) + E(\tau)\) &
   \(\Psi_{\mathrm{net}}(\nu, \rho, \delta)\) \\ 
   \bottomrule
  \end{tabular}
\end{table}

\begin{proof}[Proof of Lemma~\ref{thm:l1-intrinsic-upper-bound}]
  We leave the dependence on the dimension implicit for clarity.
  Define \(k\defeq \lceil \tau d\rceil\).

  We first show that~\eqref{eq:l1-exponent-psis} holds. The proof
  relies on an expression for spherical intrinsic volumes in terms of
  polytope angles.  For a face \(F\) of a polytope \(P\), we define
  \(\beta(F,P)\) as the internal angle of \(P\) at \(F\) and
  \(\gamma(F,P)\) as the external angle of \(P\) at \(F\)
  (see~\cite[Chapter~14]{Grunbaum1967} for the definitions).  The
  following is an important alternative characterization of the
  spherical intrinsic volumes in terms of these angles.
  \begin{fact}[\protect{\cite[Equation~(6.50)]{Schneider2008}}]
    \label{fact:angle-sum-intrinsic-volumes}
    Let \(K\) be a  polyhedral cone, and let
    \(\mathfrak{F}_i(K)\) be the set of all \(i\)-dimensional faces of
    \(K\).  Then
    \begin{equation}\label{eq:polytope-angles}
      v_i(K) = \sum_{F \in \mathfrak{F}_{i+1}(K)} \beta(\vct 0,F)
      \gamma(F, K).
    \end{equation}
  \end{fact}

  We now specialize Fact~\ref{fact:angle-sum-intrinsic-volumes} to the
  case where \(K= \Fcone(\lone{\cdot},\vct x)\).  Define the sublevel
  set \(S \defeq \{\vct w \mid \lone{\vct w}\le \lone{\vct
    x}\}\).   Recalling that our assumption
  \(\tau >0\) implies \(\lone{\vct x}>0\), we have
  \begin{equation}\label{eq:conic-def-of-l1-feas}
    \Fcone(\lone{\cdot},\vct x) = \mathrm{cone}(S-\{\vct x\}) =
    \mathrm{cone}\bigl(C-\{\vct x/\lone{\vct x}\}\bigr),
  \end{equation}
  where \(C \defeq \{\vct w\mid \lone{\vct w} \le 1\}\) is the standard
  crosspolytope.  The fact that \(\vct x\) is \(k\)-sparse is
  equivalent to the statement that \(\vct x/\lone{\vct x}\) lies in the
  relative interior of a \((k-1)\)-dimensional face of the
  crosspolytope \(C\).

  Relationship~\eqref{eq:conic-def-of-l1-feas} implies that there is a
  one-to-one correspondence between the \(i\)-dimensional faces of
  \(\Fcone(\lone{\cdot},\vct x)\) and the \(i\)-dimensional faces of
  the crosspolytope \(C\) that contain \(\vct x/\lone{\vct x}\).
  Since the internal and external angles only depend on the local
  structure of a given polytope, we find that for every nonempty face
  \(F\) of \(\Fcone(\lone{\cdot},\vct x)\) the internal and external
  angles satisfy
  \begin{equation*}
    \beta(\vct 0, F) =\beta(\vct x, \tilde F) \;\;\text{and}\;\;
    \gamma(F, K) = \gamma(\tilde F, C),
  \end{equation*}
  where \(\tilde F\) is the face of the crosspolytope \(C\) naturally
  corresponding to the face \(F\) of the feasible cone
  \(\Fcone(\lone{\cdot},\vct x)\).

  A number of important relationships due to B\"or\"oczky \&
  Henk~\cite{Brocozky1999} for faces of the crosspolytope \(C\) are
  conveniently collected in~\cite[Section~3.3]{Donoho2006}.  In
  particular, we will need the following two facts:
  \begin{enumerate}
  \item There are \(2^{i-k+2}\binom{d-k}{i-k+2}\) faces of \(C\) of
    dimension \((i+1) \ge (k-1)\) containing a given
    \((k-1)\)-dimensional face of \(C\), and
  \item The high degree of
    symmetry of the crosspolytope ensures that the internal and
    external angles at these faces depend only on the dimensional
    parameters \(k\) and \(i\).
  \end{enumerate}

  Applying the observations above to equation~\eqref{eq:polytope-angles},
  we find
  \begin{equation}\label{eq:int-vol-explicit}
    v_i\bigl(\Fcone(\lone{\cdot},\vct x)\bigr) = 2^{i-k+2}
    \binom{d-k}{i-k+2}
    \beta(T_{k-1}, T_{i+1})  \gamma( \tilde{F}_{i+1}, C)
  \end{equation}
  for \(i = \{k-1,\dotsc,d-1\}\).
  Above, \(\tilde F_{i+1}\) is any \((i+1)\)-dimensional face of the
  crosspolytope \(C\), and \(T_j\) is the \(j\)-dimensional regular
  simplex.

  The internal and external angles in~\eqref{eq:int-vol-explicit} have
  explicit expressions due to~\cite{Brocozky1999} and the
  work of Ruben~\cite{MR0121713}.  Donoho~\cite{Donoho2006} conducts
  an asymptotic investigation of these formulas.  To distill the
  essence of the analysis, Donoho gives continuous functions
  \(\Psi_{\mathrm{int}}(\theta, \tau)\) and
  \(\Psi_{\mathrm{ext}}(\theta)\) such that, for any \(\varepsilon
  >0\) and all sufficiently large \(d\), the inequalities
  \begin{align}\label{eq:psi-int}
    \frac{1}{d}\log\bigl(\beta(T_{k-1},T_{i+1})\bigr) &\le
    -\Psi_{\mathrm{int}}\left(\tfrac{i}{d},\tau\right) +\frac{\eps}{3}  \\
    \frac{1}{d}\log\bigl(\gamma(F_{i+1},C)\bigr) &\le
    -\Psi_{\mathrm{ext}}\left(\tfrac{i}{d}\right) + \frac{\eps}{3}\label{eq:psi-ext}
  \end{align}
  hold uniformly over \(i = \{\lceil \tau d \rceil,\dotsc,d-1\}\).  
  Moreover, it follows from Equation~\eqref{eq:binom-entropy-limit} that
  for sufficiently large \(d\), we have
  \begin{equation*}
    \frac{1}{d} \log\left(2^{i-k+2}\binom{d-k}{i-k+2}\right) 
    \le \Psi_{\mathrm{cont}}\left(\tfrac{i}{d},\tau\right) + \frac{\varepsilon}{3},
 \end{equation*}
 where the exponent \(\Psi_{\mathrm{cont}}\) for the number of
 \emph{cont}aining faces is defined by
 \begin{equation}\label{eq:psi-com}
  \Psi_{\mathrm{cont}}(\theta,\tau) 
  \defeq (\theta-\tau)\log(2) + (1-\theta)H\left(\frac{\theta-\tau}{1-\tau}\right).
 \end{equation}
 The function \(H(\theta)\) is the entropy defined
 by~\eqref{eq:entropy}.  Equation~\eqref{eq:l1-exponent-psis} follows
 by defining
 \begin{equation}\label{eq:psi-total}
   \Psi_{\mathrm{total}}(\theta,\tau) \defeq \Psi_{\mathrm{cont}}(\theta,\tau) -
   \Psi_{\mathrm{int}}(\theta,\tau) - \Psi_{\mathrm{ext}}(\theta)
 \end{equation}
 and taking logarithms in~\eqref{eq:int-vol-explicit}.  This is the
 first claim.
 
 We now show that, for any \(\theta<\tau\),
 Equation~\eqref{eq:l1-exponent-zeros} holds for all sufficiently
 large \(d\).  Since \(\vct x/\lone{\vct x}\) lies in a
 \((k-1)\)-dimensional face of the crosspolytope \(C\),
 every face of \( \Fcone(\lone{\cdot},\vct x)\) has dimension at least
 \((k-1)\).  It follows immediately from
 Definition~\ref{def:sphere-intrinsic-vols} that
 \begin{equation*}
   v_i\bigl(\Fcone(\lone{\cdot},\vct x)\bigr) = 0
 \end{equation*}
 for all \(i <k-1\).  Since \(k = \lceil \tau d\rceil\), we see
 that~\eqref{eq:l1-exponent-zeros} holds for all sufficiently
 large \(d\) so long as \(\theta < \tau\). This is the second claim.
\end{proof}

\subsection{Computing the exponents}
\label{sec:computing-psis}

\begin{figure}[t!]
  \centering
  \includegraphics[width=0.6\columnwidth]{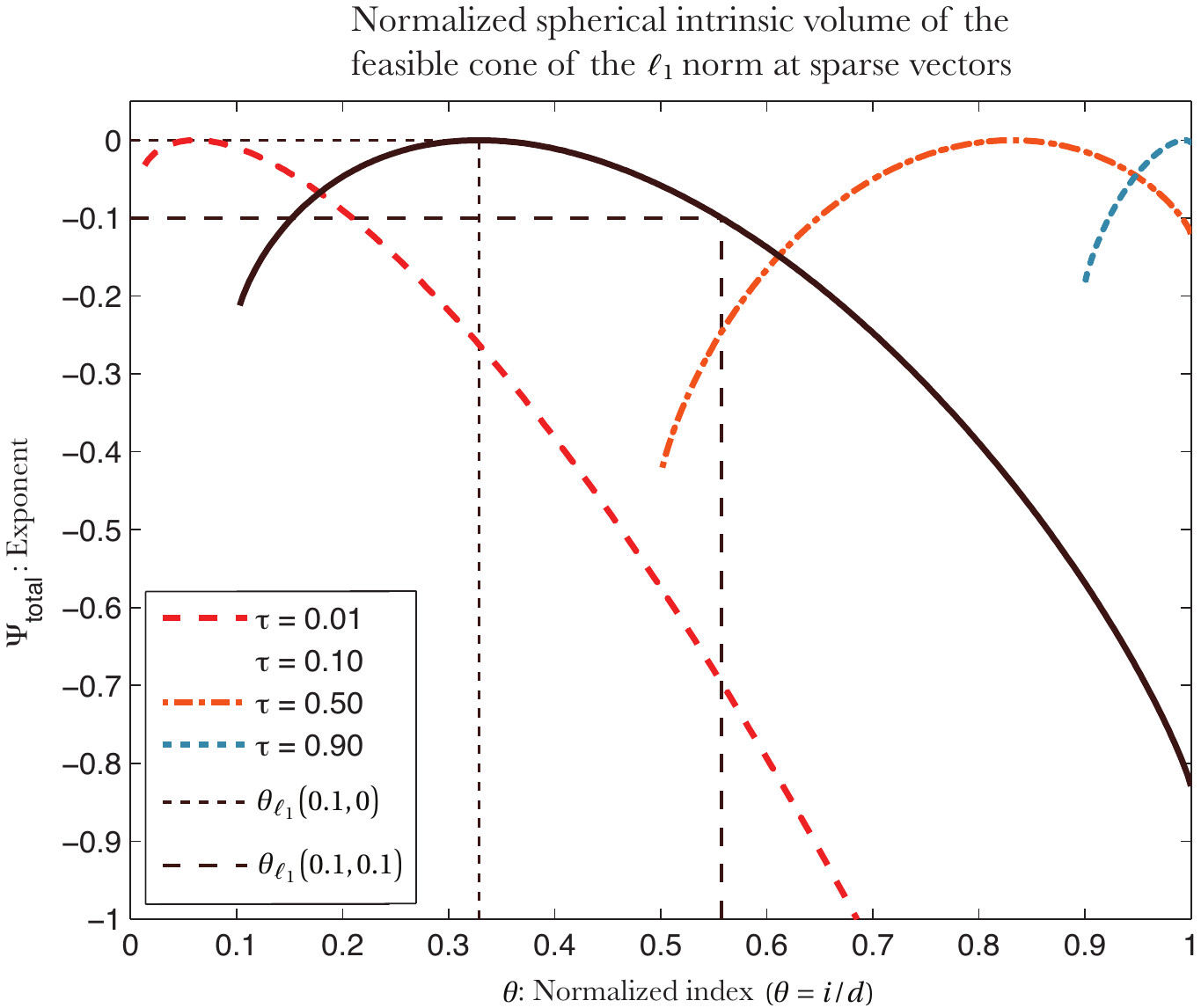}
  \caption{\textsl{Upper bounds for the exponent of the spherical
      intrinsic volumes.}  We plot \(\Psi_\mathrm{total}(\tau,\cdot)
    \) for several different values of \(\tau\).  The best upper decay
    threshold at level \(\psi\) is given by the rightmost \(\theta\)
    for which the curve intersects the horizontal line at \(-\psi\).
    The short dashes show the upper decay threshold
    \(\theta_{\ell_1}(0.1,0)\), while the long dashes show
    \(\theta_{\ell_1}(0.1,0.1)\), defined in
    Equation~\eqref{eq:theta-l1}. For each \(\tau\), the upper decay
    threshold at level zero is numerically equal to the lower decay
    threshold.}
  \label{fig:l1-psi-net}
\end{figure}

We now define the functions needed to compute \(\Psi_{\mathrm{total}}\)
in~\eqref{eq:psi-total}.  Recall that
\(\Psi_{\mathrm{cont}}\) is explicitly defined in~\eqref{eq:psi-com}.
The functions \(\Psi_{\mathrm{int}}\) and \(\Psi_{\mathrm{ext}}\) are
defined in~\cite{Donoho2006}, but we recapitulate their definitions
for completeness.  Define implicit parameters \(x = x(\theta)\) and
\(s = s(\theta,\tau)\) as  the solutions to the equations
\begin{align}\label{eq:implicit-x}
  \frac{2 x G(x)}{G'(x)} &= \frac{1-\theta}{\theta},\\
  M(s) &= 1- \frac{\tau}{\theta},\label{eq:implicit-s}
\end{align}
where \(G(x) = \frac{2}{\sqrt{\pi}} \int_0^x \exp(-t^2) \mathrm{d} t\)
is the error function \(\operatorname{erf}(x)\) and
\(M(s)\) is a variant of the Mills' ratio given by
\begin{equation*}
  M(s) = -s \mathrm{e}^{s^2/2} \int_{-\infty}^s \mathrm{e}^{-t^2/2}
  \mathrm{d}t = -s\sqrt{\frac{\pi}{2}}\; \operatorname{erfcx}\!\left(-\frac{s}{\sqrt{2}}\right),
\end{equation*}
where \(\mathrm{erfcx}(s) = \mathrm{e}^{s^2} \mathrm{erfc}(s)\) is the
scaled complementary error function.  This second form for \(M(s)\) is
convenient for numerical computations. It follows
from~\cite{Donoho2006} that \(x\) and \(s\) are well defined.
Numerical evaluation of \(x(\theta)\) and \(s(\theta, \tau)\) is
straightforward using, for example, bisection methods.

With these parameters in hand, the exponent for the internal angle
is given by\footnotemark
\begin{equation}\label{eq:int-angle}
  \Psi_{\mathrm{int}}(\theta, \tau) \defeq (\theta-\tau)\log\left(\sqrt{2\pi}
    \frac{s \theta}{\tau - \theta}\right) - \frac{\tau s^2}{2},
\end{equation}
where \(s=s(\theta,\tau)\) satisfies equation~\eqref{eq:implicit-s}.
The exponent for the external angle is
\begin{equation}
  \label{eq:ext-angle}
  \Psi_{\mathrm{ext}}(\theta) \defeq -(1-\theta) \log\bigl(G(x)\bigr) + \theta x^2,
\end{equation}
where \(x = x(\theta)\) is given by~\eqref{eq:implicit-x} above.  

Figure~\ref{fig:l1-psi-net} displays
\(\Psi_{\mathrm{total}}(\cdot,\tau)\) for a few values of the parameter \(\tau\).
Empirically, it appears that \(\Psi_{\mathrm{total}}(\cdot,\tau)\) is concave for
every value of \(\tau\in [0,1]\) and has a unique maximal value of zero.

\footnotetext{Equation~\eqref{eq:int-angle} requires a significant
  amount of wholly uninteresting algebraic simplification from the
  formulas of~\cite{Donoho2006}. The key steps
  in this simplification follow from the equations on page 638 of the
  reference.  In particular, we write \(y\) explicitly in terms of
  \(s\) with their Eq.~(6.12), and then write \(\xi\) explicitly in
  terms of \(s\) using this expression for \(y\)---see equation (6.13)
  in the reference.  Noting that the reference defines \(\gamma=
  \frac{\tau}{\theta}\) on page 631 gives~\eqref{eq:int-angle}, modulo
  trivial simplifications.}

\subsection{Defining decay thresholds}
\label{sec:comp-asympt-expon}

The exponent \(\Psi_{\mathrm{total}}\) provides decay thresholds for
the \(\ell_1\) norm at proportionally sparse vectors.  We define
\begin{equation}\label{eq:theta-l1}
  \theta_{\ell_1}(\tau,\psi) \defeq \inf\bigl\{\theta_{\star}\in [0,1]\mid
  \Psi_{\mathrm{total}}(\theta, \tau) < -\psi \text{ for all } \theta
  \in ( \theta_\star,1]\bigr\}.
\end{equation}
In words, \(\theta_{\ell_1}(\tau, \psi)\) is the rightmost point of 
intersection of the curve \(\Psi_{\mathrm{total}}(\cdot, \tau)\) with
the horizontal line at the level~\(-\psi\) (see
Figure~\ref{fig:l1-psi-net}).  Further define
\begin{equation}
  \label{eq:kappa-l1}
  \kappa_{\ell_1}(\tau) \defeq \sup\bigl\{\kappa_\star \in [0,1]\mid
  \Psi_{\mathrm{total}} (\theta,\tau) < 0 
  \text{ for all } \kappa \in [0,\kappa_\star)\bigr\}.
\end{equation}
This function \(\kappa_{\ell_1}(\tau)\) is the leftmost point of
intersection of \(\Psi_{\mathrm{total}}(\cdot,\tau)\) with the
horizontal line at level zero.  Equations~\eqref{eq:theta-l1}
and~\eqref{eq:kappa-l1} define decay thresholds for the ensemble of
feasible cones for the \(\ell_1\) norm at proportionally sparse
vectors.
\begin{proposition}[Decay thresholds for the \(\ell_1\) norm at
  proportionally sparse vectors]\label{prop:l1-decay}
  Consider the ensemble of
  Lemma~\ref{thm:l1-intrinsic-upper-bound}. The function
  \(\theta_{\ell_1}(\tau,\psi)\) is an upper decay threshold at level
  \(\psi\) for the ensemble \(\{\Fcone(\lone{\cdot},\vct
  x_0^{(d)})\mid d \in \Index\}\), while \(\kappa_{\ell_1}(\tau) \) is a
  lower decay threshold for the ensemble \(\{\Fcone(\lone{\cdot},\vct
  x_0^{(d)})\mid d \in \Index\}\).
\end{proposition}
\begin{proof}
  By definition, \(\Psi_{\mathrm{total}}(\theta,\tau)<-\psi\) for
  every \(\theta > \theta_{\ell_1}(\tau,\psi)\).  It then follows
  immediately from Lemma~\ref{thm:l1-intrinsic-upper-bound} and
  Definition~\ref{def:decay-at-psi} that \(\theta_{\ell_1}(\tau,\psi)\) is an
  upper decay threshold at level \(\psi\) for the ensemble
  \(\{\Fcone(\lone{\cdot},\vct x^{(d)})\mid d \in \Index\}\).
  The proof that \(\kappa_{\ell_1}(\tau)\) is a lower decay threshold
  for \(\bigl\{\Fcone(\lone{\cdot},\vct x^{(d)})\mid d \in \Index\bigr\}\) is
  equally straightforward, so we omit the argument.
\end{proof}
We abbreviate the upper decay threshold at level zero by
\begin{equation}
  \label{eq:upper-decay-level-zero}
  \theta_{\ell_1}(\tau) \defeq \theta_{\ell_1}(\tau,0)
\end{equation}
for consistency with Definition~\ref{def:aeub}.

Figure~\ref{fig:l1-psi-net} illustrates the definition of
\(\theta_{\ell_1}(\tau,\psi)\).  Numerically, it appears that zero is
the \emph{unique} maximal value of
\(\Psi_{\mathrm{total}}(\tau,\cdot)\) for every value \(\tau \in
(0,1]\).   If this is indeed the case, then we would be able to deduce
that \(\theta_{\ell_1}(\tau) = \kappa_{\ell_1}(\tau)\) for all values
of sparsity \(\tau\).   

\subsection{Reconciliation with~\protect{\cite{Donoho2006}}}
\label{sec:reconciliation}
We now discuss the relationship between our spherical intrinsic volume
approach and the bounds of~\cite{Donoho2006}, \cite{Donoho2005a} for basis
pursuit.  Numerically, it appears that the two approaches provide
equivalent success guarantees, but the expressions for the exponents
seem to preclude a direct proof of equivalence.  We also describe how
our approach gives matching upper bounds for region of success of
basis pursuit, which shows that our results are the best possible up
to numerical accuracy.

\subsubsection{Reconciliation with the weak threshold}
\label{sec:weak-threshold}
Recall that basis pursuit is the linear inverse
problem~(\ref{eq:generic-lin-inv}) with objective \(f(\cdot)
=\lone{\cdot}\).  By the first part of Lemma
\ref{lem:lin-inv-and-decay}, basis pursuit with a Gaussian measurement
matrix \(\mtx \Omega \in \R^{\lceil \sigma d\rceil \times d}\) succeeds at
recovering a \(k=\lceil \tau d\rceil\)-sparse vector with overwhelming
probability in high dimensions, so long as the pair \((\tau,\sigma)\) satisfies
\(\theta_{\ell_1}(\tau) < \sigma\).

We now describe the analogous result given
in~\cite[Sec.~7.1]{Donoho2006}.  Define the critical sparsity ratio
(compare with the critical proportion~\cite[Def.~2]{Donoho2006})
\begin{equation}\label{eq:tauW}
  \tau_{W}(\sigma) = \sup \bigl\{ \hat \tau
  \in [0,\sigma]\mid 
  \Psi_{\mathrm{total}}(\theta,\hat\tau)<0 \text{ for all } \theta \in [\sigma,1]\bigr\}.
\end{equation}
Then the result~\cite[Thm.~2]{Donoho2006} is equivalent to the
statement that basis pursuit with a Gaussian matrix \(\mtx \Omega \in
\R^{\lceil\sigma d \rceil\times d}\) succeeds with overwhelming
probability in high dimensions  whenever the pair \((\tau,\sigma)\)
satisfies \(\tau < \tau_W(\sigma)\).

These two approaches show strong similarities, and methods provide the
same results to numerical precision.  Indeed, under the assumption
that both \(\tau_W(\sigma)\) and \(\theta_{\ell_1}(\tau)\) are
monotonically increasing functions (this appears to hold empirically),
one can show that these approaches are equivalent.  Rather than dwell
on this fine detail, we present a matching failure region for basis
pursuit.

\subsubsection{Matching upper bound}
\label{sec:matching-upper-bound}
The following result shows links the lower decay threshold to
regions where basis pursuit fails.
\begin{proposition}\label{prop:match-upper-bp}
  Suppose \(\kappa_{\ell_1}(\tau) > \sigma\).  Then basis pursuit with
  \(n=\lceil \sigma d\rceil \) Gaussian measurements fails with
  overwhelming probability in high dimensions for the \(\tau\)-sparse
  ensemble of Lemma \ref{thm:l1-intrinsic-upper-bound}.
\end{proposition}
Since the function \(\Psi_{\mathrm{total}}(\cdot,\tau)\) has a unique
maximal value of zero up to our ability to compute the functions
involved (see Figure~\ref{fig:l1-psi-net}), we have the equality
\(\kappa_{\ell_1}(\tau) = \theta_{\ell_1}(\tau)\) to numerical
precision.  Coupling Proposition~\ref{prop:match-upper-bp} with our
discussion in Section~\ref{sec:weak-threshold} reveals that basis
pursuit with a Gaussian measurement matrix exhibits a phase transition
between success and failure at \(\sigma =\theta_{\ell_1}(\tau)\).
\begin{proof}[Proof of Proposition~\ref{prop:match-upper-bp}]
  Let \(\{\vct x_0^{(d)} \in \R^d \}\) be an ensemble of
  \(\tau\)-sparse vectors as in
  Lemma~\ref{thm:l1-intrinsic-upper-bound}.  The null space of an
  \(n\times d\) Gaussian matrix is distributed as \(\mtx Q L\), where
  \(L\) is a linear subspace of dimension \((d-n)\).  It then follows
  from~\cite[Prop.~2.1]{ChaRecPar:12} that basis pursuit with
  \(n=\lceil \sigma d\rceil \) Gaussian measurements succeeds with the
  same probability that \(\mtx Q L \cap \Fcone(\lone{\cdot},\vct x_0)
  = \{\vct 0\}\).

  By Proposition~\ref{prop:subspace-aeub}, the value \(\kappa_\star =
  (1-\sigma)\) is a lower decay threshold for \(L\).  Our assumption
  implies \(\kappa_{\ell_1}(\tau) + \kappa_\star > 1\), so by
  Theorem~\ref{thm:asympt-lower-bounds}, we see \(\mtx Q L \cap
  \Fcone(\lone{\cdot},\vct x_0)\ne \{\zerovct\}\) with overwhelming
  probability in high dimensions.  We conclude that basis pursuit fails
  with overwhelming probability in high dimensions.
\end{proof}

\section{Proof of Lemma~\ref{lem:lin-inv-and-decay} and Corollary~\ref{cor:gauss-width-bd}}
\label{sec:proof-thoer-coroll}

We begin with the proof of the two claims of Lemma
\ref{lem:lin-inv-and-decay} concerning the relationship between the
upper decay threshold and the linear inverse
problem~\eqref{eq:generic-lin-inv}.  The first result is a corollary
of~\cite[Prop.~2.1]{ChaRecPar:12}.  We drop the superscript
\(d\) for clarity.

\begin{proof}[Proof of Lemma~\ref{lem:lin-inv-and-decay},
  Part~\ref{item: thresh-implies-success}]
  Let \(\mtx \Omega \) be the \(n\times d\) Gaussian measurement
  matrix, where \(n=\lceil \sigma d\rceil\).  The null space of \(\mtx
  \Omega\) is distributed as \(\mtx Q L\), where \(\mtx Q\) is a
  random basis, and \(L\) is any fixed \((d-n)\)-dimensional subspace of
  \(\R^d\).  Therefore, the probability
  that~\eqref{eq:generic-lin-inv} succeeds is equal to the probability
  that \( \mtx Q L \cap \Fcone(f,\vct x_0) \ne \{\vct
  0\}\)~\cite[Prop.~2.1]{ChaRecPar:12}. By
  Proposition~\ref{prop:subspace-aeub}, the subspace \(L\) has an
  upper decay threshold \(\theta_L = 1-\sigma\), so  that the assumption 
  \(\theta_\star < \sigma\) implies that \(\theta_\star + \theta_L < 1\).
  The claim follows from Theorem~\ref{thm:asympt-thresh-gen}.
\end{proof}

The second claim of Lemma~\ref{lem:lin-inv-and-decay}
requires additional effort.  We require the following
technical lemma.  
\begin{lemma}\label{lemma:aeub-from-lin-inv}
  Let \(\Index\) be an infinite set of indices.  Let \(\{K^{(d)}\mid d \in \Index\}\) be an ensemble of closed convex
  cones with \(K^{(d)}\subset \R^d\) for each \(d\), and let
  \(\{L^{(d)}\mid d \in \Index\}\) be an ensemble of linear subspaces
  of ~\(\R^d\) of dimension \(d-\lceil \sigma d\rceil\).  If there
  exists an \(\varepsilon >0\) such that for every sufficiently large
  \(d\),
  \begin{equation*}
    \mathbb{P}\bigl\{ K^{(d)} \cap \mtx Q L^{(d)} \ne
      \{\mathbf{0}\}\bigr\} \le \mathrm{e}^{-\varepsilon d},
  \end{equation*}
  then \(\{K^{(d)}\}\) has an upper decay threshold
  \(\theta_\star = \sigma\).
\end{lemma}
Again, the spherical kinematic formula~\eqref{eq:sphere-kin} is at the
heart of the proof.  
\begin{proof}[Proof of Lemma~\protect{\ref{lemma:aeub-from-lin-inv}}]
  We split the argument into two cases: first, we consider the case
  where \(K^{(d)}\) is not a subspace, and then consider the case
  where \(K^{(d)}\) is a subspace.  The general mixed-cone case
  follows by applying these arguments to the subsequences consisting
  of only one type of cone.

  We drop the superscript \(d\) for clarity.  Let \(n = \lceil \sigma
  d\rceil\).  We first assume that \(K\) is not a subspace. By the spherical
  kinematic formula, the probability of interest \(P\) is given by
   \begin{equation*}
   P \defeq \mathbb{P}\bigl\{K\cap \mtx Q L \ne
       \{\mathbf{0}\}\bigr\} = \sum_{k=0}^{d-1}(1+(-1)^k)
     \sum_{i=k}^{d-1}v_i(L) v_{d-1-i + k}(K ).
   \end{equation*}
   By Proposition~\ref{prop:subspace-iv}, \(v_i(L) =
   \delta_{i,d-n-1}\), so by replacing \(k\) with \(k-n\), the probability above reduces to
   \begin{equation}\label{eq:one-subspace}
     P = \sum_{k=n}^{d-1} (1+(-1)^{k-n}) v_k(K).
   \end{equation}
   By assumption, we have \(P\le \mathrm{e}^{-\varepsilon d}\) for all
   sufficiently large \(d\), so the positivity of spherical
   intrinsic volumes
   (Fact~\ref{fact:obvious}.\ref{item:obv-positivity}) implies
   \begin{equation*}
     v_k(K) \le \mathrm{e}^{-\varepsilon d},   \text{ for any \(k\ge n\) such that \( k
     \equiv n \mod 2\). }
  \end{equation*}
 
  It requires an additional geometric observation to remove the
  dependence on parity. Let \(\tilde L\) be a \((d-n-1)\)-dimensional
  subspace contained in \(L\).  By containment, it is immediate that
  \begin{equation*}
    \tilde P \defeq \mathbb{P}\left\{K\cap \mtx Q \tilde L \ne
       \{\mathbf{0}\}\right\}   \le  \mathbb{P}\left\{K\cap \mtx Q L \ne
       \{\mathbf{0}\}\right\}  \le \mathrm{e}^{-\varepsilon d},
  \end{equation*}
  where the last inequality is by assumption.  But the same manipulations
  as before show
  \begin{equation*}
    \tilde P = \sum_{k=n+1}^{d-1}(1+(-1)^{k-n-1}) v_k(K)
  \end{equation*}
  so we have \(v_k(K)\le \mathrm{e}^{-\varepsilon d} \)
  for every \(k \ge n+1\) such that \(k \equiv n+1 \mod 2\).   
  In summary, for every \(d\)
  sufficiently large and any \(k \ge n =\lceil \sigma d\rceil\),  we
  have \(v_k(K) \le \mathrm{e}^{-\varepsilon d}\).   By definition,
  \(\sigma\) is an upper decay threshold for the ensemble of cones
  \(K\).  This completes the first case.

  Now suppose that \(K\) is a subspace and define \(m \defeq
  \dim(K)\).  Since \(\dim(L) = d- \lceil \sigma d\rceil\), we have
  \begin{equation*}
    P\defeq    \Prob{K\cap \mtx Q L \ne \{\zerovct \}} =
    \left\{\begin{array}{ll}0, &
        m<  \lceil \sigma d\rceil  \\ 1, & \text{otherwise}\end{array}\right.
  \end{equation*}
because randomly oriented subspaces are almost always in general position.
  The assumption that \(P \le \econst^{-\eps d}\) for all sufficiently
  large \(d\) requires that \(m< \lceil \sigma d \rceil\) for all
  sufficiently large \(d\).  By Proposition~\ref{prop:subspace-aeub},
  the scalar \(\sigma\) is an upper decay threshold for subspace
  \(K\).  This is the result for the second case, so we are done.
\end{proof}

\begin{proof}[Proof of Lemma~\ref{lem:lin-inv-and-decay},
  Part~\ref{item:success-implies-thresh}]
  The results of~\cite{ChaRecPar:12} imply that the linear
  inverse problem~\eqref{eq:generic-lin-inv} with a Gaussian
  measurement matrix \(\mtx \Omega\) succeeds with the same probability
  that a randomly oriented \((n-d)\)-dimensional subspace \(\mtx Q L\)
  strikes the feasible cone \(\Fcone(f,\vct x_0)\) trivially.  The
  result then follows from Lemma~\ref{lemma:aeub-from-lin-inv}.
\end{proof}

We conclude with the proof of Corollary~\ref{cor:gauss-width-bd}, which asserts that  \(\theta_\star\) is an upper decay threshold for the ensemble \(\{\Fcone(f^{(d)},\vct x_0^{(d)})\mid d\in \mathcal{D}\}\)  whenever the bound~\eqref{eq:limsup-width}  holds.
\begin{proof}[Proof of Corollary~\ref{cor:gauss-width-bd}]
  To shorten notation, we drop the explicit dependence on \(d\), and we define the width \(W:= W\bigl(\Fcone(f,\vct x_0)\cap \mathsf{S}^{d-1}\bigr)\).   The result\footnote{This result is a corollary of a Gaussian process inequality due to    Gordon~\cite{Gordon1987}, \cite{Gordon1988}. }%
  ~\cite[Cor.~3.3(1)]{ChaRecPar:12} states
  that for any \(\eps>0\),
  \begin{equation}\label{eq:CRPW-bd}
    n  \ge \bigl(W+\eps\sqrt{d}\bigr)^2 +1 \ \implies\text{\eqref{eq:lin-inv-by-dim} fails with probability}\le    \econst^{-\eps^2d/2}.
  \end{equation}
  
  Fix \(\eps \in (0,1)\) and define \(\theta_\eps \defeq \theta_\star + 2\eps(\theta_\star^{1/2}+1)\).  We claim that \(\theta_\eps\) is an upper decay threshold for the ensemble \(\{\Fcone(f,\vct x_0)\}\).  To see this, choose the number of measurements \(n= \lceil \theta_\eps d\rceil\) in the linear inverse problem~\eqref{eq:lin-inv-by-dim}. Then for \(d\ge \eps^{-1}\), we have
  \begin{equation*}
    n \ge \theta_\eps d  \ge \bigl[\theta_\star +\eps(2\theta_\star^{1/2}+1)\bigr] d + 1
  \end{equation*}
  by our choice of \(n\).
  We  bound the bracketed expression  using convexity of the map \(\eps \mapsto \left(\theta_\star^{1/2} + \eps\right)^2\):
  \begin{equation*}
    \theta_\star +\eps(2\theta_\star^{1/2}+1) = (1-\eps)\theta_\star + \eps\left(\theta_\star^{1/2}+1\right)^2 \ge \left(\theta_\star^{1/2} + \eps\right)^2
  \end{equation*}
  because \(\eps \in (0,1)\) and a convex function lies below the chord connecting its endpoints. Combining the two displayed equations above,
  \begin{equation*}
    n -1 \ge \left(\theta_\star^{1/2} + \eps\right)^2 d \ge \bigl(W + \eps\sqrt{d}\bigr)^2
  \end{equation*}
  where the second inequality holds for all sufficiently large \(d\) by assumption~\eqref{eq:limsup-width}.  The implication~\eqref{eq:CRPW-bd} thus implies that the linear inverse problem~\eqref{eq:lin-inv-by-dim} succeeds with overwhelming probability in high  dimensions when \(n = \lceil \theta_\eps d\rceil\).  From the second part of Lemma~\ref{lem:lin-inv-and-decay}, we conclude that \(\theta_\eps\) is an upper decay threshold for \(\{\Fcone(f,\vct x_0)\}\), as claimed.    The proof is completed by taking \(\eps \to  0\) and verifying that a limit of decay thresholds is itself a decay threshold.  We omit this straightforward, but technical, argument.
\end{proof}

\small

\let\oldbibliography\thebibliography
\renewcommand{\thebibliography}[1]{%
  \oldbibliography{#1}%
  \setlength{\itemsep}{0pt}%
}
\bibliographystyle{abbrv}
\bibliography{final}

\end{document}